\documentclass[runningheads]{llncs}
\usepackage{graphicx}
\usepackage{todonotes}
\usepackage{xspace}
\usepackage{amsmath}
\usepackage{color}
\usepackage{amssymb}
\usepackage{paralist}
\usepackage[colorlinks=true]{hyperref}
\ifpdf
\hypersetup{
  pdftitle={Schematic Representation of Large Biconnected Graphs},
  pdfauthor={G. Di Battista, F. Frati, M. Patrignani, M. Tais}
}
\fi
\hypersetup{colorlinks,linkcolor={blue},citecolor={blue},urlcolor={blue}}

\usepackage[caption=false]{subfig} 

\usepackage{lineno}

\usepackage{fullpage}

\let\doendproof\endproof
\renewcommand\endproof{~\hfill\qed\doendproof}

\newcommand{\isnested}[0]{\Supset}
\newcommand{\wraps}[0]{\Subset}
\newcommand{\remove}[1]{}

\begin{document}
%
%
%
\title{Schematic Representation of \\Large Biconnected Graphs\thanks{This research was supported in part by MIUR Project ``AHeAD'' under PRIN 20174LF3T8, by H2020-MSCA-RISE Proj.\ ``CONNECT'' n$^\circ$ 734922, and by Roma Tre University Azione 4 Project ``GeoView''.}}

\author{Giuseppe Di Battista \and Fabrizio Frati \and Maurizio Patrignani \and Marco Tais}
\authorrunning{Giuseppe Di Battista et al.}
\titlerunning{Schematic Representation of Large Biconnected Graphs}
%
\institute{Roma Tre University, Rome, Italy\\ 
\email{\{gdb,frati,patrigna,tais\}@dia.uniroma3.it}} 

\maketitle

\begin{abstract}
Suppose that a biconnected graph is given, consisting of a large component plus several other smaller components, each separated from the main component by a separation pair.
We investigate the existence and the computation time of schematic representations of the structure of such a graph where the main component is drawn as a disk, the vertices that take part in separation pairs are points on the boundary of the disk, and the small components are placed outside the disk and are represented as non-intersecting lunes connecting their separation~pairs.
We consider several drawing conventions for such schematic representations, according to different ways to account for the size of the small components.
We map the problem of testing for the existence of such representations to the one of testing for the existence of suitably constrained $1$-page book-embeddings and propose several polynomial-time algorithms.
\end{abstract}

\section{Introduction}\label{se:intro}

Many of today's applications are based on large-scale networks, having billions of vertices and edges. This spurred an intense research activity devoted to finding methods for the visualization of very large graphs. 

Several recent contributions focus on algorithms that produce drawings where either the graph is only partially represented or it is schematically visualized. Examples of the first type are proxy drawings \cite{DBLP:conf/apvis/NguyenMLE18,DBLP:journals/tvcg/NguyenHEM17}, where a graph that is too large to be fully visualized is represented by the drawing of a much smaller proxy graph that preserves the main features of the original graph. Examples of the second type are graph thumbnails \cite{DBLP:journals/tvcg/YoghourdjianDKM18}, where each connected component of a graph is represented by a disk and biconnected components are represented by disks contained into the disk of the connected component they belong to. 

Among the characteristics that are emphasized by the above mentioned drawings, a crucial role is played by connectivity. Following this line of thought, we study schematic representations of graphs that emphasize their connectivity features. We start from the following observation: quite often, real-life very large graphs have one large connected component and several much smaller other components (see, e.g., \cite{fzb-fbna-16,mn-n-18}). This happens to biconnected and 
triconnected components too (see, e.g., \cite{ceccarelli} for an analysis of the graphs in \cite{snapnets}). 

Hence, we concentrate on a single biconnected graph (that can be a biconnected component of a larger graph) consisting of a large component plus several other smaller components, each separated from the large component by a separation pair. We propose to represent the large component as a disk, to draw the vertices of such a component that take part in separation pairs as points on the boundary of the disk, and to represent the small components as non-intersecting lunes connecting their separation pairs placed outside the disk.
See Fig.~\ref{fig:opening-figure}.
This representation is designed to emphasize the arrangement of the components with respect to the separation pairs.
 For simplicity, we assume that each separation pair separates just one small component from the large one.

\begin{figure}[tb]
	\centering
	\subfloat[]{\label{fig:max-constrained-embedding}\includegraphics[scale=0.13]{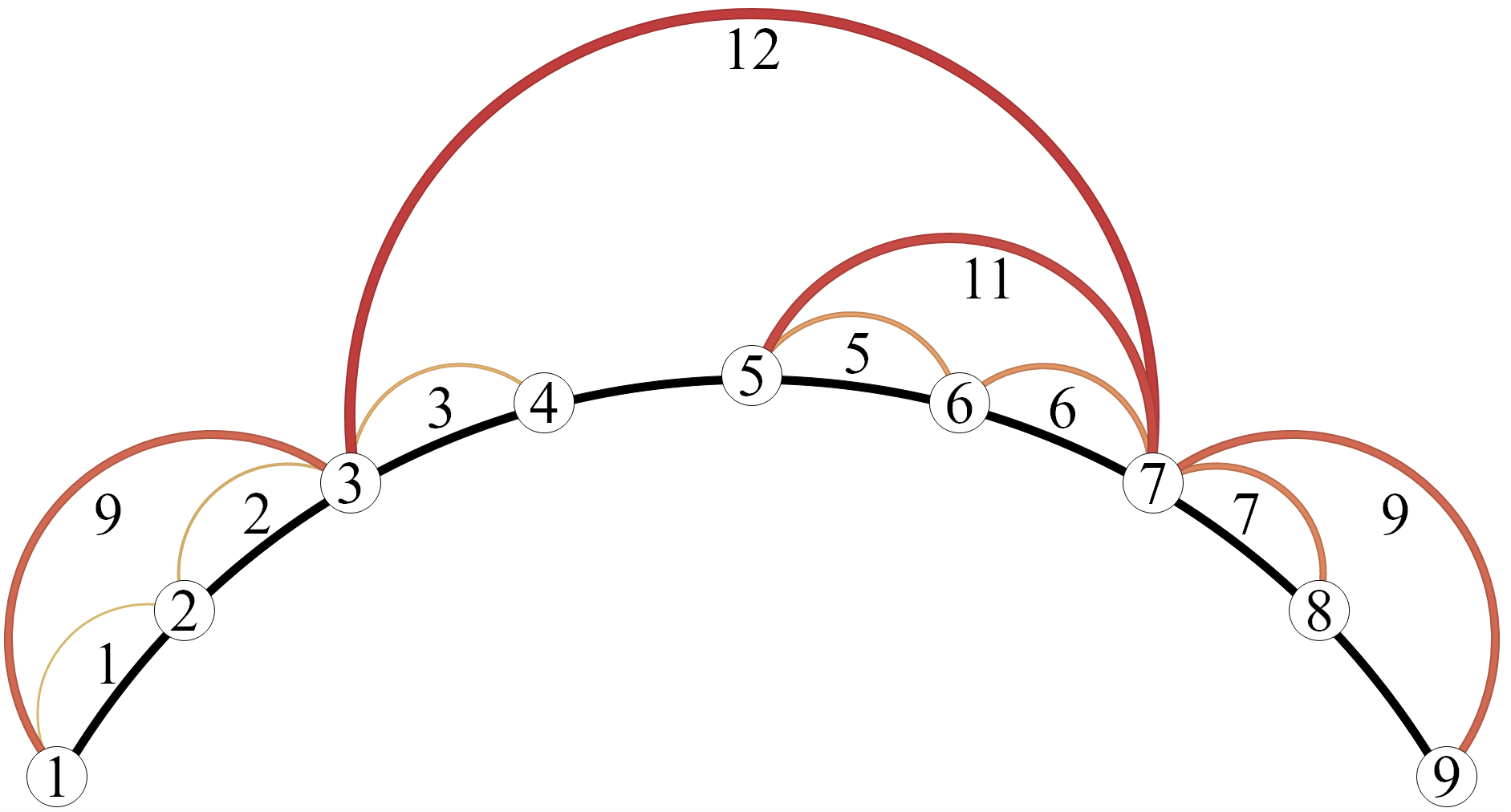}}
	\hfil
	\subfloat[]{\label{fig:2D-embedding}\includegraphics[scale=0.18]{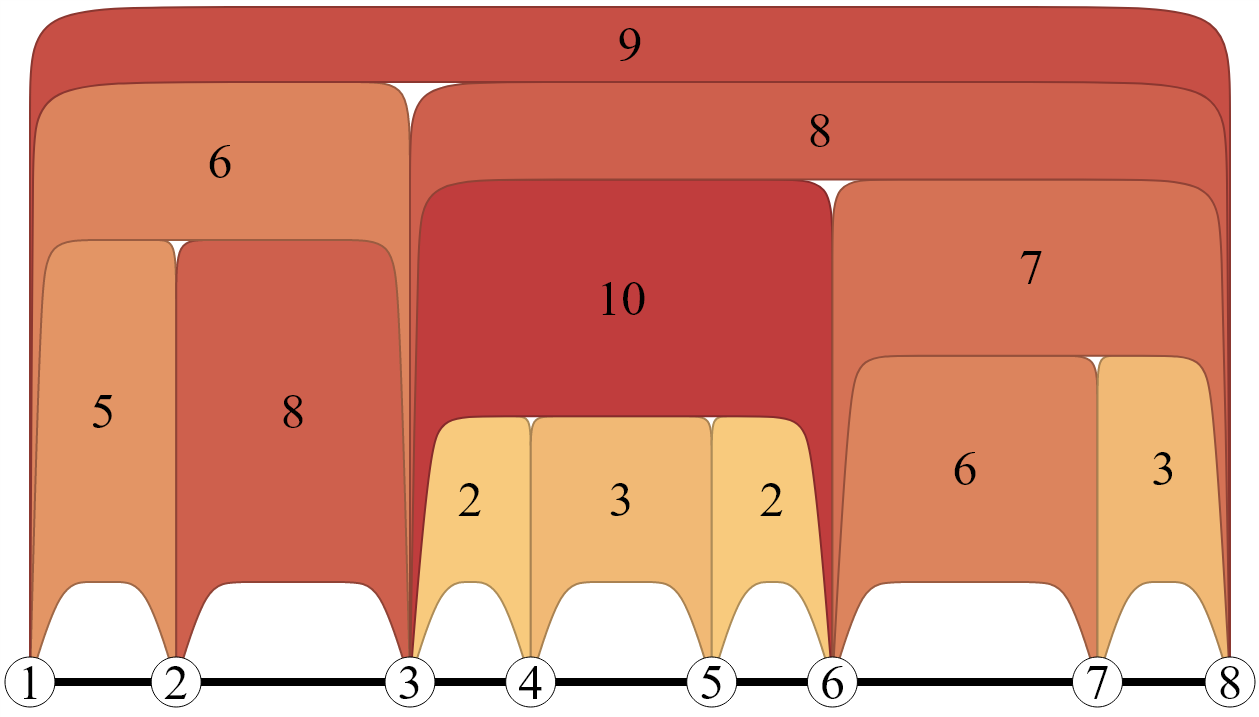}}
	\hfil
	\caption{Schematic representations of biconnected graphs. (a) A {\sc max}-constrained book-embedding. (b) A two-dimensional book-embedding; for simplicity the vertices are aligned on a straight-line.}\label{fig:opening-figure}
\end{figure}

More formally, our input is a weighted graph $G=(V,E,\omega)$, where each vertex in $V$ participates in at least one separation pair, each edge $(u,v)$ of $E$ represents a small component separated from the large one by the separation pair $\{u,v\}$, and $\omega$ assigns a positive weight to each edge. The weight of an edge represents a feature that should be emphasized in the schematic representation. As an example, it might represent the number of vertices or edges of the corresponding small component. 

We study one-dimensional and two-dimensional representations. In both cases, the vertices of $G$ form a sequence of linearly ordered points that are placed along the  boundary of a disk.
In the one-dimensional representations, we draw each edge as an arc and impose that arcs do not cross. Also, consider two edges $(u,v)$ and $(x,y)$ and suppose that the weight of $(u,v)$ is larger than that of $(x,y)$. Then we impose that  $(u,v)$ is drawn outside $(x,y)$, so to represent the weight by means of the edge length. We call \emph{{\sc max}-constrained book-embedding} this type of representation (see Fig.~\ref{fig:max-constrained-embedding}). In Section~\ref{se:max}, we present an optimal $O(n \log n)$-time algorithm that tests whether an $n$-vertex graph admits such a representation.
We also study a more constrained type of representations. Namely, let $(u,v)$ be an edge and consider the sequence of edges $(u_1,v_1),\dots,(u_k,v_k)$ that are drawn immediately below $(u,v)$; then we may want that $\omega(u,v)>\sum_{i=1}^k \omega(u_i,v_i)$. We call \emph{{\sc sum}-constrained book-embedding} this type of representation. In Section~\ref{se:sum}, we present an $O(n^3\log n)$-time algorithm that tests whether an $n$-vertex graph admits such a representation. Both {\sc max}- and {\sc sum}-constrained book-embeddings are $1$-page book-embeddings satisfying specific constraints. Hence, a necessary condition for $G$ to admit these types of representations is outerplanarity~\cite{BERNHART1979320}.

Since there exist weighted outerplanar graphs that  admit neither a {\sc max}- nor a {\sc sum}-constrained book-embedding, we study how to represent without crossings a weighted outerplanar graph with edges that have, in addition to their length, also a thickness: each edge is represented with a lune with area proportional to its weight. We call \emph{two-dimensional book-embeddings} these representations. See Fig.~\ref{fig:2D-embedding}. First, in Section~\ref{se:two-dimensional}, we show that all weighted outerplanar graphs admit a two-dimensional book-embedding and discuss the area requirements of such representations. Second, in Section~\ref{se:minres}, we show that, if a finite resolution rule is imposed, then there are graphs that do not admit any two-dimensional book-embedding and we present an $O(n^4)$-time algorithm that tests whether an $n$-vertex graph admits such a representation.

Conclusions and open problems are presented in Section~\ref{se:conclusions}.


%


\section{Preliminaries}\label{se:preliminaries}


We introduce some definitions and preliminaries.

{\bf Block-cut-vertex tree.} A {\em cut-vertex} in a connected graph $G$ is a vertex whose removal disconnects $G$. A graph with no cut-vertex is {\em biconnected}. A {\em block} of $G$ is a maximal (in terms of vertices and edges) subgraph of $G$ which is biconnected. The {\em block-cut-vertex tree} $T$ of $G$~\cite{h-gt-69,ht-aeagm-73} has a {\em B-node} for each block of $G$ and a {\em C-node} for each cut-vertex of $G$; a B-node $b$ and a C-node $c$ are adjacent if $c$ is a vertex of the block of $G$ represented by~$b$. We denote by $G(b)$ the block of $G$ represented by a B-node $b$. We often identify a C-node of $T$ and the corresponding cut-vertex of $G$.


{\bf Planar drawings.} A \emph{drawing} of a graph maps each vertex to a point in the plane and each edge to a Jordan arc between its end-vertices. A drawing is \emph{planar} if no two edges intersect, except at common end-vertices. A planar drawing partitions the plane into connected regions, called \emph{faces}. The bounded faces are \emph{internal}, while the unbounded face is the \emph{outer face}.

{\bf Outerplanar graphs.} An \emph{outerplanar drawing} is a planar drawing such that all the vertices are incident to the outer face. An \emph{outerplanar graph} is a graph that admits an outerplanar drawing. Two outerplanar drawings are \emph{equivalent} if the clockwise order of the edges incident to each vertex is the same in both drawings. An \emph{outerplane embedding} is an equivalence class of outerplanar drawings. A biconnected outerplanar graph has a unique outerplane embedding~\cite{mw-opbe-90,s-cog-79}. 
Given the outerplane embedding $\Gamma$ of an $n$-vertex biconnected outerplanar graph $G$, we define the \emph{extended dual tree} $\mathcal T$ of $\Gamma$ as follows (refer to Fig.~\ref{fig:outerplanardual}). We first construct the dual graph $\mathcal D$ of $\Gamma$; we then split the vertex of $\mathcal D$ corresponding to the outer face of $\Gamma$ into $n$ degree-$1$ vertices, each incident to an edge that is dual to an edge of $G$ incident to the outer face of $\Gamma$. Note that $\mathcal T$ can be constructed in $O(n)$ time. Further, each edge of $\mathcal T$ is dual to an edge of $G$; moreover, the edges incident to leaves of $\mathcal T$ are dual to edges incident to the outer face of $\Gamma$. 

\begin{figure}[tb]
	\centering
	\includegraphics[scale=0.7]{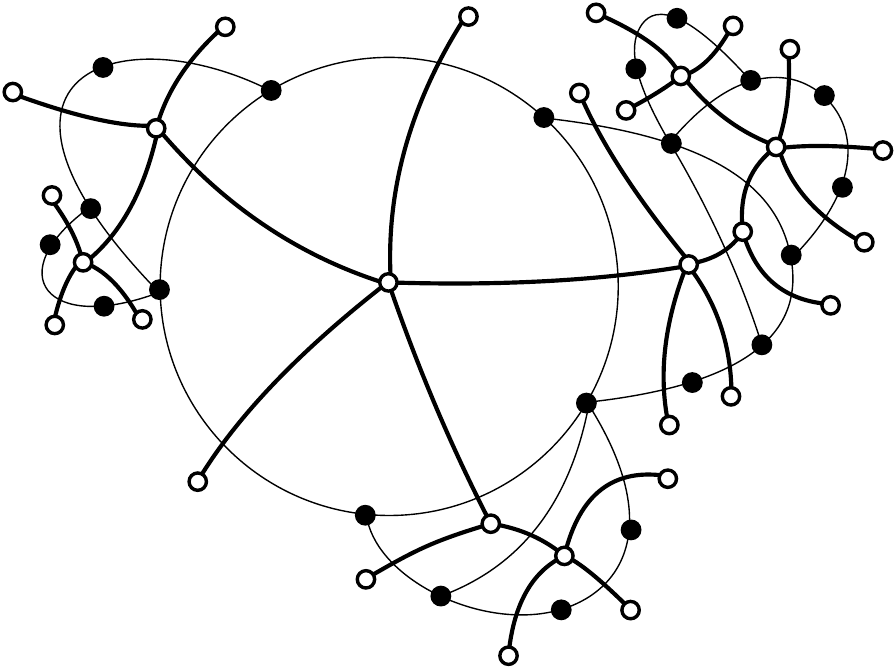}
	\caption{The extended dual tree $\mathcal T$ of an outerplane embedding $\Gamma$ of a $2$-connected outerplanar graph $G$; the vertices and the edges of $\mathcal T$ are represented by white disks and thick lines, respectively.}
	\label{fig:outerplanardual}
\end{figure}


{\bf Book-embeddings.} Given a graph $G$ and a linear order $\mathcal L$ of its vertices, we write $u \prec_{\mathcal L} v$ to represent the fact that $u$ precedes $v$ in $\mathcal L$; we say that two edges $(u,v)$ and $(w,z)$ of $G$ \emph{cross} if $u\prec_{\mathcal L} w \prec_{\mathcal L} v \prec_{\mathcal L} z$. A \emph{$1$-page book-embedding} of a graph is a linear order $\mathcal L$ of its vertices such that no two edges cross. 

The \emph{flip} of a $1$-page book-embedding $\mathcal L$ is a $1$-page book-embedding $\mathcal{L}'$ such that, for any pair of distinct vertices $u$ and $v$, we have that $u \prec_{\mathcal{L}'} v$ if and only if $v \prec_{\mathcal{L}} u$.

Given a linear order $\mathcal L$ of the vertices of a graph, by $u \preceq_{\mathcal L} v$ we mean that either $u \prec_{\mathcal L} v$ or $u=v$.
For a pair of distinct edges $e_1=(u_1,v_1)$ and $e_2=(u_2,v_2)$ of $G$ such that $u_1 \preceq_{\mathcal L} u_2 \prec_{\mathcal L} v_2 \preceq_{\mathcal L} v_1$, we say that $e_2$ is \emph{nested into} $e_1$ (denoted as $e_2 \isnested e_1$) and $e_1$ \emph{wraps around} $e_2$ (denoted as $e_1 \wraps e_2$).
Further, a subgraph $G'$ of $G$ \emph{lies under} (resp. \emph{lies strictly under}) an edge $(u,v)$ of $G$, where $u \prec_{\mathcal L} v$, if for every vertex $w$ of $G'$, we have $u\preceq_{\mathcal L} w\preceq_{\mathcal L} v$ (resp. $u\prec_{\mathcal L} w\prec_{\mathcal L} v$). Then a subgraph $G'$ of $G$  \emph{lies under} (resp. \emph{lies strictly under}) a subgraph $G''$ of $G$ if there exists an edge $(u,v)$ of $G''$ such that $G'$ lies under (resp. lies strictly under) $(u,v)$.

Consider a vertex $v$ in a book-embedding $\mathcal{L}$. The \emph{lowest-left edge} incident to $v$ is the edge $(u,v)$ such that: (i) $u\prec_{\mathcal{L}} v$ and (ii) no neighbor $w$ of $v$ is such that $u\prec_{\mathcal{L}} w \prec_{\mathcal{L}} v$; note that the lowest-left edge incident to $v$ is undefined if no neighbor of $v$ precedes $v$ in $\mathcal{L}$. The \emph{lowest-right edge} incident to $v$ is defined analogously.


In the rest of this paper, a \emph{weighted graph} $G=(V,E,\omega)$ is a graph equipped with a function $\omega$ that assigns a positive weight to each edge of $E$. 

\section{{\sc max}-Constrained Book-Embeddings}\label{se:max}

In this section, we study a first type of one-dimensional representations. We are given a weighted graph $G=(V,E,\omega)$. We draw the vertices in $V$ as a sequence of points linearly ordered on the boundary of a disk and the edges in $E$ as non-intersecting arcs positioned outside the disk, placing edges with larger weight outside edges of smaller weight.


More formally, a \emph{{\sc max}-constrained book-embedding} of a weighted outerplanar graph $G=(V,E,\omega)$ is a $1$-page book-embedding $\mathcal{L}$ such that, for any two distinct edges $e_1=(u,v)$ and $e_2=(x,y)$ in $E$ with $u \preceq_{\mathcal{L}} x \prec_{\mathcal{L}} y \preceq_{\mathcal{L}} v$, we have that $\omega(e_1)>\omega(e_2)$. That is, if $e_1$ wraps around $e_2$, then $\omega(e_1)>\omega(e_2)$. We do not specify the actual drawing of the edges since, if $G$ has a {\sc max}-constrained book-embedding, then they can be easily represented by non-crossing Jordan arcs. 
An example of {\sc max}-constrained book-embedding is in Fig. \ref{fig:max-constrained-embedding}. Observe, for instance, how the edges $(5,6)$ and $(6,7)$ that have weight $5$ and $6$, respectively, are below the edge $(5,7)$ that has weight $11$ and how such edge is below the edge $(3,7)$ whose weight is $12$.
We have the following preliminary observation.

\begin{property} \label{pro:max-weigth}
Let $G=(V,E,\omega)$ be a weighted outerplanar graph and let $e_M \in E$ be an edge such that $\omega(e_M) \geq \omega(e)$, for every $e\in E$.  In any {\sc max}-constrained book-embedding
of $G$, there exists no edge that wraps around~$e_M$.
\end{property}

The goal of this section is to prove the following theorem.

\begin{theorem}\label{th:linear-max-outerplanar}
Let $G=(V,E,\omega)$ be an $n$-vertex weighted outerplanar graph. There exists an $O(n \log n)$-time algorithm that tests whether $G$ admits a  {\sc max}-constrai\-ned book-embedding and, in the positive case, constructs such an embedding.
\end{theorem}

We call {\sc max-be-drawer} the algorithm in the statement of Theorem \ref{th:linear-max-outerplanar}. We first describe such an algorithm for biconnected graphs and later extend it to simply-connected graphs. We have the following structural lemma.

\begin{lemma}\label{le:linear-max-outerplanar-biconnected-characterization}
Let $G=(V,E,\omega)$ be an $n$-vertex biconnected weighted outerplanar graph. If there exists a {\sc max}-constrained book-embedding $\mathcal{L}$ of $G$ then
\begin{enumerate}
    \item there is a single edge $e_M \in E$ of maximum weight;
    \item $e_M$ is incident to the outer face of the outerplane embedding of $G$;
    \item the endvertices of $e_M$ are the first and the last vertex of $\mathcal{L}$; and 
    \item $\mathcal{L}$ is unique, up to a flip.
\end{enumerate}
\end{lemma}
\begin{proof}
Suppose that a {\sc max}-constrained book-embedding $\mathcal{L}$ of $G$ exists, as otherwise there is nothing to prove. Since $G$ is a biconnected outerplanar graph, there exists an edge $e'$ of $G$ such that $e' \wraps e$ in $\mathcal{L}$, for each $e \in E$ such that $e \neq e'$; note that $\mathcal{L}$ induces an outerplanar drawing of $G$ such that $e'$ is incident to the outer face. By Property~\ref{pro:max-weigth} and by the fact that $\mathcal{L}$ is a {\sc max}-constrained book-embedding, we have that $\omega(e') > \omega(e)$ for any edge $e \neq e'$ in $E$. Therefore $e' = e_M$ is the unique edge of $G$ with maximum weight. 
Since $e_M \wraps e$, for each edge $e \in E$ such that $e \neq e_M$, it follows that the end-vertices of $e_M$ are the first and the last vertex in $\mathcal{L}$. Since $G$ is biconnected, it has a unique $1$-page book-embedding in which the end-vertices of $e_M$ are the first and the last vertex~\cite{mw-opbe-90,s-cog-79}. Therefore, $\mathcal{L}$ is unique, up to a flip.   
\end{proof}

A first algorithmic contribution is given in the following lemma.

\begin{lemma}\label{le:linear-max-outerplanar-biconnected-algo}
Let $G=(V,E,\omega)$ be an $n$-vertex biconnected weighted outerplanar graph. There exists an $O(n)$-time algorithm that tests whether $G$ admits a {\sc max}-constrained book-embedding and, in the positive case, constructs such an embedding in $O(n)$ time.
\end{lemma}
\begin{proof}
First, we determine in $O(n)$ time whether $G$ has a unique edge $e_M$ with maximum weight; if not, by Lemma~\ref{le:linear-max-outerplanar-biconnected-characterization} we can conclude that $G$ admits no {\sc max}-constrained book-embedding. By~\cite{d-iroga-07,m-laarogmog-79,w-rolt-87}, we can determine in $O(n)$ time the unique, up to a flip, $1$-page book-embedding $\mathcal{L}$ such that $e_M \wraps e$ for each edge $e \in E$ with $e \neq e_M$. 

It remains to test whether $\prec_\mathcal{L}$ meets the requirements of a {\sc max}-constrained book-embedding. We construct in $O(n)$ time the extended dual tree $\mathcal T$ of the outerplane embedding of $G$. We root $\mathcal T$ at the leaf $\rho$ such that the edge of $\mathcal T$ incident to $\rho$ is dual to $e_M$. We visit $\mathcal T$ and perform the following checks in total $O(n)$ time. Consider an edge $(\alpha,\beta)$ of $\mathcal T$ such that $\alpha$ is the parent of $\beta$ and let $e$ be the edge of $G$ dual to $(\alpha,\beta)$. Consider the edges $(\beta,\gamma_1),\dots,(\beta,\gamma_k)$ of $\mathcal T$ from $\beta$ to its children and let $e_1,\dots,e_k$ be the edges of $G$ dual to $(\beta,\gamma_1),\dots,(\beta,\gamma_k)$, respectively. For $i=1,\dots,k$, we check whether $\omega(e)>\omega(e_i)$. If one of these checks fails, we conclude that $G$ admits no {\sc max}-constrained book-embedding, otherwise $\mathcal{L}$ is a {\sc max}-constrained book-embedding of $G$.
\end{proof}

We now show how Algorithm {\sc max-be-drawer} deals with a not necessarily biconnected $n$-vertex outerplanar graph~$G$. We can assume that $G$ is connected. Indeed, if $G$ is not connected, then it admits a  {\sc max}-constrained book-embedding if and only if every connected component of it admits a  {\sc max}-constrained book-embedding.

First, we compute in $O(n)$ time the block-cut-vertex tree $T$ of $G$~\cite{h-gt-69,ht-aeagm-73}. We root $T$ at any block $b^*$ containing an edge $e_M$ with maximum weight. For a B-node $b$ of $T$, we denote by $G^+(b)$ the subgraph of $G$ consisting of all the blocks $G(b')$ such that $b'$ is a B-node in the subtree of $T$ rooted at $b$. Also, for each B-node $b$ of $T$ we compute in overall $O(n)$ time the value $W^+(b)$ of the maximum weight of an edge of $G^+(b)$.

We visit (in arbitrary order) $T$. For each B-node $b$, we perform the following checks and computations. 

\begin{enumerate}
\item We check whether $G(b)$ admits a  {\sc max}-constrained book-embedding; this can be done in a time that is linear in the number of vertices of $G(b)$, by Lemma~\ref{le:linear-max-outerplanar-biconnected-algo}. If not, we conclude that $G$ admits no  {\sc max}-constrained book-embedding (Failure Condition $1$). If yes, we compute a  {\sc max}-constrained book-embedding (again by Lemma~\ref{le:linear-max-outerplanar-biconnected-algo}) and call it $\mathcal{L}(b)$. 
\item If $b \neq b^*$, consider the C-node $c$ that is the parent of $b$ in $T$. We check in constant time whether $c$ is the first or the last vertex of $\mathcal{L}(b)$. If not, we conclude that $G$ admits no  {\sc max}-constrained book-embedding (Failure Condition $2$). Otherwise, we possibly flip in constant time $\mathcal{L}(b)$ so that $c$ is the first vertex of $\mathcal{L}(b)$.
\item For each C-node $c$ of $T$ that is adjacent to $b$, we store two values $\ell_b(c)$ and $r_b(c)$. These are the weights of the lowest-left and lowest-right edges incident to $c$ in $\mathcal{L}(b)$, respectively; if a vertex preceding or following $c$ in $\mathcal{L}(b)$ does not exist, then we set $\ell_b(c)$ or $r_b(c)$ to $\infty$, respectively. This can be done in constant time for each C-node.
\end{enumerate}


Algorithm {\sc max-be-drawer} now performs a bottom-up visit of $T$. After visiting a B-node $b$, we either conclude that $G$ admits no {\sc max}-constrained book-embedding or we determine a linear order $\mathcal{L}^+(b)$ for the vertices in $G^+(b)$ such that, if $b \neq b^*$, the parent of $b$ in $T$ is the first vertex of $\mathcal{L}^+(b)$. This is done as follows. 

If $b$ is a leaf of $T$, then we set in constant time $\mathcal{L}^+(b)=\mathcal{L}(b)$. 

If $b$ is an internal node of $T$, then we proceed as follows. We initialize $\mathcal{L}^+(b)$ to $\mathcal{L}(b)$; recall that the parent of $b$ in $T$, if $b \neq b^*$, is the first vertex of $\mathcal{L}(b)$.

Let $c_1,\dots,c_k$ be the C-nodes children of $b$ in $T$. For each $i=1,\dots,k$, let $b_{i,1},\dots,b_{i,m_i}$ be the B-nodes children of $c_i$. Since we already visited $b_{i,j}$, for $i=1,\dots,k$ and $j=1,\dots,m_i$, we have a linear order $\mathcal{L}^+(b_{i,j})$ of the vertices of $G^+(b_{i,j})$ such that $c_i$ is the first vertex of $\mathcal{L}^+(b_{i,j})$. We now process each C-node $c_i$ independently, for each $i=1,\dots,k$. 

We order the B-nodes $b_{i,1},\dots,b_{i,m_i}$ children of $c_i$ in decreasing order of value $W^+(b_{i,j})$; that is, $W^+(b_{i,1}) \geq W^+(b_{i,2}) \geq \dots \geq W^+(b_{i,m_i})$. This can be done in $O(m_i \log m_i)$ time. We now process the B-nodes $b_{i,1},\dots,b_{i,m_i}$ in this order (see Fig.~\ref{fig:max-biconnected}). When processing a node $b_{i,j}$, for $j=1,\dots,m_i$, we insert the vertices of $G^+(b_{i,j})$ into the ordering $\mathcal L^+(b)$, by replacing $c_i$ with either $\mathcal L^+(b_{i,j})$ (that is, $\mathcal L^+(b_{i,j})$ is inserted \emph{to the right} of $c_i$) or the flip of $\mathcal L^+(b_{i,j})$ (that is, $\mathcal L^+(b_{i,j})$ is inserted \emph{to the left} of $c_i$). This operation can be performed in constant time. Further, the choice on whether we insert $\mathcal L^+(b_{i,j})$ to the left or to the right of $c_i$ is performed as described in the following. 

We use two variables, called $L(c_i)$ and $R(c_i)$, and maintain the invariant that, while processing the B-nodes $b_{i,1},\dots,b_{i,m_i}$, they represent the weight of the lowest-left and lowest-right edges incident to $c_i$ in $\mathcal L^+(b)$. The variables $L(c_i)$ and $R(c_i)$ are initialized to $\ell_b(c_i)$ and $r_b(c_i)$, respectively, hence the invariant is satisfied before any B-node $b_{i,j}$ is processed.  

\begin{figure}[tb]
	\centering
	\subfloat[]{\label{fig:before-b-i-j}\includegraphics[page=3,scale=0.42]{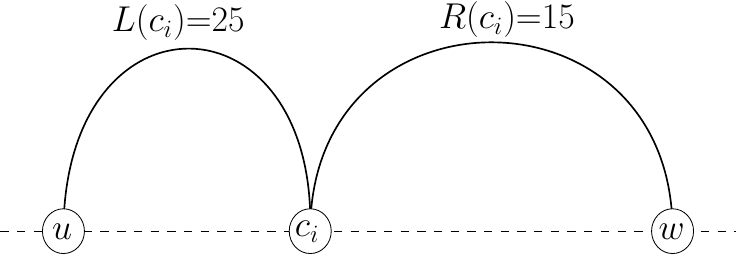}}
	\hfil
	\subfloat[]{\label{fig:right-b-i-j}\includegraphics[page=2,scale=0.42]{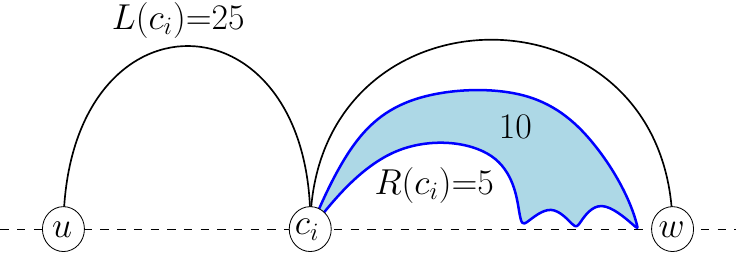}}
	\hfil
	\subfloat[]{\label{fig:left-b-i-j}\includegraphics[page=4,scale=0.42]{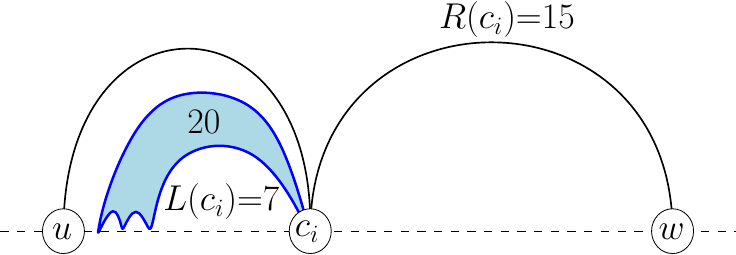}}
	\hfil
	\caption{A figure to illustrate how an ordering $\mathcal{L}^+(b_{i,j})$ is inserted into an ordering $\mathcal{L}^+(b)$. (a) The ordering $\mathcal{L}^+(b)$ before the insertion of $\mathcal{L}^+(b_{i,j})$; only $c_i$ and the lowest-left and lowest-right edges incident to $c_i$ in $\mathcal{L}^+(b)$ are shown; in this example, $L(c_i)=25$ and $R(c_i)=15$. (b) The ordering $\mathcal{L}^+(b)$ if $\mathcal{L}^+(b_{i,j})$ is inserted to the right of $c_i$, as it happens if $G^+(b_{i,j})$ is such that $W^+(b_{i,j})=10$; in this example, $r_{b_{i,j}}(c_i)=5$. (c) The ordering $\mathcal{L}^+(b)$ if $\mathcal{L}^+(b_{i,j})$ is inserted to the left of $c_i$, as it happens if $G^+(b_{i,j})$ is such that $W^+(b_{i,j})=20$; in this example, $r_{b_{i,j}}(c_i)=7$.}\label{fig:max-biconnected}
\end{figure}

\begin{itemize}
    \item If $W^+(b_{i,j}) \geq L(c_i)$ and $W^+(b_{i,j}) \geq R(c_i)$, then we conclude that $G$ admits no {\sc max}-constrained book-embedding (Failure Condition $3$).
    \item Otherwise, if $W^+(b_{i,j})< R(c_i)$, as in Figs.~\ref{fig:before-b-i-j} and~\ref{fig:right-b-i-j}, then we insert the vertices of $G^+(b_{i,j})$ into the ordering $\mathcal L^+(b)$, by replacing $c_i$ with $\mathcal L^+(b_{i,j})$;  we update $R(c_i)$ with the value of $r_{b_{i,j}}(c_i)$. 
    \item Otherwise, we have $W^+(b_{i,j}) \geq R(c_i)$ and $W^+(b_{i,j})< L(c_i)$, as in Figs.~\ref{fig:before-b-i-j} and~\ref{fig:left-b-i-j}; then we insert the vertices of $G^+(b_{i,j})$ into the ordering $\mathcal L^+(b)$, by replacing $c_i$ with the flip of $\mathcal L^+(b_{i,j})$; we update $L(c_i)$ with the value of $r_{b_{i,j}}(c_i)$.
\end{itemize}

When visiting the root $b^*$ of $T$, the algorithm computes an order $\mathcal{L}:=\mathcal{L}^+(b^*)$ of all the vertices of $G$. 

The next two lemmata prove the correctness of Algorithm {\sc max-be-drawer}.

\begin{lemma}\label{le:success-linear-max-outeplanar}
If Algorithm {\sc max-be-drawer} constructs an ordering $\mathcal{L}$, then $\mathcal{L}$ is a  {\sc max}-constrained book-embedding of~$G$.
\end{lemma}

\begin{proof}
We prove, by induction on $T$, that the linear order $\mathcal{L}^+(b)$ constructed by the algorithm is a  {\sc max}-constrained book-embedding of $G^+_{b}$. This implies the statement of the lemma with $b=b^*$. Our inductive proof also proves the following property for the constructed book-embeddings: If $b \neq b^*$, then the parent $c$ of $b$ is the first vertex in $\mathcal{L}^+(b)$. 

In the base case, $b$ is a leaf of $T$. Since Algorithm {\sc max-be-drawer} did not terminate because of Failure Condition 1, by Lemma~\ref{le:linear-max-outerplanar-biconnected-algo} we have that the order $\mathcal L^+(b)=\mathcal{L}(b)$ constructed by the algorithm is a  {\sc max}-constrained book-embedding of $G^+(b)=G(b)$. Further, since Algorithm {\sc max-be-drawer} did not terminate because of Failure Condition 2, we have that the parent $c$ of $b$ is the first vertex in $\mathcal{L}^+(b)=\mathcal{L}(b)$. 

In the inductive case, $b$ is a non-leaf node of $T$. Let $c_1,\dots,c_k$ and, for $i=1,\dots,k$, let $b_{i,1},\dots,b_{i,m_i}$ be defined as in the algorithm's description. By the property, we have that the linear order $\mathcal{L}^+(b_{i,j})$ is such that $c_i$ is the first vertex of $\mathcal{L}^+(b_{i,j})$, for each $i=1,\dots,k$ and $j=1,\dots,m_i$. Further, since the algorithm did not terminate because of Failure Condition 2, we have that, if $b \neq b^*$, the parent $c$ of $b$ is the first vertex in $\mathcal{L}(b)$. Recall that the algorithm initializes $\mathcal{L}^+(b)=\mathcal{L}(b)$.

Recall that the algorithm processes independently each C-node $c_i$ child of $b$. In order to argue that the insertion of the orders $\mathcal{L}^+(b_{i,j})$ into the order $\mathcal{L}^+(b)$ results in a  {\sc max}-constrained book-embedding of~$G^+(b)$ satisfying the property, we show that, for each $j=1,\dots,m_i$, after the insertion of the order $\mathcal{L}^+(b_{i,j})$ into $\mathcal{L}^+(b)$, we have that $L(c_i)$ and $R(c_i)$ are the weights of the lowest-left and of the lowest-right edges incident to $c_i$, respectively (where $L(c_i)=\infty$ or $R(c_i)=\infty$ if the lowest-left edge of $c_i$ or the lowest-right edge of $c_i$ is undefined, respectively). Observe that this is the case before the insertion of any order $\mathcal{L}^+(b_{i,j})$ into $\mathcal{L}^+(b)$, given that $L(c_i)$ and $R(c_i)$ are initialized to $\ell_b(c_i)$ and $r_b(c_i)$, respectively.    

When we insert an order $\mathcal{L}^+(b_{i,j})$ into $\mathcal L^+(b)$, we insert $\mathcal{L}^+(b_{i,j})$ to the right of $c_i$ only if $W^+(b_{i,j})<R(c_i)$. Since $R(c_i)$ is the weight of the lowest-right edge incident to $c_i$ in $\mathcal{L}^+(b)$ before the insertion of $\mathcal{L}^+(b_{i,j})$ and since all the edges of $G^+(b_{i,j})$ lie under the lowest-right edge incident to $c_i$, no edge of $G^+(b_{i,j})$ has a weight larger than the weight of the lowest-right edge incident to $c_i$. Then $\mathcal{L}^+(b)$ after the insertion is a  {\sc max}-constrained book-embedding, given that $\mathcal{L}^+(b)$  before the insertion and $\mathcal{L}^+(b_{i,j})$ are both  {\sc max}-constrained book-embeddings. Note that the lowest-right edge incident to $c_i$ after the insertion in $\mathcal{L}^+(b)$ is the lowest-right edge incident to $c_i$ in $\mathcal{L}^+(b_{i,j})$, and indeed the algorithm updates $R(c_i)=r_{b_{i,j}}(c_i)$, which is the weight of such an edge. For each cut-vertex $c_j$ different from $c_i$, both the lowest-right edge and the lowest-left edge incident to $c_j$ remain unchanged and so do the values $L(c_j)$ and $R(c_j)$. The case in which $\mathcal{L}^+(b_{i,j})$ is inserted to the left of $c_i$ in $\mathcal{L}^+(b)$ is analogous. Observe that, since Algorithm {\sc max-be-drawer} did not terminate because of Failure Condition 3, we have that $W^+(b_{i,j})<L(c_i)$ or $W^+(b_{i,j})<R(c_i)$ holds true.

If $b \neq b^*$, then, since the algorithm did not terminate because of Failure Condition 2, the parent $c$ of $b$ is the first vertex of $\mathcal{L}(b)$. Since the only block of $G^+(b)$ vertex $c$ belongs to is $G(b)$, we have that $c$ is the first vertex of $\mathcal{L}^+(b)$, as well. 
\end{proof}

\begin{lemma}\label{le:failure-linear-max-outeplanar}
If Algorithm {\sc max-be-drawer} fails, then $G$ does non admit a  {\sc max}-constrained book-embedding. 
\end{lemma}
\begin{proof}
Suppose that Algorithm {\sc max-be-drawer} fails. This can happen because of Failure Condition~$1$, $2$, or $3$. We discuss the three cases. 

Suppose that Failure Condition $1$ is verified for a B-node $b$ of $T$. It is immediate that a  {\sc max}-constrained book-embedding of $G$ restricted to the vertices and edges of $G(b)$ would yield a  {\sc max}-constrained book-embedding of $G(b)$. Hence, if $G(b)$ admits no  {\sc max}-constrained book-embedding, neither does $G$. 

\begin{figure}[htb]
	\centering
	\includegraphics[scale=0.5]{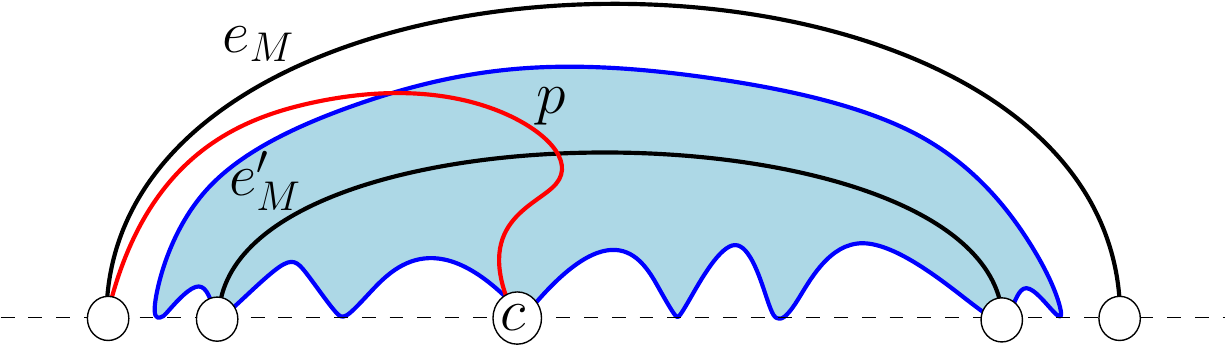}
	\caption{Illustration for the proof of the extreme-parent property. The shaded region represents $G^+(b)$.}\label{fig:extreme-property}
\end{figure}

In order to prove that, if Failure Condition $2$ is verified for a B-node $b \neq b^*$ of $T$, then $G$ admits no  {\sc max}-constrained book-embedding, we prove the following stronger statement (which we call the \emph{extreme-parent property}). Let $\mathcal L$ be any {\sc max}-constrained book-embedding of $G$, let $b$ be any B-node of $T$ different from $b^*$, let $c$ be the C-node parent of $b$ in $T$, and let $\mathcal L^+(b)$ be the {\sc max}-constrained book-embedding of $G^+(b)$ obtained by restricting $\mathcal L$ to the vertices and edges of $G^+(b)$. Then $c$ is the first or the last vertex of $\mathcal L^+(b)$. The extreme-parent property implies that, if Failure Condition $2$ is verified for a B-node $b \neq b^*$ of $T$, that is, if the parent $c$ of $b$ is neither the first nor the last vertex in the unique (up to a flip) {\sc max}-constrained book-embedding of $G(b)$, then $G$ admits no {\sc max}-constrained book-embedding.

We now prove the extreme-parent property. Suppose, for a contradiction, that $c$ is neither the first nor the last vertex of $\mathcal{L}^+(b)$; refer to Fig.~\ref{fig:extreme-property}. Since $c$ belongs to exactly one block of $G^+(b)$, namely $G(b)$, and since $G^+(b)$ is connected, the assumption that $c$ is neither the first nor the last vertex of $\mathcal{L}^+(b)$ implies that there exists an edge $e'_M$ of $G^+(b)$ whose end-vertices are one before and one after $c$ in $\mathcal{L}^+(b)$. Consider the path $P$ in $T$ from $c$ to $b^*$. Consider any path $p$ in $G$ whose vertices and edges belong to the blocks corresponding to B-nodes in $P$ and whose end-vertices are $c$ and one of the end-vertices of $e_M$ different from $c$ (recall that $e_M$ is an edge of $G$ with maximum weight and belongs to $G(b^*)$). Since $c$ is the cut-vertex parent of $b$ and $b^*$ is the root of $T$, we have that neither $p$ nor $e_M$ contains any vertex of $G^+(b)$ except, possibly, for $c$; in particular, neither $p$ nor $e_M$ contains either of the end-vertices of $e'_M$. Since $\omega(e_M) \geq \omega(e'_M)$, we have that $e_M$ is not nested into $e'_M$ in $\mathcal{L}^+(b)$. Hence, we have that $p$ crosses $e'_M$, a contradiction which proves the extreme-parent property.

Suppose that Failure Condition $3$ is verified for a B-node $b_{i,j}$ which is a child of a C-node $c_i$ whose parent B-node is $b$, that is, $W^+(b_{i,j}) \geq L(c_i)$ and $W^+(b_{i,j}) \geq R(c_i)$. We prove that this implies that $G$ admits no {\sc max}-constrained book-embedding. In order to do that, we are going to exploit the extreme-parent property, as well as the following observation: Let $b'\neq b^*$ be a B-node of $T$, let $c'$ be the parent of $b'$ in $T$, and let $\mathcal L^+(b')$ be a {\sc max}-constrained book-embedding of $G^+(b')$ such that $c'$ is the first (resp.\ last) vertex of $\mathcal L^+(b')$; then the weight of the lowest-right (resp.\ lowest-left) edge incident to $c'$ in $\mathcal L^+(b')$ is equal to the smallest weight of any edge incident to $c'$ in $G^+(b')$. Indeed, if the observation were not true, the  smallest-weight edge incident to $c'$ in $G^+(b')$ would wrap around a different edge incident to $c'$ in $\mathcal L^+(b')$, which would violate the conditions of a {\sc max}-constrained book-embedding. Let $w^+(b')$ denote the minimum weight of any edge incident to the parent $c'$ of $b'$ in $G^+(b')$.

Recall that $L(c_i)$ and $R(c_i)$ are the weights of the lowest-left and lowest-right edges incident to $c_i$ in $\mathcal L^+(b)$ before the temptative insertion of $\mathcal L(b_{i,j})$. Let $b_{i,\ell}$ and $b_{i,r}$ the B-nodes such that $\mathcal L^+(b_{i,l})$ and $\mathcal L^+(b_{i,r})$ were the last orders inserted to the left and to the right of $c_i$, respectively, before processing $b_{i,j}$. 
Observe that one or both of $b_{i,\ell}$ and $b_{i,r}$ may not exist. We distinguish four cases.

\begin{itemize}
    \item Suppose first that both $b_{i,\ell}$ and $b_{i,r}$ exist. Then the lowest-left and lowest-right edges incident to $c_i$ in $\mathcal L^+(b)$ before the temptative insertion of $\mathcal L(b_{i,j})$ belong to $G^+(b_{i,\ell})$ and $G^+(b_{i,r})$, respectively. Then, by the above observation, the inequalities $W^+(b_{i,j})\geq L(c_i)$ and  $W^+(b_{i,j})\geq R(c_i)$ of Failure Condition 3 imply that $W^+(b_{i,j})\geq w^+(b_{i,\ell})$ and  $W^+(b_{i,j})\geq w^+(b_{i,r})$.
    
    By the extreme-parent property, in any {\sc max}-constrained book-embedding of $G$, the vertex $c_i$ is the first or the last vertex among the ones of $G^+(b_{i,\ell})$, of $G^+(b_{i,r})$, and of $G^+(b_{i,j})$; that is, $G^+(b_{i,\ell})$ lies entirely to the left or entirely to the right of $c_i$, and so do $G^+(b_{i,r})$ and $G^+(b_{i,j})$. 
    
    Further, $G^+(b_{i,\ell})$ and $G^+(b_{i,r})$ cannot lie on the same side of $c_i$. Namely,  because of the ordering of the B-nodes that are children of $c_i$, we have that $W^+(b_{i,\ell}) \geq W^+(b_{i,j})$; by $W^+(b_{i,j})\geq w^+(b_{i,r})$ it then follows that $W^+(b_{i,\ell}) \geq w^+(b_{i,r})$, and hence $G^+(b_{i,\ell})$ cannot lie under $G^+(b_{i,r})$. Analogously, we have that $W^+(b_{i,r}) \geq W^+(b_{i,j}) \geq w^+(b_{i,\ell})$, hence $G^+(b_{i,r})$ cannot lie under $G^+(b_{i,\ell})$.
    
    By $W^+(b_{i,j})\geq w^+(b_{i,\ell})$, it directly follows that $G^+(b_{i,j})$ cannot lie under $G^+(b_{i,\ell})$. Moreover, $G^+(b_{i,\ell})$ cannot lie under $G^+(b_{i,j})$, given that $W^+(b_{i,\ell}) \geq W^+(b_{i,j}) \geq w^+(b_{i,j})$. Hence, $G^+(b_{i,\ell})$ and $G^+(b_{i,j})$ cannot lie on the same side of $c_i$. An analogous proof shows that $G^+(b_{i,r})$ and $G^+(b_{i,j})$ cannot lie on the same side of $c_i$.
    
    Since at least two out of $G^+(b_{i,\ell})$, $G^+(b_{i,r})$, and $G^+(b_{i,j})$ have to lie on the same side of $c_i$, it follows that $G$ admits no {\sc max}-constrained book-embedding.
    
    \item Suppose next that $b_{i,\ell}$ exists and $b_{i,r}$ does not. Then the lowest-left edge incident to $c_i$ in $\mathcal L^+(b)$ before the temptative insertion of $\mathcal L(b_{i,j})$ belongs to $G^+(b_{i,\ell})$. By the above observation, the inequality $W^+(b_{i,j})\geq L(c_i)$ of Failure Condition 3 implies that $W^+(b_{i,j})\geq w^+(b_{i,\ell})$. Further, since $W^+(b_{i,j})\geq R(c_i)$, we have that the lowest-right edge $e_r$ incident to $c_i$ in $\mathcal L^+(b)$ before the temptative insertion of $\mathcal L(b_{i,j})$ exists (as otherwise we would have $R(c_i)=\infty$) and belongs to $G(b)$; then $W^+(b_{i,j})\geq R(c_i)$ implies that $W^+(b_{i,j})\geq \omega(e_r)$.
    
    By the extreme-parent property, in any {\sc max}-constrained book-embedding of $G$, the graph $G^+(b_{i,\ell})$ lies entirely to the left or entirely to the right of $c_i$, and so does $G^+(b_{i,j})$. 
    
    By $W^+(b_{i,j})\geq w^+(b_{i,\ell})$, it directly follows that $G^+(b_{i,j})$ cannot lie under $G^+(b_{i,\ell})$. Moreover, $G^+(b_{i,\ell})$ cannot lie under $G^+(b_{i,j})$, given that $W^+(b_{i,\ell}) \geq W^+(b_{i,j}) \geq w^+(b_{i,j})$. Hence, $G^+(b_{i,\ell})$ and $G^+(b_{i,j})$ cannot lie on the same side of $c_i$.
    
    Further, neither $G^+(b_{i,j})$ nor $G^+(b_{i,\ell})$ can lie under $e_r$. This follows by $W^+(b_{i,\ell})\geq W^+(b_{i,j})\geq \omega(e_r)$. Hence, neither $G^+(b_{i,j})$ nor $G^+(b_{i,\ell})$ can lie under $G(b)$.
    
    Finally, $G(b)$ cannot lie under $G^+(b_{i,j})$ or $G^+(b_{i,\ell})$, as this would violate the extreme-parent property (if $b\neq b^*$) or would imply that $e_M$ is nested into an edge of $G^+(b_{i,j})$ or $G^+(b_{i,\ell})$ (if $b= b^*$). 
    \item The case in which $b_{i,r}$ exists and $b_{i,\ell}$ does not is symmetric to the previous one. 
    \item Finally, suppose that neither $b_{i,\ell}$ nor $b_{i,r}$ exists. Since $W^+(b_{i,j})\geq L(c_i)$ and $W^+(b_{i,j})\geq R(c_i)$, it follows that the lowest-left and lowest-right edges incident to $c_i$ in the unique (up to a flip) embedding $\mathcal L(b)$ of $G(b)$ both exist and have a weight not larger than $W^+(b_{i,j})$. Hence, $G^+(b_{i,j})$ cannot lie under $G(b)$; further, $c_i$ is neither the first nor the last vertex of $\mathcal L(b)$ (as otherwise we would have $L(c_i)=\infty$ or $R(c_i)=\infty$, respectively). The latter, together with the biconnectivity of $G(b)$, also implies that $G(b)$ cannot lie under $G^+(b_{i,j})$. It follows that $G$ admits no {\sc max}-constrained book-embedding.
 \end{itemize}

This concludes the proof of the lemma.
\end{proof}

Lemmata~\ref{le:success-linear-max-outeplanar} and~\ref{le:failure-linear-max-outeplanar} prove the correctness of Algorithm {\sc max-be-drawer}. Its running time is dominated by the $O(m_i \log m_i)$-time sorting that is performed on the $m_i$ children of each C-node $c_i$. Hence, the overall time complexity is $O(n \log n)$. This concludes the proof of Theorem~\ref{th:linear-max-outerplanar}.

The upper bound in Theorem~\ref{th:linear-max-outerplanar} is tight, as computing a {\sc max}-constrained book-embedding has a time complexity that is lower-bounded by that of a sorting algorithm. Indeed, given a set $S$ of $n$ distinct real numbers, one can construct a star $T$ with a center $c$ whose $n$ edges have the weights in $S$. Any {\sc max}-constrained book-embedding of $T$ partitions the edges into two ordered sequences, one to the left of $c$ and one to the right of $c$; a total ordering of $S$ can be constructed by merging these sequences in $O(n)$ time.    

\section{{\sc sum}-Constrained Book-Embeddings}\label{se:sum}


\begin{figure}[tb]
	\centering
	\subfloat[]{\label{fig:sum-constrained-embedding}\includegraphics[width=0.45\columnwidth]{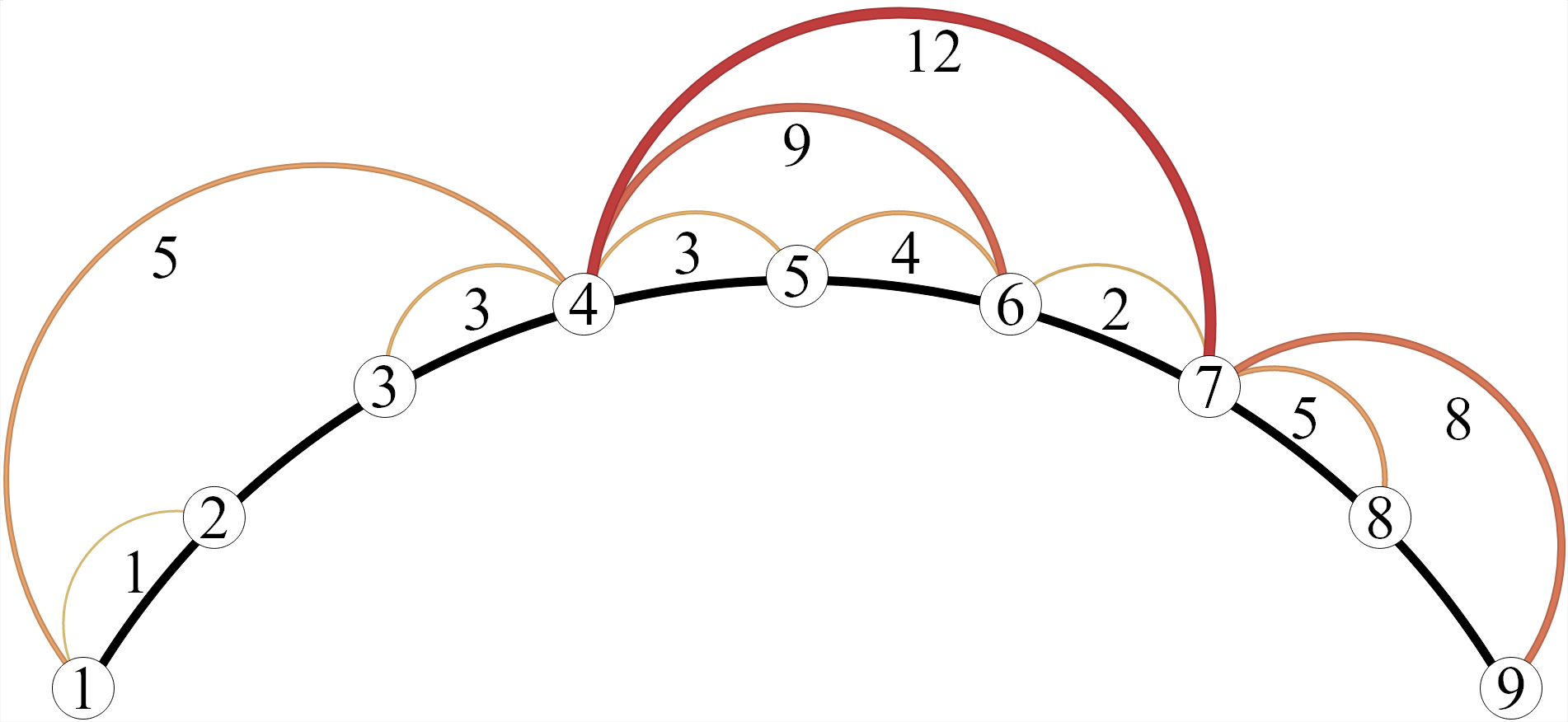}}
	\hfil
	\subfloat[]{\label{fig:minres-constrained-embedding}\includegraphics[width=0.45\columnwidth]{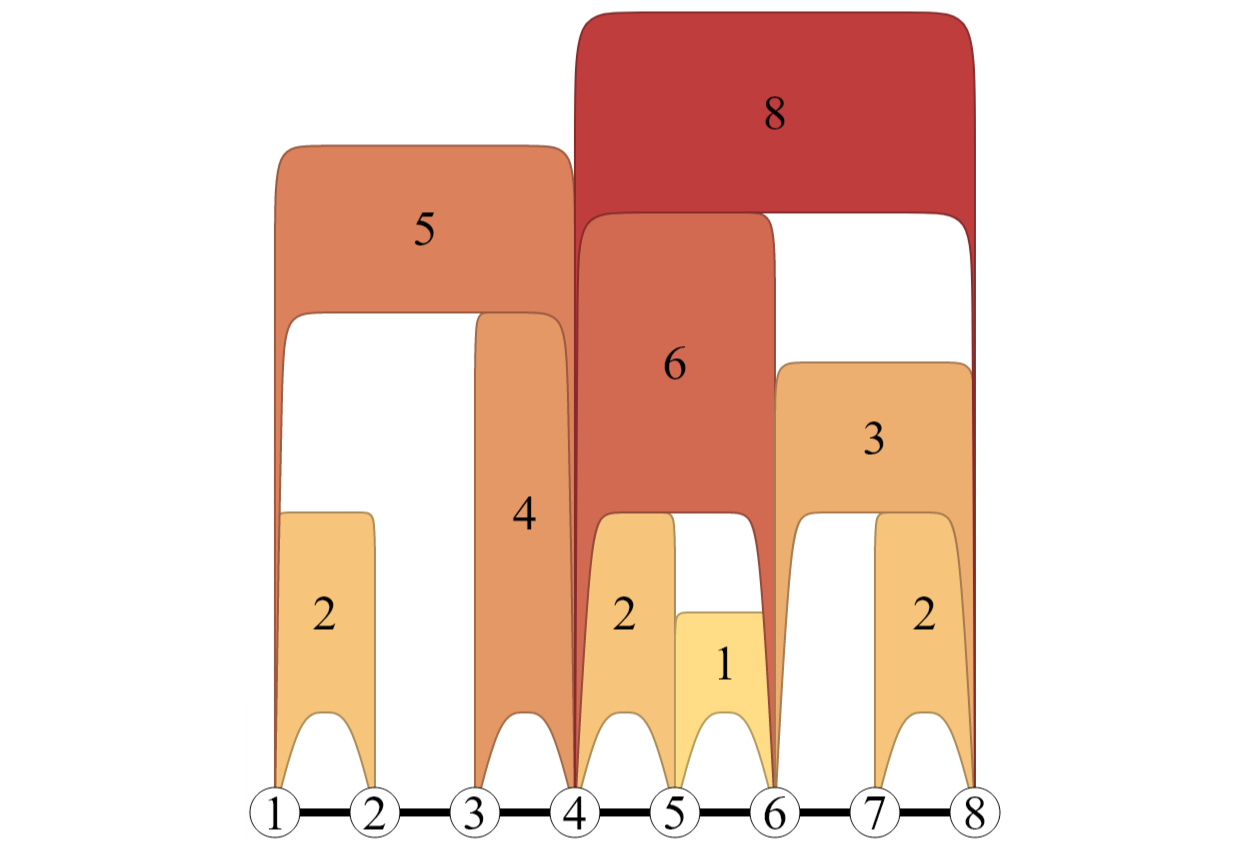}}
	\hfil
	\caption{Schematic representations of biconnected graphs. (a) A $1$-page {\sc sum}-constrained book-embedding. (b) A {\sc minres}-constrained two-dimensional book-embedding; for simplicity the vertices are aligned on a straight-line.}\label{fig:sum-minres-figure}
\end{figure}

Even if in a {\sc max}-constrained book-embedding an edge cannot wrap around an edge with a larger weight, we may still have that an edge $e$ that wraps around a sequence of edges $e_1,\dots,e_k$ with $\omega(e)<\sum_{i=1}^k \omega(e_i)$. This might cause the resulting visualization to not effectively convey the information related to the edge weights. Hence, we study a second type of one-dimensional representations that are more restrictive than {\sc max}-constrained book-embeddings and that allow us to better take into account the relationships between the weights of the edges.

A \emph{{\sc sum}-constrained book-embedding} of a weighted outerplanar graph $G=(V,E,\omega)$ is a $1$-page book-embedding $\mathcal{L}$ satisfying the following property. Let $e=(u,v)$ be any edge in $E$ with $u \prec_{\mathcal{L}} v$. Let $e_1=(u_1,v_1),\dots,e_k=(u_k,v_k)$ be any sequence of edges in $E$ such that $u \preceq_{\mathcal{L}} u_1 \prec_{\mathcal{L}} v_1 \preceq_{\mathcal{L}} \dots \preceq_{\mathcal{L}} u_k \prec_{\mathcal{L}} v_k \preceq_{\mathcal{L}} v$. Then $\omega(e)>\sum_{i=1}^k \omega(e_i)$. Observe that the {\sc max}-constrained book-embedding of Fig. \ref{fig:max-constrained-embedding} is not a {\sc sum}-constrained book-embedding, since it contains vertices $3$, $4$, $5$, and $7$ (in this order) and the sum of the weights of $(3,4)$ and $(5,7)$ is $14$, while the weight of $(3,7)$ is $12$. An example of {\sc sum}-constrained book-embedding is in Fig. \ref{fig:sum-constrained-embedding}.

The goal of this section is to prove the following theorem.

\begin{theorem}\label{th:linear-sum-outerplanar}
Let $G=(V,E,\omega)$ be an $n$-vertex weighted outerplanar graph. There exists an $O(n^3 \log n)$-time algorithm that tests whether $G$ admits a {\sc sum}-constrained book-embedding and, in the positive case, constructs such an embedding. 
\end{theorem}

We first deal with biconnected outerplanar graphs. Note that Lemma~\ref{le:linear-max-outerplanar-biconnected-characterization} holds true also in the current setting, given that a {\sc sum}-constrained book-embedding is a {\sc max}-constrained book-embedding. We get the following lemma, whose proof follows almost verbatim the one of Lemma~\ref{le:linear-max-outerplanar-biconnected-algo}.

\begin{lemma}\label{le:linear-sum-outerplanar-biconnected-algo}
Let $G=(V,E,\omega)$ be an $n$-vertex biconnected weighted outerplanar graph. There exists an $O(n)$-time algorithm that tests whether $G$ admits a {\sc sum}-constrained book-embedding and, in the positive case, \mbox{constructs such an embedding.}
\end{lemma}

\begin{proof}
First, we determine in $O(n)$ time whether $G$ has a unique edge $e_M$ with maximum weight; if not, by Lemma~\ref{le:linear-max-outerplanar-biconnected-characterization} we can conclude that $G$ admits no {\sc max}-constrained book-embedding. By~\cite{d-iroga-07,m-laarogmog-79,w-rolt-87}, we can determine in $O(n)$ time the unique, up to a flip, $1$-page book-embedding $\mathcal{L}$ such that $e_M \wraps e$ for each edge $e \in E$ with $e \neq e_M$. 

It remains to test whether $\prec_\mathcal{L}$ meets the requirements of a {\sc max}-constrained book-embedding. We construct in $O(n)$ time the extended dual tree $\mathcal T$ of the outerplane embedding of $G$. We root $\mathcal T$ at the leaf $\rho$ such that the edge of $\mathcal T$ incident to $\rho$ is dual to $e_M$. We visit $\mathcal T$ and perform the following checks in total $O(n)$ time. Consider an edge $(\alpha,\beta)$ of $\mathcal T$ such that $\alpha$ is the parent of $\beta$ and let $e$ be the edge of $G$ dual to $(\alpha,\beta)$. Consider the edges $(\beta,\gamma_1),\dots,(\beta,\gamma_k)$ of $\mathcal T$ from $\beta$ to its children and let $e_1,\dots,e_k$ be the edges of $G$ dual to $(\beta,\gamma_1),\dots,(\beta,\gamma_k)$, respectively. For $i=1,\dots,k$, we check whether $\omega(e)>\sum_{i=1}^k\omega(e_i)$. If one of these checks fails, we conclude that $G$ admits no {\sc max}-constrained book-embedding, otherwise $\mathcal{L}$ is a {\sc max}-constrained book-embedding of $G$.
\end{proof}

We now deal with a not necessarily biconnected $n$-vertex outerplanar graph~$G$. As for {\sc max}-constrained book-embeddings, we can assume that $G$ is connected. We present an algorithm, called {\sc sum-be-drawer}, that tests in $O(n^3 \log n)$ time whether $G$ admits a {\sc sum}-constrained book-embedding and, in the positive case, constructs such an embedding. 

First, we compute in $O(n)$ time the block-cut-vertex tree $T$ of $G$~\cite{h-gt-69,ht-aeagm-73}. We root $T$ at any B-node $b^*$ containing an edge with maximum weight. Then, for a B-node $b$, the graph $G^+(b)$ is defined as for {\sc max}-constrained book-embeddings; further, for a C-node $c$ of $T$, we denote by $G^+(c)$ the subgraph of $G$ consisting of all the blocks $G(b')$ such that $b'$ is a B-node in the subtree of $T$ rooted at $c$. We equip each B-node $b$ with the maximum weight $W(b)$ of any edge of $G(b)$.

We visit (in arbitrary order) $T$. For each B-node $b$, the algorithm {\sc sum-be-drawer} performs the following checks and computations. 

\begin{enumerate}
\item We check whether $G(b)$ admits a  {\sc sum}-constrained book-embedding;  this can be done in a time that is linear in the number of vertices of $G(b)$, by Lemma~\ref{le:linear-sum-outerplanar-biconnected-algo}. If not, we conclude that $G$ admits no  {\sc sum}-constrained book-embedding (Failure Condition $1$). If yes, we compute a  {\sc sum}-constrained book-embedding (again by Lemma~\ref{le:linear-sum-outerplanar-biconnected-algo}) and call it $\mathcal{L}(b)$. 
\item If $b \neq b^*$, consider the C-node $c$ that is the parent of $b$ in $T$. We check in constant time whether $c$ is the first or the last vertex of $\mathcal{L}(b)$. If not, we conclude that $G$ admits no  {\sc sum}-constrained book-embedding (Failure Condition $2$). Otherwise, we possibly flip in constant time $\mathcal{L}(b)$ so that $c$ is the first vertex of~$\mathcal{L}(b)$.
\end{enumerate}

\begin{figure}[htb]
	\centering
	\subfloat[$\tau_{\mathcal  L}=21$, $\alpha_{\mathcal  L}=1$, $\lambda_{\mathcal  L}(4)=9$, $\rho_{\mathcal  L}(4)=12$]{\includegraphics[page=1,width=0.48\columnwidth]{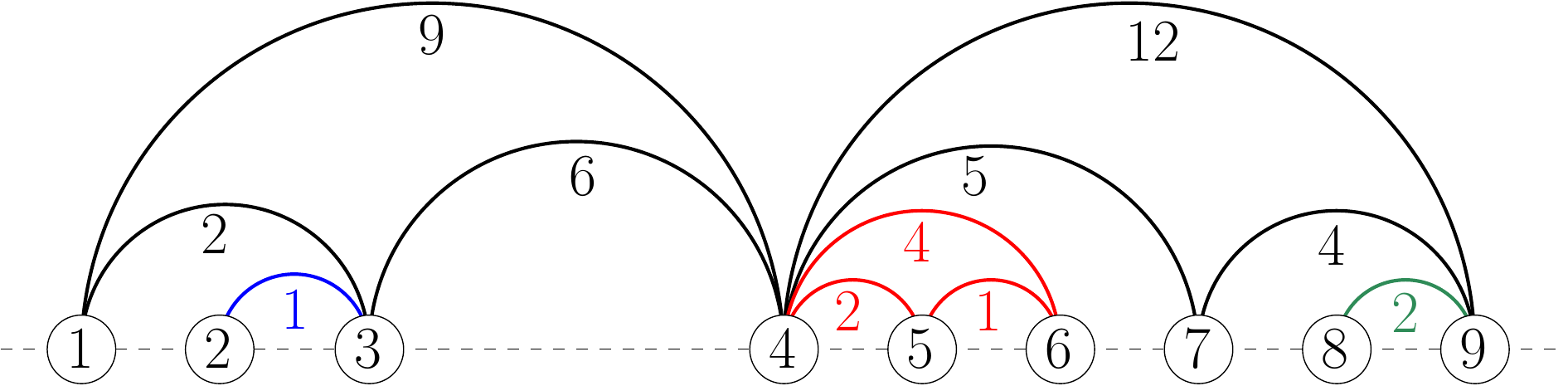}
	\label{fig:max1}}
	\hfil
	\subfloat[$\tau_{\mathcal  L}=21$, $\alpha_{\mathcal  L}=2$, $\lambda_{\mathcal  L}(4)=9$, $\rho_{\mathcal  L}(4)=12$]{\includegraphics[page=2,width=0.48\columnwidth]{figures/esempi-max.pdf}
	\label{fig:max2}}
	\hfil
	\subfloat[$\tau_{\mathcal  L}=21$, $\alpha_{\mathcal  L}=2$, $\lambda_{\mathcal  L}(4)=9$, $\rho_{\mathcal  L}(4)=12$]{\includegraphics[page=3,width=0.48\columnwidth]{figures/esempi-max.pdf}
	\label{fig:max3}}
	\hfil
	\subfloat[$\tau_{\mathcal  L}=23$, $\alpha_{\mathcal  L}=1$, $\lambda_{\mathcal  L}(4)=9$, $\rho_{\mathcal  L}(4)=14$]{\includegraphics[page=4,width=0.48\columnwidth]{figures/esempi-max.pdf}
	\label{fig:max4}}
	\hfil
	\caption{(a) and (b) are left-right equivalent w.r.t. $4$; (c) left-right dominates (d) w.r.t. $4$; (b) and (c) are up-down equivalent; (b) up-down dominates (a).}\label{fig:esempi-max}
\end{figure}

We introduce some definitions (refer to Fig.~\ref{fig:esempi-max}).
Let $\mathcal{L}$ be a $1$-page book-embedding of $G$. We say that a vertex $c$ is \emph{visible} if there exists no edge $e$ of $G$ such that $c$ is strictly under $e$ in $\mathcal L$; for example, the vertices $1$, $4$, and $9$ in Fig.~\ref{fig:max1} are visible. 

The \emph{total extension} $\tau_{\mathcal L}$ of $\mathcal{L}$ is the sum of the weights of all the edges $e$ that satisfy the following property: there is no edge $e'$ such that $e'\wraps e$ in $\mathcal{L}$.

Let $c$ be a visible vertex of $\mathcal L$. Then the \emph{extension of $\mathcal{L}$ to the left of $c$} is the sum of the weights of all the edges $e$ that satisfy the following properties: (i) there is no edge $e'$ such that $e'\wraps e$ in $\mathcal{L}$; and (ii) for each end-vertex $v$ of $e$, we have $v\preceq_{\mathcal{L}} c$. The \emph{extension of $\mathcal{L}$ to the right of $c$} is defined analogously. The extensions of $\mathcal{L}$ to the left and to the right of $c$ are denoted by  $\lambda_{\mathcal{L}}(c)$ and $\rho_{\mathcal{L}}(c)$, respectively. 

Let $u$ be the first vertex of $\mathcal L$. The \emph{free space $\alpha_{\mathcal L}$ of $\mathcal{L}$} is the weight of the lowest-right edge $(u,v)$ of $u$ in $\mathcal{L}$ minus the extension of the subgraph of $G$ induced by $v$ and by the vertices that are strictly under $(u,v)$. 

Now, let $\mathcal{L}$ and $\mathcal{L}'$ be two $1$-page book-embeddings of $G$ and let $c$ be a vertex of $G$ that is visible both in $\mathcal{L}$ and in $\mathcal{L}'$. 
We say that $\mathcal{L}$ and $\mathcal{L}'$ are \emph{left-right equivalent} with respect to $c$ if $\lambda_{\mathcal L}(c) = \lambda_{\mathcal L'}(c)$ and $\rho_{\mathcal L}(c)= \rho_{\mathcal L'}(c)$. We also say that $\mathcal{L}$ \emph{left-right dominates} $\mathcal{L}'$ with respect to $c$ if $\lambda_{\mathcal L}(c) \leq \lambda_{\mathcal L'}(c)$, $\rho_{\mathcal L}(c)\leq \rho_{\mathcal L'}(c)$, and at least one of the two inequalities is strict. 

If the first vertex of $\mathcal{L}$ is the same as the first vertex of $\mathcal{L}'$, we say that $\mathcal{L}$ is \emph{up-down equivalent} to $\mathcal{L}'$ if $\tau_{\mathcal L} = \tau_{\mathcal L'}$ and $\alpha_{\mathcal L} = \alpha_{\mathcal L'}$. Further, we say that $\mathcal{L}$ \emph{up-down dominates} $\mathcal{L}'$ if $\tau_{\mathcal L} \leq \tau_{\mathcal L'}$, $\alpha_{\mathcal L} \geq \alpha_{\mathcal L'}$, and at least one of the two inequalities is strict.

The algorithm {\sc sum-be-drawer} now performs a bottom-up visit of $T$. 

After visiting each C-node $c$, the algorithm {\sc sum-be-drawer} either concludes that $G$ admits no {\sc sum}-constrained book-embedding or determines a sequence of {\sc sum}-constrained book-embeddings $\mathcal{L}^+_1(c), \dots, \mathcal{L}^+_{k}(c)$ of $G^+(c)$ such~that: 
\begin{enumerate}[(C1)]
    \item\label{gio} for any $i=1,\dots,k$, we have that $c$ is visible in $\mathcal{L}^+_i(c)$; 
    \item $\lambda_{\mathcal L^+_1(c)}(c)<\dots<\lambda_{\mathcal L^+_k(c)}(c)$ and $\rho_{\mathcal L^+_1(c)}(c)>\dots>\rho_{\mathcal L^+_k(c)}(c)$; and
    \item for every {\sc sum}-constrained book-embedding $\mathcal L$ of $G^+(c)$ that respects (C1), there exists an index $i\in \{1,\dots,k\}$ such that $\mathcal{L}^+_i(c)$ left-right dominates or is left-right equivalent to $\mathcal L$ with respect to $c$.
\end{enumerate}

Note that no {\sc sum}-constrained book-embedding $\mathcal{L}^+_i(c)$ left-right dominates or is left-right equivalent to a distinct embedding  $\mathcal{L}^+_j(c)$ with respect to $c$, by Property (C2).

After visiting a B-node $b \neq b^*$, the algorithm {\sc sum-be-drawer} either concludes that $G$ admits no {\sc sum}-constrained book-embedding or determines a sequence of {\sc sum}-constrained book-embeddings $\mathcal{L}^+_1(b), \dots, \mathcal{L}^+_k(b)$ \mbox{of $G^+(b)$ such that:}
\begin{enumerate}[(B1)]
    \item the parent $c$ of $b$ in $T$ is the first vertex of $\mathcal{L}_i^+(b)$, for $i=1,\dots,k$;
    \item $\alpha_{\mathcal L^+_1(b)}<\dots<\alpha_{\mathcal L^+_k(b)}$ and $\tau_{\mathcal L^+_1(b)}<\dots<\tau_{\mathcal L^+_k(b)}$; and
    \item for every {\sc sum}-constrained book-embedding $\mathcal L$ of $G^+(b)$ that respects (B1), there exists an index $i\in \{1,\dots,k\}$ such that $\mathcal{L}^+_i(b)$ up-down dominates or is up-down equivalent to $\mathcal L$.
\end{enumerate}

Note that no {\sc sum}-constrained book-embeddings $\mathcal{L}^+_i(b)$ up-down dominates or is up-down equivalent to  a distinct embedding $\mathcal{L}^+_j(b)$, by Property (B2).

Restricting the attention to embeddings satisfying Condition~(C1) or Condition~(B1) is not a loss of generality, because of the following two lemmata.

\begin{lemma} \label{le:necessity-c1}
Suppose that $G$ admits a {\sc sum}-constrained book-embedding $\mathcal L$. Let $c$ be a C-node of $T$ and let $\mathcal L^+(c)$ be the restriction of $\mathcal L$ to the vertices and edges of $G^+(c)$. Then $c$ is visible in $\mathcal L^+(c)$.
\end{lemma}

\begin{proof}
This proof is very similar to the one of the \emph{extreme-parent property} in Lemma~\ref{le:failure-linear-max-outeplanar}. 

Suppose, for a contradiction, that $c$ is not visible in $\mathcal{L}^+(c)$; that is, there exists an edge $e'_M$ of $G^+(c)$ whose end-vertices are one before and one after $c$ in $\mathcal{L}^+(c)$. Consider the path $P$ in $T$ from $c$ to $b^*$. Further, consider any path $p$ in $G$ whose vertices and edges belong to the blocks corresponding to B-nodes in $P$ and whose end-vertices are $c$ and one of the end-vertices of $e_M$ different from $c$ (recall that $e_M$ is an edge of $G$ with maximum weight and belongs to $G(b^*)$). Since $b^*$ is the root of $T$, we have that neither $p$ nor $e_M$ contains any vertex of $G^+(c)$ except, possibly, for $c$; in particular, neither $p$ nor $e_M$ contains either of the end-vertices of $e'_M$. Since $\omega(e_M) \geq \omega(e'_M)$, we have that $e_M$ is not nested into $e'_M$ in $\mathcal{L}^+(b)$. Hence, we have that $p$ crosses $e'_M$, a contradiction.
\end{proof}

\begin{lemma} \label{le:necessity-b1}
Suppose that $G$ admits a {\sc sum}-constrained book-embedding $\mathcal L$. Let $b\neq b^*$ be a B-node of $T$ and let $\mathcal L^+(b)$ be the restriction of $\mathcal L$ to the vertices and edges of $G^+(b)$. Then the parent $c$ of $b$ in $T$ is either the first or the last vertex of $\mathcal L^+(b)$.
\end{lemma}

\begin{proof}
The lemma asserts that $\mathcal L^+(b)$ satisfies the \emph{extreme-parent property}; this property, which was stated in the context of {\sc max}-constrained book-embeddings, was shown to be satisfied in the proof of Lemma~\ref{le:failure-linear-max-outeplanar}. Since the {\sc sum}-constrained book-embedding $\mathcal L^+(b)$ is also a {\sc max}-constrained book-embedding, that proof can be followed verbatim to prove the statement of the lemma. 
\end{proof}

We are also going to use the following two lemmata, which bound the number of distinct {\sc sum}-constrained book-embeddings we construct during the visit of $T$.

\begin{lemma} \label{le:number-of-orderings-C}
Let $H=(V_H,E_H,\omega_H)$ be an $n$-vertex weighted outerplanar graph. For a vertex $c$ of $H$, let $\mathcal S$ be a set of {\sc sum}-constrained book-embeddings of $H$ such that:

\begin{enumerate}[$(\gamma 1)$]
    \item for each $\mathcal L \in \mathcal S$, we have that $c$ is visible in  $\mathcal{L}$; and
    \item for any $\mathcal L,\mathcal L' \in \mathcal S$, we have that $\mathcal L$ does not left-right dominate and is not left-right equivalent to $\mathcal L'$ with respect to $c$.
\end{enumerate}
Then  $\mathcal S$ contains $O(n)$ embeddings.
\end{lemma}

\begin{proof}
The proof is based on the following two claims.

First, for any value $\lambda\geq 0$, there exists at most one {\sc sum}-constrained book-embedding $\mathcal L \in \mathcal S$ whose extension $\lambda_{\mathcal L}(c)$ to the left of $c$ is equal to $\lambda$. Indeed, suppose, for a contradiction, that $\mathcal S$ contains two {\sc sum}-constrained book-embeddings $\mathcal L$ and $\mathcal L'$ in which $c$ is visible with  $\lambda_{\mathcal L}(c)=\lambda_{\mathcal L'}(c)=\lambda$. If $\tau_{\mathcal L}(c)<\tau_{\mathcal L'}(c)$, or $\tau_{\mathcal L}(c)=\tau_{\mathcal L'}(c)$, or $\tau_{\mathcal L}(c)>\tau_{\mathcal L'}(c)$, we have that $\mathcal L$ left-right dominates $\mathcal L'$, or that $\mathcal L$ is left-right equivalent to $\mathcal L'$, or that $\mathcal L'$ left-right dominates $\mathcal L$ with respect to $c$, respectively; in all the cases, this contradicts Property~$(\gamma2)$. It follows that the number of embeddings $\mathcal L$ in $\mathcal S$ is at most equal to the number of distinct values $\lambda\geq 0$ such that $H$ admits a {\sc sum}-constrained book-embedding $\mathcal L$ in which $c$ is visible and $\lambda_{\mathcal L}(c)=\lambda$.

\begin{figure}[htb]
	\centering
	\includegraphics[scale=0.5]{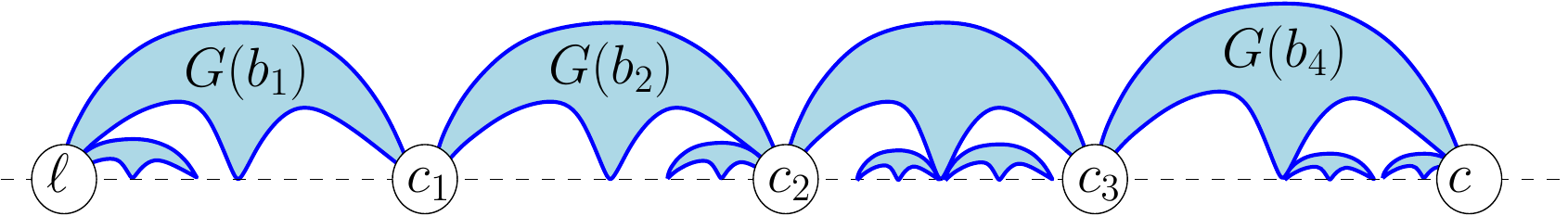}
    \caption{Illustration for the proof of Lemma~\ref{le:number-of-orderings-C}. In this {\sc sum}-constrained book-embedding, only $c$ and the vertices to the left of $c$ are shown.}\label{fig:number-embeddings}
\end{figure}
Second, for every vertex $\ell$ of $H$, all the {\sc sum}-constrained book-embeddings in which $c$ is visible and $\ell$ is the first vertex have the same extension to the left of $c$. This claim, together with the previous one, implies that the number of embeddings in $\mathcal S$ is at most $n$.  We now prove the claim; refer to Fig.~\ref{fig:number-embeddings}. Consider any vertex $\ell$ of $H$. If there is no {\sc sum}-constrained book-embedding of $H$ in which $c$ is visible and $\ell$ is the first vertex, then the claim is vacuously true. Otherwise, let $\mathcal L$ be any {\sc sum}-constrained book-embedding of $H$ in which $c$ is visible and $\ell$ is the first vertex. If $\ell=c$, then obviously we have $\lambda_{\mathcal L}(c)=0$ and there is nothing to prove. Assume hence that $\ell\neq c$. Let $T_H$ be the block-cut-vertex tree of $H$ and let $(b_1,c_1,b_2,c_2,\dots,b_{k-1},c_{k-1},b_k)$ be the shortest path in $T_H$ such that $b_1,b_2,\dots,b_k$ are B-nodes, $c_1,c_2,\dots,c_{k-1}$ are C-nodes, $\ell$ belongs to $G(b_1)$, and $c$ belongs to $G(b_k)$. For sake of simplicity, let $c_0:=\ell$ and $c_k:=c$. Since $\ell$ and $c$ are visible in $\mathcal L$, and since no two edges cross in $\mathcal L$, it follows that: (i) $\ell,c_1,c_2,\dots,c_{k-1},c$ occur in this order in $\mathcal L$; and (ii) for $j=1,\dots,k$, all the vertices of $G(b_j)$ occur between $c_{j-1}$ and $c_j$ in $\mathcal L$. By Lemma~\ref{le:linear-max-outerplanar-biconnected-characterization}, the edges $(c_0,c_1),\dots,(c_{k-1},c_k)$ belong to $H$; further, since $\ell$ is the first vertex of $\mathcal L$, since $c$ is visible in $\mathcal L$, and since no two edges cross in $\mathcal L$, it follows that none of the edges $(c_0,c_1),\dots,(c_{k-1},c_k)$ lies under another edge of $H$ in $\mathcal L$. Hence, the extension $\lambda_{\mathcal L}(c)$ of $\mathcal L$ to the left of $c$ is equal to $\sum_{j=1}^k \omega_H((c_{j-1},c_j))$. As no assumption was made on $\mathcal L$, other than $c$ is visible and $\ell$ is the first vertex, the claim and hence the lemma follow.
\end{proof}


\begin{lemma} \label{le:number-of-orderings-B}
Let $H=(V_H,E_H,\omega_H)$ be an $n$-vertex weighted outerplanar graph. Let $\mathcal S$ be a set of {\sc sum}-constrained book-embeddings of $H$ such that:

\begin{enumerate}[$(\beta 1)$]
    \item all the orderings $\mathcal L \in \mathcal S$ have the same first vertex $\ell$; and
    \item for any $\mathcal L,\mathcal L' \in \mathcal S$, we have that $\mathcal L$ does not up-down dominate and is not up-down equivalent to $\mathcal L'$.
\end{enumerate}
Then  $\mathcal S$ contains $O(n)$ embeddings.
\end{lemma}

\begin{proof}
The proof is based on two claims, very similarly to the proof of Lemma~\ref{le:number-of-orderings-B}. 

First, for any any value $\tau\geq 0$, there exists at most one {\sc sum}-constrained book-embedding $\mathcal L \in \mathcal S$ whose total extension $\tau_{\mathcal L}$ is equal to $\tau$.  Indeed, if there were two such embeddings $\mathcal L$ and $\mathcal L'$, then either one would up-down dominate the other one, or they would be up-down equivalent, depending on the values $\alpha_{\mathcal L}$ and $\alpha_{\mathcal L'}$ of their free space.

Second, for every vertex $r$ of $H$, all the {\sc sum}-constrained book-embeddings in which $\ell$ and $r$ are the first and the last vertex, respectively, have the same total extension. This claim, together with the previous one, implies that the number of embeddings in $\mathcal S$ is at most $n$ (in fact, at most $n-1$ if $n>1$, as in this case $r\neq \ell$).  We now prove the claim. Consider any vertex $r$ of $H$. If there is no {\sc sum}-constrained book-embedding of $H$ in which $\ell$ and $r$ are the first and the last vertex, respectively, then the claim is vacuously true. Otherwise, let $\mathcal L$ be any {\sc sum}-constrained book-embedding of $H$ in which $\ell$ and $r$ are the first and the last vertex, respectively. If $\ell=r$, then obviously we have $\tau_{\mathcal L}=0$ and there is nothing to prove. Assume hence that $\ell\neq r$. Let $T_H$ be the block-cut-vertex tree of $H$ and let $(b_1,c_1,b_2,c_2,\dots,b_{k-1},c_{k-1},b_k)$ be the shortest path in $T_H$ such that $b_1,b_2,\dots,b_k$ are B-nodes, $c_1,c_2,\dots,c_{k-1}$ are C-nodes, $\ell$ belongs to $G(b_1)$, and $r$ belongs to $G(b_k)$. For sake of simplicity, let $c_0:=\ell$ and $c_k:=r$. Since $\ell$ and $r$ are the first and the last vertex in $\mathcal L$, respectively, since no two edges cross in $\mathcal L$, and by Lemma~\ref{le:linear-max-outerplanar-biconnected-characterization}, it follows that the total extension $\tau_{\mathcal L}$ of $\mathcal L$ is equal to $\sum_{j=1}^k \omega_H((c_{j-1},c_j))$. As no assumption was made on $\mathcal L$, other than $\ell$ and $r$ are the first and the last vertex in $\mathcal L$, respectively, the claim and hence the lemma follow.
\end{proof}

We now describe the bottom-up visit of $T$ performed by the algorithm {\sc sum-be-drawer}.

{\bf Processing a leaf.} If $b$ is a leaf of $T$, then the sequence of {\sc sum}-constrained book-embeddings of $G^+(b)$ constructed by the algorithm {\sc sum-be-drawer} contains a single embedding $\mathcal{L}_1^+(b)=\mathcal{L}(b)$. Hence, this sequence can be computed in constant time. We have the following.

\begin{lemma} \label{le:sum-correctness-leaf}
We have that $\mathcal{L}_1^+(b)$ is a {\sc sum}-constrained book-embedding  satisfying Properties (B1)--(B3).
\end{lemma} 

\begin{proof}
Note that $\mathcal{L}_1^+(b)=\mathcal{L}(b)$ is a {\sc sum}-constrained book-embedding because {\sc sum-be-drawer} did not terminate because of Failure Condition 1. Further, $\mathcal{L}(b)$ satisfies Property (B1) because {\sc sum-be-drawer} did not terminate because of Failure Condition 2. Observe that $\mathcal{L}_1^+(b)$ vacuously satisfies Property (B2) and satisfies Property (B3) because $G(b)$ admits a unique {\sc sum}-constrained book-embedding in which the parent of $b$ is the first vertex, by Lemma~\ref{le:linear-max-outerplanar-biconnected-characterization}. 
\end{proof}

{\bf Processing a C-node.} We process a C-node $c$ as follows. Let $b_1, \dots, b_h$ be the B-nodes children of $c$. By the bottom-up visit, we assume to have, for each $b_i$ with $i=1,\dots,h$, a sequence $\mathcal{L}^+_1(b_i), \mathcal{L}^+_2(b_i), \dots, \mathcal{L}^+_{k_i}(b_i)$ of {\sc sum}-constrained book-embeddings of $G^+(b_i)$ satisfying Properties~(B1)--(B3). We relabel the B-nodes $b_1, \dots, b_h$ in such a way that $W(b_i)\leq W(b_{i+1})$, for $i=1,\dots,h-1$; this takes $O(n \log n)$ time. We now process the B-nodes $b_1, \dots, b_h$ in this order. While processing these nodes, we construct $h$ sequences $\mathcal{S}_1, \dots, \mathcal{S}_h$; the sequence $\mathcal{S}_i$ contains $O(n)$ {\sc sum}-constrained book-embeddings of $G^+(b_1)\cup\dots\cup G^+(b_i)$ satisfying Properties~$(\gamma 1)$ and $(\gamma 2)$ of Lemma~\ref{le:number-of-orderings-C}. When constructing an ordering $\mathcal{L}$ in a sequence $\mathcal{S}_i$, we also compute $\lambda_{\mathcal{L}}(c)$ and $\rho_{\mathcal{L}}(c)$. 

When processing $b_1$, we let $\mathcal{S}_1$ consist of two {\sc sum}-constrained book-embeddings, namely $\mathcal{L}^+_1(b_1)$ and its flip, in this order. Then $\mathcal{S}_1$ clearly satisfies Properties~$(\gamma 1)$ and $(\gamma 2)$ of Lemma~\ref{le:number-of-orderings-C}. Note that the extensions of $\mathcal{L}^+_1(b_1)$ to the left and to the right of $c$ are $0$ and $\tau_{\mathcal{L}^+_1(b_1)}$, respectively, while the extensions of the flip of $\mathcal{L}^+_1(b_1)$ to the left and to the right of $c$ are $\tau_{\mathcal{L}^+_1(b_1)}$ and $0$, respectively. Also note that $\mathcal{L}^+_1(b_1)$ and its flip are  {\sc sum}-constrained book-embeddings of $G^+(b_1)$ with minimum total extension. Namely, for every {\sc sum}-constrained book-embedding $\mathcal L$ of $G^+(b_1)$, by Condition (B3), there exists an index $j\in\{1,\dots,k_1\}$ such that $\mathcal{L}^+_j(b_1)$ up-down dominates or is up-down equivalent to $\mathcal L$, hence $\tau_{\mathcal{L}^+_j(b_1)}\leq \tau_{\mathcal L}$. Further, by Condition (B2), we have  $\tau_{\mathcal{L}^+_1(b_1)}\leq \tau_{\mathcal{L}^+_j(b_1)}$. 

\begin{figure}[tb]
	\centering
	\subfloat[]{\includegraphics[scale=0.55,page=1]{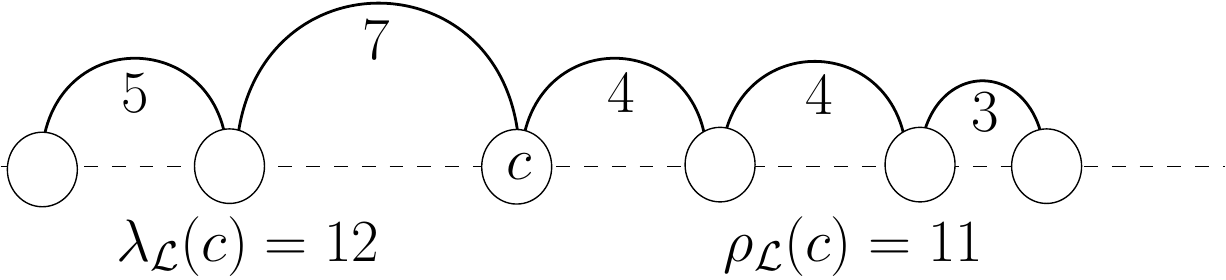}}
	\hfil
		\subfloat[]{\includegraphics[scale=0.55,page=2]{figures/SUM-insertion.pdf}}
	\hfil
		\subfloat[]{\includegraphics[scale=0.55,page=3]{figures/SUM-insertion.pdf}}
	\hfil
	\caption{(a) A {\sc sum}-constrained book-embedding $\mathcal{L}$ of $G^+(b_1)\cup\dots\cup G^+(b_{i-1})$ in $\mathcal{S}_{i-1}$. Only the edges that do not lie under any other edge are shown. (b)--(c) Combining an embedding $\mathcal{L}^+_j(b_i)$ with $\mathcal{L}$. If $\alpha_{\mathcal L^+_j(b_i)}>\rho_{\mathcal L}(c)$, then a {\sc sum}-constrained book-embedding of $G^+(b_1)\cup\dots\cup G^+(b_i)$ is constructed by placing the vertices of $\mathcal{L}^+_j(b_i) \setminus \{c\}$ to the right of $\mathcal{L}$, in the same relative order as they appear in $\mathcal{L}^+_j(b_i)$, as in (b). If $\alpha_{\mathcal L^+_j(b_i)}>\lambda_{\mathcal L}(c)$ , then a {\sc sum}-constrained book-embedding of $G^+(b_1)\cup\dots\cup G^+(b_i)$ is constructed by placing the vertices of $\mathcal{L}^+_j(b_i) \setminus \{c\}$ to the left of $\mathcal{L}$, in the opposite relative order as they appear in $\mathcal{L}^+_j(b_i)$, as in (c).}\label{fig:SUM-insertion}
\end{figure}

Suppose that, for some $i\in \{2,\dots,h\}$, the B-node $b_{i-1}$ has been processed and that the sequence $\mathcal{S}_{i-1}$ has been constructed. We process $b_i$ as follows; refer to Fig.~\ref{fig:SUM-insertion}. We initialize $\mathcal{S}_i=\emptyset$. We individually consider each of the embeddings in $\mathcal{S}_{i-1}$, say $\mathcal{L}$; since $\mathcal{S}_{i-1}$ satisfies Properties~$(\gamma 1)$ and $(\gamma 2)$ of Lemma~\ref{le:number-of-orderings-C}, there are $O(n)$  of these embeddings. We now consider each embedding $\mathcal{L}^+_j(b_i)$, with $j=1,\dots,k_i$, and we try to combine it with $\mathcal{L}$; note that, by Lemma~\ref{le:number-of-orderings-B}, we have $k_i \in O(n)$. This is done as follows. 

\begin{itemize}
	\item If $\alpha_{\mathcal L^+_j(b_i)}>\rho_{\mathcal L}(c)$, then we construct a {\sc sum}-constrained book-embedding of $G^+(b_1)\cup\dots\cup G^+(b_i)$ by placing the vertices of $\mathcal{L}^+_j(b_i) \setminus \{c\}$ to the right of $\mathcal{L}$, in the same relative order as they appear in $\mathcal{L}^+_j(b_i)$; we insert the constructed embedding into $\mathcal S_i$ and note that its extension to the left of $c$ is equal to $\lambda_{\mathcal{L}}(c)$, while its extension to the right of $c$ is equal to $\tau_{\mathcal L^+_j(b_i)}$.
	\item Symmetrically, if $\alpha_{\mathcal L^+_j(b_i)}>\lambda_{\mathcal L}(c)$, we construct a {\sc sum}-constrained book-embedding of $G^+(b_1)\cup\dots\cup G^+(b_i)$ by placing the vertices of $\mathcal{L}^+_j(b_i) \setminus \{c\}$ to the left of $\mathcal{L}$, in the opposite relative order as they appear in $\mathcal{L}^+_j(b_i)$. We insert the constructed embedding into $\mathcal S_i$ and note that its extension to the right of $c$ is equal to $\rho_{\mathcal{L}}(c)$, while its extension to the left of $c$ is equal to $\tau_{\mathcal L^+_j(b_i)}$.
\end{itemize} 

After we considered each of the $O(n)$ embeddings in $\mathcal{S}_{i-1}$, if $\mathcal{S}_{i}$ is empty, we conclude that $G$ admits no {\sc sum}-constrained book-embedding. Otherwise, we order and polish the sequence $\mathcal{S}_{i}$ by removing left-right dominated embeddings and by leaving only one copy of left-right equivalent embeddings. This is done in $O(n^2  \log n)$ time as follows. 


Since $|\mathcal{S}_{i-1}|$ and $k_i$ are both in $O(n)$, it follows that the cardinality of $\mathcal{S}_{i}$ before the polishing is $O(n^2)$. We order $\mathcal{S}_{i}$ in $O(n^2 \log n)$ time primarily based on the value of the left extension with respect to $c$ and secondarily based on the value of the right extension with respect to $c$. Then we scan $\mathcal S_i$; during the scan, we process the elements of $\mathcal S_i$ one by one. 

When we process an element $\mathcal L$, we compare it with its predecessor $\mathcal L'$. Note that, because of the ordering, we have $\lambda_{\mathcal L'}(c)\leq \lambda_{\mathcal L}(c)$. If $\rho_{\mathcal L'}(c)\leq \rho_{\mathcal L}(c)$, then we remove $\mathcal L$  from $\mathcal S_i$. Note that this scan takes $O(n^2)$ time.

This concludes the description of the processing of $b_i$ and the consequent construction of the sequence $\mathcal S_i$. As described, this processing takes $O(n^2\log n)$ time, and hence $O(h n^2\log n)$ time over all the B-nodes that are children of $c$. After processing the last B-node $b_h$, the sequence $\mathcal S_h$ contains the required {\sc sum}-constrained book-embeddings of $G^+(c)$ satisfying Properties~(C1)--(C3), as proved in the following.

\begin{lemma}\label{le:correctness-C}
We have that $\mathcal S_h$ is a (possibly empty) sequence $\mathcal L^+_1(c), \dots, \mathcal L^+_k(c)$ of {\sc sum}-constrained book-embeddings of $G^+(c)$ satisfying  Properties~(C1)--(C3).
\end{lemma}


\begin{proof}
We show that every embedding of $G^+(c)$ in $\mathcal S_h$ is a  {\sc sum}-constrained book-embedding satisfying Property (C1);  namely, we prove, by induction on $i$, that every embedding of $G^+(b_1)\cup \cdots \cup G^+(b_i)$ in $\mathcal S_i$ is a {\sc sum}-constrained book-embedding such that $c$ is visible. 

In the base case, we have $i=1$. Then $\mathcal S_1$ contains $\mathcal{L}^+_1(b_1)$ and its flip. These two embeddings are {\sc sum}-constrained book-embeddings such that $c$ is visible, by definition and since $\mathcal{L}^+_1(b_1)$ satisfies Property~(B1),


Now inductively assume that, for some $i\in \{2,\dots,h\}$, every embedding of $G^+(b_1)\cup \cdots \cup G^+(b_{i-1})$ in $\mathcal S_{i-1}$ is a {\sc sum}-constrained book-embedding such that $c$ is visible. Every embedding $\mathcal L^*$ we insert into $\mathcal S_i$ is constructed from an embedding $\mathcal L$ in $\mathcal S_{i-1}$ and an embedding $\mathcal{L}^+_j(b_i)$ of $G^+(b_{i})$ taken from the sequence $\mathcal{L}^+_1(b_i), \mathcal{L}^+_2(b_i), \dots, \mathcal{L}^+_{k_i}(b_i)$. Indeed, $\mathcal L^*$ is either constructed by placing the vertices of $\mathcal{L}^+_j(b_i) \setminus \{c\}$ to the right of $\mathcal L$, in the same relative order as they appear in $\mathcal{L}^+_j(b_i)$, or is constructed by placing the vertices of $\mathcal{L}^+_j(b_i) \setminus \{c\}$ to the left of $\mathcal L$, in the opposite relative order as they appear in $\mathcal{L}^+_j(b_i)$. In both cases, $c$ is visible in the resulting embedding. 
Further, $\mathcal L^*$ is a {\sc sum}-constrained book-embedding. Namely, assume that the vertices of $\mathcal{L}^+_j(b_i) \setminus \{c\}$ are placed to the right of $\mathcal L$ in $\mathcal L^*$, the other case is analogous. Then $\mathcal L^*$ is a {\sc sum}-constrained book-embedding given that $\mathcal L$ and $\mathcal{L}^+_j(b_i)$ are {\sc sum}-constrained book-embeddings and given that the free space of $\mathcal{L}^+_j(b_i)$ is larger than the extension of $\mathcal L$ to the right of $c$, by construction.

Concerning Property~(C2), let $\mathcal L^+_p(c)$ and $\mathcal L^+_q(c)$ be any two embeddings in $\mathcal S_h$ such that $p<q$. By the ordering of $\mathcal S_h$, we have that $\mathcal L^+_q(c)$ does not left-right dominate $\mathcal L^+_p(c)$ with respect to $c$. Suppose, for a contradiction, that:

\begin{enumerate}[(i)]
    \item $\mathcal L^+_p(c)$ left-right dominates or is left-right equivalent to $\mathcal L^+_q(c)$ with respect to $c$; that is $\lambda_{\mathcal L^+_p(c)}(c)\leq \lambda_{\mathcal L^+_q(c)}(c)$ and $\rho_{\mathcal L^+_p(c)}(c)\leq \rho_{\mathcal L^+_q(c)}(c)$; and
    \item there are no two embeddings $\mathcal L^+_r(c)$ and $\mathcal L^+_s(c)$ with $r<s$ such that $\mathcal L^+_r(c)$ left-right dominates or is left-right equivalent to $\mathcal L^+_s(c)$, and such that $s-r<q-p$; that is, $\mathcal L^+_p(c)$ and $\mathcal L^+_q(c)$ are the ``closest'' embeddings in $\mathcal S_h$ such that $\mathcal L^+_p(c)$ left-right dominates or is left-right equivalent to $\mathcal L^+_q(c)$.
\end{enumerate} 
If $q-p=1$ (that is, $\mathcal L^+_p(c)$ and $\mathcal L^+_q(c)$ are consecutive in $\mathcal S_h$), then we would have removed $\mathcal L^+_q(c)$ from $\mathcal S_h$ during its processing, a contradiction. If $q-p>1$, then consider any ordering $\mathcal L^+_x(c)$ that appears between $\mathcal L^+_p(c)$ and $\mathcal L^+_q(c)$ in $\mathcal S_h$. Because of the ordering of the embeddings in $\mathcal S_h$, we have $\lambda_{\mathcal L^+_p(c)}(c)\leq \lambda_{\mathcal L^+_x(c)}(c)\leq \lambda_{\mathcal L^+_q(c)}(c)$. Since $\mathcal L^+_p(c)$ left-right dominates or is left-right equivalent to $\mathcal L$, we have that $\rho_{\mathcal L^+_p(c)}(c)\leq \rho_{\mathcal L^+_q(c)}(c)$. If $\rho_{\mathcal L^+_x(c)}(c)\geq \rho_{\mathcal L^+_p(c)}(c)$, then $\mathcal L^+_p(c)$ left-right dominates or is left-right equivalent to $\mathcal L^+_x(c)$ with respect to $c$, contradicting the minimality of $q-p$. Otherwise, $\rho_{\mathcal L^+_x(c)}(c)<\rho_{\mathcal L^+_p(c)}(c)$, which implies that $\rho_{\mathcal L^+_x(c)}(c)<\rho_{\mathcal L^+_q(c)}(c)$, hence $\mathcal L^+_x(c)$ left-right dominates $\mathcal L^+_q(c)$ with respect to $c$, again contradicting the minimality of $q-p$. This contradiction proves that no embedding in $\mathcal S_h$ left-right dominates or is left-right equivalent to a distinct embedding in $\mathcal S_h$ with respect to $c$. Hence, no two embeddings have the same extension to the left or to the right of $c$. By the ordering of the embeddings in   $\mathcal S_h$, we have $\lambda_{\mathcal L^+_1(c)}(c)<\dots<\lambda_{\mathcal L^+_k(c)}(c)$ and $\rho_{\mathcal L^+_1(c)}(c)>\dots>\rho_{\mathcal L^+_k(c)}(c)$. Property~(C2) follows.

Finally, we prove that $\mathcal S_h$ satisfies Property~(C3). Suppose, for a contradiction, that there exists a {\sc sum}-constrained book-embedding $\mathcal L^{\diamond}$ of $G^+(c)$ satisfying Property~(C1) and such that no embedding in $\mathcal S_h$ left-right dominates or is left-right equivalent to $\mathcal L^{\diamond}$ with respect to $c$.
For $i=1,\dots,h$, let $\mathcal L^{\diamond}_i$ be the restriction of $\mathcal L^{\diamond}$ to the vertices and edges of $G^+(b_1)\cup \cdots \cup G^+(b_i)$; note that $\mathcal L^{\diamond}_h=\mathcal L^{\diamond}$. We prove, by induction on $i$, the following statement, which contradicts the above supposition: There exists a {\sc sum}-constrained book-embedding $\mathcal L^*_i$ in $\mathcal S_i$ which left-right dominates or is left-right equivalent to $\mathcal L^{\diamond}_i$ with respect to $c$.

In the base case, we have $i=1$. Then since $\mathcal{L}^+_1(b_1), \mathcal{L}^+_2(b_1), \dots, \mathcal{L}^+_{k_1}(b_1)$ satisfy Property~(B3), there exists an index $j\in \{1,\dots,k_1\}$ such that $\mathcal{L}^+_j(b_1)$ up-down dominates or is up-down equivalent to $\mathcal L^{\diamond}_1$, hence the total extension of $\mathcal{L}^+_j(b_1)$ is smaller than or equal to the total extension of $\mathcal L^{\diamond}_1$. By Property~(B2), we have that the total extension of $\mathcal{L}^+_j(b_1)$ is larger than or equal to the total extension of $\mathcal{L}^+_1(b_1)$ (where equality holds only if $j=1$). Hence, the total extension of $\mathcal{L}^+_1(b_1)$ is smaller than or equal to the total extension of $\mathcal L^{\diamond}_1$. Since $\mathcal L^{\diamond}_1$ satisfies Property~(C1), we have that either all the vertices of $\mathcal L^{\diamond}_1 \setminus\{c\}$ are to the right of $c$, or they all are to the left of $c$; then, respectively, either $\mathcal{L}^+_1(b_1)$ or its flip left-right dominates or is left-right equivalent to $\mathcal L^{\diamond}_1$ with respect to $c$. Since both $\mathcal{L}^+_1(b_1)$ and its flip are in $\mathcal S_1$, the base case of the statement follows.

Now inductively assume that, for some $i\in \{2,\dots,h\}$, there exists a {\sc sum}-constrained book-embedding $\mathcal L^*_{i-1}$ in $\mathcal S_{i-1}$ which left-right dominates or is left-right equivalent to $\mathcal L^{\diamond}_{i-1}$ with respect to $c$. 

We construct a {\sc sum}-constrained book-embedding which left-right dominates or is left-right equivalent to $\mathcal L^{\diamond}_i$  with respect to $c$ and such that it belongs to~$\mathcal S_i$. 

Let $\mathcal L^{\diamond}(b_i)$ be the restriction of $\mathcal L^{\diamond}_i$ to the vertices and edges of $G^+(b_i)$. Since $\mathcal{L}^+_1(b_i), \mathcal{L}^+_2(b_i), \dots, \mathcal{L}^+_{k_i}(b_i)$ satisfy Property~(B3), there exists an index $j\in\{1,\dots,k_i\}$ such that $\mathcal{L}^+_j(b_i)$ up-down dominates or is up-down equivalent to $\mathcal L^{\diamond}(b_i)$ (or its flip). Since $W(b_1)<\dots<W(b_i)$, it follows that $G^+(b_i)$ does not lie under any edge of $G^+(b_1)\cup \cdots \cup G^+(b_{i-1})$ in $\mathcal L^{\diamond}_i$. Further, since $\mathcal L^{\diamond}_i$ satisfies Property~(C1), it follows that either all the vertices of $G^+(b_i)\setminus \{c\}$ lie to the right of $c$ in $\mathcal L^{\diamond}_i$, or they all lie to the left of $c$; suppose that we are in the former case, as the discussion for the latter case is analogous. 

Let $\mathcal L^*_i$ be the embedding obtained by placing the vertices of $\mathcal{L}^+_j(b_i) \setminus \{c\}$ to the right of $\mathcal L^*_{i-1}$, in the same relative order as they appear in $\mathcal{L}^+_j(b_i)$. 
Then $\lambda_{\mathcal L^*_i}(c)=\lambda_{\mathcal L^*_{i-1}}(c)\leq \lambda_{\mathcal L^{\diamond}_{i-1}}(c)=\lambda_{\mathcal L^{\diamond}_{i}}(c)$, where the inequality exploits the inductive hypothesis. Further, $\rho_{\mathcal L^*_i}(c)$ coincides with the total extension of $\mathcal{L}^+_j(b_i)$, which is smaller than or equal to the total extension of $\mathcal L^{\diamond}(b_i)$, given that $\mathcal{L}^+_j(b_i)$ up-down dominates or is up-down equivalent to $\mathcal L^{\diamond}(b_i)$; hence, $\rho_{\mathcal L^*_i}(c)\leq \rho_{\mathcal L^{\diamond}_i}(c)$. This proves that $\mathcal L^*_i$ left-right dominates or is left-right equivalent to $\mathcal L^{\diamond}_i$.

Finally, we prove that $\mathcal S_i$ (before the polishing) contains $\mathcal L^*_i$. By induction, $\mathcal S_{i-1}$ contains $\mathcal L^*_{i-1}$. Hence, by construction, $\mathcal S_i$ contains $\mathcal L^*_i$ as long as $\alpha_{\mathcal{L}^+_j(b_i)}>\rho_{\mathcal L^*_{i-1}}(c)$. We prove that this is indeed the case. First, since $\mathcal{L}^+_j(b_i)$ up-down dominates $\mathcal L^{\diamond}(b_i)$, we have that $\alpha_{\mathcal{L}^+_j(b_i)}\geq \alpha_{\mathcal L^{\diamond}(b_i)}$. Second, since $\mathcal L^{\diamond}_i$ is a {\sc sum}-constrained book-embedding, we have that $\alpha_{\mathcal L^{\diamond}(b_i)}>\rho_{\mathcal L^{\diamond}_{i-1}}(c)$. Finally, since $\mathcal L^*_{i-1}$ left-right dominates or is left-right equivalent to $\mathcal L^{\diamond}_{i-1}$ with respect to $c$, we have that $\rho_{\mathcal L^{\diamond}_{i-1}}(c)\geq \rho_{\mathcal L^*_{i-1}}(c)$. The three inequalities imply that $\alpha_{\mathcal{L}^+_j(b_i)}>\rho_{\mathcal L^*_{i-1}}(c)$.

Since, before the polishing, $\mathcal S_i$ contains $\mathcal L^*_i$, after the polishing it contains either $\mathcal L^*_i$ or a different {\sc sum}-constrained book-embedding of $G^+(b_1)\cup \cdots \cup G^+(b_i)$ which left-right dominates or is left-right equivalent to $\mathcal L^*_i$ with respect to $c$; indeed, $\mathcal L^*_i$ is removed from $\mathcal S_i$ only if it is compared with such an embedding. In both cases, $\mathcal S_i$ contains a {\sc sum}-constrained book-embedding of $G^+(b_1)\cup \cdots \cup G^+(b_i)$ which left-right dominates or is left-right equivalent to $\mathcal L^{\diamond}_i$. This concludes the induction and hence the proof of the lemma.
\end{proof}

{\bf Processing an internal B-node different from the root.} We now describe how to process an internal B-node $b\neq b^*$ of $T$. The goal is  either to conclude that $G^+(b)$ does not admit a {\sc sum}-constrained book-embedding satisfying Property (B1), which by Lemma~\ref{le:necessity-b1} implies that $G$ does not admit any {\sc sum}-constrained book-embedding, or to construct a sequence $\mathcal{L}^+_1(b), \mathcal{L}^+_2(b), \dots, \mathcal{L}^+_{k(b)}(b)$ of {\sc sum}-constrained book-embeddings satisfying Properties (B1)--(B3).

First, if the algorithm {\sc sum-be-drawer} did not terminate because of Failure Conditions~1--2, we have a {\sc sum}-constrained book-embedding $\mathcal L(b)=(v_0,v_1,\dots,v_k)$ of $G(b)$ in which the parent $c$ of $b$ in $T$ is the first vertex, that is, $v_0=c$.  Further, let $c_1, \dots, c_h$ be the C-nodes that are children of $c$, labeled in the same order as they appear in $\mathcal L(b)$. Since the algorithm {\sc sum-be-drawer} did not terminate when visiting $c_1, \dots, c_h$, we have, for each $c_i$ with $i=1,\dots,h$, a sequence $\mathcal{L}^+_1(c_i), \mathcal{L}^+_2(c_i), \dots, \mathcal{L}^+_{k_i}(c_i)$ of {\sc sum}-constrained book-embeddings of $G^+(c_i)$ satisfying Properties~(C1)--(C3). 

Observe that some vertices $v_i$ might not be in $\{c,c_1,\dots,c_h\}$. Specifically, we distinguish the case in which $v_1=c_1$ from the one in which $v_1\neq c_1$.

Suppose first that $v_1\neq c_1$ is not a cut-vertex of $G^+(b)$. In this case, if the algorithm {\sc sum-be-drawer} constructs a sequence $\mathcal{L}^+_1(b), \mathcal{L}^+_2(b), \dots, \mathcal{L}^+_{k(b)}(b)$ of {\sc sum}-constrained book-embeddings satisfying Properties (B1)--(B3), that is, if it does not conclude then $k(b)=1$, that is, the sequence contains a single embedding.
The idea is to process the C-nodes $c_1,\dots,c_h$ in this order and, for each C-node $c_i$, to choose a {\sc sum}-constrained book-embedding $\mathcal{L}^+_{j}(c_i)$ for $G^+(c_i)$ in such a way that the extension of $\mathcal{L}^+_{j}(c_i)$ to the right of $c_i$ is minimum. However, by Property~(C2), the smaller the extension of $\mathcal{L}^+_{j}(c_i)$ to the right of $c_i$, the larger the extension of $\mathcal{L}^+_{j}(c_i)$ to the left of $c_i$. Hence, we need to select $\mathcal{L}^+_{j}(c_i)$ so that its extension to the right of $c_i$ is minimum, subject to the constraint that it ``fits'' on the left. We formalize this idea as follows.

We process the C-nodes $c_1, \dots, c_h$ in this order. Before any C-node is processed, we initialize $\mathcal L^*_0:= \mathcal L(b)$ and, for $i=1,\dots,k$, we initialize a variable $\ell(v_i)$ to the weight of the edge $(v_{i-1},v_i)$; roughly speaking, throughout the embedding construction, $\ell(v_i)$ represents the amount of ``remaining free space'' to the left of~$v_i$. 

When we process $c_i$, we construct a {\sc sum}-constrained book-embedding $\mathcal L^*_i$ of $G(b)\cup G^+(c_1)\cup \dots \cup G^+(c_i)$. This is done by choosing a {\sc sum}-constrained book-embedding $\mathcal{L}^+_{j}(c_i)$ for $G^+(c_i)$ and by replacing $c_i$ with $\mathcal{L}^+_{j}(c_i)$ in $\mathcal L^*_{i-1}$. The choice of $\mathcal{L}^+_{j}(c_i)$ is performed as follows. Let $x$ be such that $c_i=v_x$. Then we let $\mathcal{L}^+_{j}(c_i)$ be the embedding such that: 

\begin{enumerate}[(i)]
\item $\lambda_{\mathcal{L}^+_{j}(c_i)}<\ell(v_x)$, that is, $\mathcal{L}^+_{j}(c_i)$ fits to the left of $v_x$; and
\item $\lambda_{\mathcal{L}^+_{j}(c_i)}$ is maximum, among all the embeddings in $\mathcal{L}^+_{1}(c_i),\dots,\mathcal{L}^+_{k_i}(c_i)$ that satisfy constraint (i).
\end{enumerate}

If no such embedding exists, then we conclude that $G$ admits no {\sc sum}-constrained book-embedding. Otherwise, if $x<k$, we check whether $\rho_{\mathcal{L}^+_{j}(c_i)}<\ell(v_{x+1})$. In the negative case, that is, if $\mathcal{L}^+_{j}(c_i)$ does not fit to the right of $v_x$, then we conclude that $G$ admits no {\sc sum}-constrained book-embedding. In the positive case, we constructed $\mathcal L^*_i$; then we decrease $\ell(v_{x+1})$ by $\rho_{\mathcal{L}^+_{j}(c_i)}$, as the remaining free space to the left of $v_{x+1}$ decreased  by $\rho_{\mathcal{L}^+_{j}(c_i)}$ when replacing $c_i$ with $\mathcal{L}^+_{j}(c_i)$, and proceed. If $n_i$ denotes the number of vertices in $G^+(c_i)$, by Lemma~\ref{le:number-of-orderings-C} we have $O(n_i)$ embeddings for $G^+(c_i)$, hence $c_i$ is processed in $O(n_i)$ time and then the C-nodes $c_1, \dots, c_h$ are processed in total $O(n)$ time. 

Suppose next that $v_1=c_1$. In this case, it might be possible that the algorithm {\sc sum-be-drawer} constructs a sequence $\mathcal{L}^+_1(b), \mathcal{L}^+_2(b), \dots, \mathcal{L}^+_{k(b)}(b)$ of {\sc sum}-constrained book-embeddings satisfying Properties (B1)--(B3) with $k(b)>1$. Differently from the case in which $v_1\neq c_1$,  we cannot perform an ``optimal'' choice for the embedding of $G^+(c_1)$. Namely, on one hand we would like to select an embedding of $G^+(c_1)$ among $\mathcal{L}^+_{1}(c_1),\dots,\mathcal{L}^+_{k_1}(c_1)$ that ``consumes'' as little space as possible to the left of $c_1$, so that the free space $\alpha_{\mathcal L}$ of the {\sc sum}-constrained book-embedding $\mathcal L$ of $G^+(b)$ we are constructing is large. On the other hand, we would like to select an embedding of $G^+(c_1)$ among $\mathcal{L}^+_{1}(c_1),\dots,\mathcal{L}^+_{k_1}(c_1)$ that ``consumes'' as little space as possible to the right of $c_1$, in order to leave room for an embedding of $G^+(c_2)$.  These two objectives are in contrast, by Property~(C2) of the sequence $\mathcal{L}^+_{1}(c_1),\dots,\mathcal{L}^+_{k_1}(c_1)$. Hence, we will consider all the $O(n)$ possible choices for the embedding of $G^+(c_1)$. For each of these choices, we process the C-nodes $c_2,\dots,c_h$ in this order, similarly to the case in which $v_1\neq c_1$. Namely, for each C-node $c_i$ with $i\geq 2$, we choose a {\sc sum}-constrained book-embedding $\mathcal{L}^+_{j}(c_i)$ for $G^+(c_i)$ in such a way that the extension of $\mathcal{L}^+_{j}(c_i)$ to the right of $c_i$ is minimum subject to the constraint that $\mathcal{L}^+_{j}(c_i)$ ``fits'' on the left. We formalize this idea as follows.

We initialize $\mathcal L^*_{1,0}:=\mathcal L^*_{2,0}:=\dots:=\mathcal L^*_{k_1,0}:= \mathcal L(b)$. Recall that $k_1$ is the number of embeddings $\mathcal{L}^+_1(c_1), \dots, \mathcal{L}^+_{k_1}(c_1)$ of $G^+(c_1)$. Note that $k_1 \in O(n)$, by Lemma~\ref{le:number-of-orderings-C}.

Starting from each embedding $\mathcal L^*_{j,0}$, we will try to construct a {\sc sum}-constrained book-embedding $\mathcal L^*_{j,h}$ of $G^+(b)$. 
For each $j=1,\dots,k_1$, we process the C-nodes $c_1, \dots, c_h$ in this order. When we process $c_i$, we possibly construct a {\sc sum}-constrained book-embedding $\mathcal L^*_{j,i}$ of $G(b)\cup G^+(c_1)\cup \dots \cup G^+(c_i)$. 
Before any C-node is processed, for $j=1,\dots,k_1$ and for $i=1,\dots,k$, we initialize a variable $\ell_j(v_i)$ to the weight of the edge $(v_{i-1},v_i)$, similarly to the case $v_1\neq c_1$.

For $j=1,\dots,k_1$, we start by processing $c_1$. Namely, we check whether $\lambda_{\mathcal{L}^+_{j}(c_1)}\geq \ell_j(v_1)$, that is, whether $\mathcal{L}^+_{j}(c_1)$ does not fit to the left of $v_1$; in the positive case, we discard the embedding $\mathcal L^*_{j,0}$ and proceed. Further, we check whether $\rho_{\mathcal{L}^+_{j}(c_1)}\geq \ell_j(v_2)$, that is, whether $\mathcal{L}^+_{j}(c_1)$ does not fit to the right of $v_1$; in the positive case, we discard the embedding $\mathcal L^*_{j,0}$ and proceed. If both checks fail, then we replace $c_1$ with $\mathcal{L}^+_{j}(c_1)$, thus constructing a {\sc sum}-constrained book-embedding $\mathcal L^*_{j,1}$ of $G(b)\cup G^+(c_1)$; further, we decrease $\ell_j(v_2)$ by $\rho_{\mathcal{L}^+_{j}(c_1)}$.  

Now, for $j=1,\dots,k_1$ and for $i=2,\dots,h$, when we process $c_i$, we construct a {\sc sum}-constrained book-embedding $\mathcal L^*_{j,i}$ of $G(b)\cup G^+(c_1)\cup \dots \cup G^+(c_i)$. This is done by choosing a {\sc sum}-constrained book-embedding $\mathcal{L}^+_{m}(c_i)$ for $G^+(c_i)$ and by replacing $c_i$ with $\mathcal{L}^+_{m}(c_i)$ in $\mathcal L^*_{j,i-1}$. The choice of $\mathcal{L}^+_{m}(c_i)$ is performed as in the case in which $v_1=c_1$. Namely, let $x$ be such that $c_i=v_x$. Then we let $\mathcal{L}^+_{m}(c_i)$ be the embedding such that: 

\begin{enumerate} [(i)]
    \item $\lambda_{\mathcal{L}^+_{m}(c_i)}<\ell_j(v_x)$; and 
    \item $\lambda_{\mathcal{L}^+_{m}(c_i)}$ is maximum, among all the embeddings in $\mathcal{L}^+_{1}(c_i),\dots,\mathcal{L}^+_{k_i}(c_i)$ that satisfy constraint (i).
\end{enumerate} 
If no such embedding exists, then we discard the embedding $\mathcal{L}^*_{j,0}$  and proceed. Otherwise, if $x<k$, we check whether $\rho_{\mathcal{L}^+_{m}(c_i)}<\ell_j(v_{x+1})$. In the negative case, we discard the embedding $\mathcal{L}^*_{j,0}$ and proceed. In the positive case, we constructed $\mathcal L^*_{j,i}$; then we decrease $\ell_j(v_{x+1})$ by $\rho_{\mathcal{L}^+_{m}(c_i)}(c_i)$ and proceed. 

If the above algorithm did not construct any embedding $\mathcal{L}^*_{j,h}$ of $G^+(b)$, then we $G$ admits no {\sc sum}-constrained book-embedding. Otherwise, we have at most $k_1\in O(n)$  embeddings $\mathcal{L}^*_{1,h},\dots,\mathcal{L}^*_{k_1,h}$ of $G^+(b)$.

We discuss the time complexity of the algorithm. For each of the $O(n)$ embeddings $\mathcal L^*_{j,0}$ of $G(b)$, we select a  single embedding $\mathcal L^+_j(c_1)$ for $G^+(c_1)$ and, for every $i=2,\dots,h$, we select a single embedding $\mathcal L^+_m(c_i)$ for $G^+(c_i)$  by choosing it among $O(n_i)$ embeddings, where $n_i$ denotes the number of vertices in $G^+(c_i)$. Thus, the algorithm takes $O(n)$ time for each of the $O(n)$ embeddings $\mathcal L^*_{j,0}$ of $G(b)$, and thus $O(n^2)$ time in total.

Denote by $\mathcal S$ the sequence of constructed embeddings. We polish $\mathcal S$ so that no embedding up-down dominates or is up-down equivalent to another embedding in the sequence. This could be done in $O(n \log n)$ time by following the same approach employed when dealing with C-nodes. However, this can actually  be done easily in $O(n)$ time in this case, as the embeddings of $G^+(b)$ have been constructed in decreasing order of free space. Hence, it suffices to check whether each embedding $\mathcal L$ in $\mathcal S$ is up-down dominated or is up-down equivalent to the embedding preceding it; in the positive case, $\mathcal L$ can be removed from $\mathcal S$. Finally, $\mathcal S$ is inverted so that the embeddings appear in increasing order of free space.

This concludes the description of the algorithm for an internal B-node different from the root. 

\begin{lemma}\label{le:correctness-B}
We have that $\mathcal S$ is a (possibly empty) sequence $\mathcal L^+_1(b), \dots, \mathcal L^+_{k(b)}(b)$ of {\sc sum}-constrained book-embeddings of $G^+(b)$ satisfying  Properties~(B1)--(B3).
\end{lemma}

\begin{proof}
First, we show that every embedding $\mathcal L^+_j(b)\in \mathcal S$ of $G^+(b)$  is a  {\sc sum}-constrained book-embedding satisfying Property (B1). Namely, $\mathcal L^+_j(b)$ is constructed starting from a {\sc sum}-constrained book-embedding $\mathcal L(b)=(v_0,v_1,\dots,v_k)$ of $G(b)$ in which $c=v_0$ and by then replacing, for $i=1,\dots,h$, the vertex $c_i$ with a {\sc sum}-constrained book-embedding of $G^+(c_i)$; since $v_0\notin \{c_1,\dots,c_h\}$, we have that $\mathcal L^+_j(b)$ satisfies Property~(B1).  
We denote by $\mathcal L^+_{f(j,i)}(c_i)$ the {\sc sum}-constrained book-embedding of $G^+(c_i)$ that replaces $c_i$ in $\mathcal L^+_j(b)$. With a slight abuse of notation, we also denote by $\ell_j(v_1),\dots,\ell_j(v_k)$ the variables used in the construction of $\mathcal L^+_j(b)$.

Since $\mathcal L^+_{f(j,1)}(c_1),\dots,\mathcal L^+_{f(j,h)}(c_h)$ are {\sc sum}-constrained book-embeddings, in order to prove that $\mathcal L^+_j(b)$ is a {\sc sum}-constrained book-embedding, it suffices to prove that, for $x=0,\dots,{k-1}$, the weight of the edge $(v_x,v_{x+1})$ of $G(b)$ is larger than the sum of: 
\begin{enumerate}[(i)]
    \item the extension $\rho_{\mathcal L^+_{f(j,p)}(c_p)}(c_p)$ of $\mathcal L^+_{f(j,p)}(c_p)$ to the right of $c_p$, if $v_x=c_p$ (or $0$ if $v_x$ is not a cut-vertex of $G^+(b)$); and  
    \item the extension $\lambda_{\mathcal L^+_{f(j,q)}(c_{q})}(c_{q})$ of $\mathcal L^+_{f(j,q)}(c_q)$ to the left of $c_q$, if $v_{x+1}=c_q$ (or $0$ if $v_{x+1}$ is not a cut-vertex of $G^+(b)$).  
\end{enumerate}
Assume that $v_x=c_p$ and that $v_{x+1}=c_{p+1}$; the case in which at most one of $v_x$ and $v_{x+1}$ is a cut-vertex of $G^+(b)$ is analogous and simpler. Recall that the value $\ell_j(v_{x+1})$ is initialized to the weight of the edge $(v_x,v_{x+1})$. By construction, when $v_x=c_p$ is replaced by $\mathcal L^+_{f(j,p)}(c_p)$ we have $\rho_{\mathcal L^+_{f(j,p)}(c_p)}(c_p)<\ell_j(v_{x+1})$; further, when such a replacement is performed, the value of $\ell_j(v_{x+1})$ is decreased by  $\rho_{\mathcal L^+_{f(j,p)}(c_p)}(c_p)$. Further, when $v_{x+1}=c_{p+1}$ is replaced by $\mathcal L^+_{f(j,p+1)}(c_{p+1})$ we have $\lambda_{\mathcal L^+_{f(j,p+1)}(c_{p+1})}(c_{p+1})<\ell_j(v_{x+1})$. This implies that the weight of the edge $(v_x,v_{x+1})$ is larger than $\rho_{\mathcal L^+_{f(j,p)}(c_p)}(c_{p})+\lambda_{\mathcal L^+_{f(j,p+1)}(c_{p+1})}(c_{p+1})$. 

Property~(B2) is trivially satisfied if $v_1\neq c_1$, as in this case $\mathcal S$ contains a single {\sc sum}-constrained book-embedding; further, it is directly ensured by the final ordering and polishing that is performed on the sequence $\mathcal S$, in the case in which $v_1= c_1$.

Finally, we prove that $\mathcal S$ satisfies Property~(B3). Suppose, for a contradiction, that there exists a {\sc sum}-constrained book-embedding $\mathcal L^{\diamond}$ of $G^+(b)$ satisfying Property~(B1) and such that no embedding in $\mathcal S$ up-down dominates or is up-down equivalent to $\mathcal L^{\diamond}$. Let $\mathcal L^{\diamond}_0$ be the restriction of $\mathcal L^{\diamond}$ to $G(b)$; further, for $i=1,\dots,h$, let $\mathcal L^{\diamond}_i(c_i)$ be the restriction of $\mathcal L^{\diamond}$ to $G^+(c_i)$ and let $\mathcal L^{\diamond}_i$ be the restriction of $\mathcal L^{\diamond}$ to $G(b)\cup G^+(c_1)\cup \cdots \cup G^+(c_i)$; note that $\mathcal L^{\diamond}_h=\mathcal L^{\diamond}$. Finally, for $i=1,\dots,h$, let $x(i)$ be such that $v_{x(i)}=c_i$. Throughout this proof, we assume that $v_1=c_1$. The case in which $v_1\neq c_1$ is analogous and simpler.

We prove, by induction on $i$, the following statement: The algorithm {\sc sum-be-drawer} constructs (and does not discard) a {\sc sum}-constrained book-embedding $\mathcal L^*_{j,i}$ of $G(b)\cup G^+(c_1)\cup \dots \cup G^+(c_i)$ such that: 

\begin{enumerate}[(1)]
\item $\mathcal L^*_{j,i}$ up-down dominates or is up-down equivalent to $\mathcal L^{\diamond}_i$; and
\item let $\mathcal L^*_{j,i}(c_i)$ be the restriction of $\mathcal L^*_{j,i}$ to $G^+(c_i)$; if $i<h$ and $x(i+1)=x(i)+1$ (that is, if the cut-vertices $c_i$ and $c_{i+1}$ are consecutive in $\mathcal L(b)$), then the extension $\rho_{\mathcal L^*_{j,i}(c_i)}(c_i)$ of $\mathcal L^*_{j,i}(c_i)$ to the right of $c_i$ is smaller than or equal to the extension  $\rho_{\mathcal L^{\diamond}_i(c_i)}(c_i)$ of $\mathcal L^{\diamond}_i(c_i)$ to the right of $c_i$; roughly speaking, this ensures that the  ``remaining free space'' to the left of~$v_{x(i+1)}$ in $\mathcal L^*_{j,i}$ is at least as much as the one in $\mathcal L^{\diamond}_i$.
\end{enumerate}

By construction, the algorithm {\sc sum-be-drawer} constructs (and does not discard) $k_1$ {\sc sum}-constrained book-embeddings $\mathcal L^*_{1,0},\dots,\mathcal L^*_{k_1,0}$; the restriction of each of such embeddings to $G(b)$ is $\mathcal L(b)$. Further, $L^{\diamond}_0$ also coincides with $\mathcal L(b)$, by Lemma~\ref{le:linear-max-outerplanar-biconnected-characterization} and by the assumption that $\mathcal L^{\diamond}$ satisfies Property~(B1). This ensures that each of $\mathcal L^*_{1,0},\dots,\mathcal L^*_{k_1,0}$ is up-down equivalent to $\mathcal L^{\diamond}_0$.

We now prove the induction. In the base case, we have $i=1$. Since $\mathcal L^+_1(c_1),\dots,\mathcal L^+_{k_1}(c_1)$ satisfy Properties (C1)--(C3), there exists a {\sc sum}-constrained book-embedding $\mathcal L^+_j(c_1)$ that left-right dominates or is left-right equivalent to $\mathcal L^{\diamond}_1(c_1)$ with respect to $c_1$; that is, $\lambda_{\mathcal L^+_j(c_1)}(c_1)\leq \lambda_{\mathcal L^{\diamond}_{1}(c_{1})}(c_{1})$ and $\rho_{\mathcal L^+_j(c_1)}(c_1)\leq \rho_{\mathcal L^{\diamond}_{1}(c_{1})}(c_{1})$. Since $\mathcal L^{\diamond}_1$ is a {\sc sum}-constrained book-embedding, the weight of the edge $(v_{x(1)-1},v_{x(1)})$ is larger than $\lambda_{\mathcal L^{\diamond}_{1}(c_{1})}(c_{1})$, hence it is larger than $\lambda_{\mathcal L^+_j(c_1)}(c_1)$, and the weight of the edge $(v_{x(1)},v_{x(1)+1})$ is larger than $\rho_{\mathcal L^{\diamond}_{1}(c_{1})}(c_{1})$, hence it is larger than $\rho_{\mathcal L^+_j(c_1)}(c_1)$. It follows that the algorithm {\sc sum-be-drawer} constructs (and does not discard) a {\sc sum}-constrained book-embedding $\mathcal L^*_{j,1}$ of $G(b)\cup G^+(c_1)$ by replacing $c_1$ with $\mathcal L^+_j(c_1)$ in $\mathcal L^*_{j,0}$. 

We prove that $\mathcal L^*_{j,1}$ satisfies Condition~(1).

\begin{itemize}
    \item If $x(1)>1$, then the free spaces of $\mathcal L^*_{j,1}$ and $\mathcal L^{\diamond}_{1}$ both coincide with the weight of the edge $(v_0,v_1)$ of $G(b)$, hence $\alpha_{\mathcal L^*_{j,1}}=\alpha_{\mathcal L^{\diamond}_{1}}$. If $x(1)=1$, then the free space of $\mathcal L^*_{j,1}$ coincides with the weight of the edge $(v_0,v_1)$ minus the extension $\lambda_{\mathcal L^+_j(c_1)}(c_1)$ of $\mathcal L^+_j(c_1)$ to the left of $c_1$, while the free space of $\mathcal L^{\diamond}_{1}$ coincides with the weight of the edge $(v_0,v_1)$ minus the extension $\lambda_{\mathcal L^{\diamond}_{1}(c_1)}(c_1)$ of $\mathcal L^{\diamond}_{1}(c_1)$ to the left of $c_1$. Since  $\lambda_{\mathcal L^+_j(c_1)}(c_1)\leq \lambda_{\mathcal L^{\diamond}_{1}(c_{1})}(c_{1})$, it follows that $\alpha_{\mathcal L^*_{j,1}}\geq \alpha_{\mathcal L^{\diamond}_{1}}$.
    \item If $x(1)<k$, then the total extensions of $\mathcal L^*_{j,1}$ and $\mathcal L^{\diamond}_{1}$ both coincide with the weight of the edge $(v_0,v_k)$ of $G(b)$, hence $\tau_{\mathcal L^*_{j,1}}=\tau_{\mathcal L^{\diamond}_{1}}$. If $x(1)=k$, that is, $c_1=v_k$, then the total extension of $\mathcal L^*_{j,1}$ coincides with the weight of the edge $(v_0,v_k)$ plus the extension $\rho_{\mathcal L^+_j(c_1)}(c_1)$ of $\mathcal L^+_j(c_1)$ to the right of $c_1$, while the total extension of $\mathcal L^{\diamond}_{1}$ coincides with the weight of the edge $(v_0,v_k)$ plus the extension $\rho_{\mathcal L^{\diamond}_{1}(c_1)}(c_1)$ of $\mathcal L^{\diamond}_{1}(c_1)$ to the right of $c_1$. Since  $\rho_{\mathcal L^+_j(c_1)}(c_1)\leq \rho_{\mathcal L^{\diamond}_{1}(c_{1})}(c_{1})$, it follows that $\tau_{\mathcal L^*_{j,1}}\leq \tau_{\mathcal L^{\diamond}_{1}}$.
\end{itemize}

We also observe that $\mathcal L^*_{j,1}$ satisfies Condition~(2). Indeed, by construction, the extension $\rho_{\mathcal L^*_{j,1}(c_1)}(c_1)$ of $\mathcal L^*_{j,1}(c_1)$ to the right of $c_1$ is smaller than or equal to the extension  $\rho_{\mathcal L^{\diamond}_1(c_1)}(c_1)$ of $\mathcal L^{\diamond}_1(c_1)$ to the right of $c_1$.

Now suppose that, for some $i\in \{2,\dots,h\}$, the algorithm {\sc sum-be-drawer} constructs (and does not discard) a {\sc sum}-constrained book-embedding $\mathcal L^*_{j,i-1}$ of $G(b)\cup G^+(c_1)\cup \dots \cup G^+(c_{i-1})$ such that Conditions (1) and (2) are satisfied. Since $\mathcal L^+_1(c_i),\dots,\mathcal L^+_{k_i}(c_i)$ satisfy Properties (C1)--(C3), there exists a {\sc sum}-constrained book-embedding $\mathcal L^+_p(c_i)$ that left-right dominates or is left-right equivalent to $\mathcal L^{\diamond}_i(c_i)$ with respect to $c_i$; that is, $\lambda_{\mathcal L^+_p(c_i)}(c_i)\leq \lambda_{\mathcal L^{\diamond}_{i}(c_{i})}(c_{i})$ and $\rho_{\mathcal L^+_p(c_i)}(c_i)\leq \rho_{\mathcal L^{\diamond}_{i}(c_{i})}(c_{i})$. By Condition~(2) for $\mathcal L^*_{j,i-1}$, we have $\rho_{\mathcal L^*_{j,i-1}(c_{i-1})}(c_{i-1})\leq \rho_{\mathcal L^{\diamond}_{i-1}(c_{i-1})}(c_{i-1})$. We distinguish two cases.

\begin{itemize}
    \item Suppose first that $x(i)>x(i-1)+1$, that is, $c_{i-1}$ and $c_i$ are not consecutive in $\mathcal L(b)$. Since $\mathcal L^{\diamond}_i$ is a {\sc sum}-constrained book-embedding, the weight of the edge $(v_{x(i)-1},v_{x(i)})$ is larger than $\lambda_{\mathcal L^{\diamond}_{i}(c_{i})}(c_{i})$, hence it is larger than $\lambda_{\mathcal L^+_p(c_i)}(c_i)$, and the weight of the edge $(v_{x(i)},v_{x(i)+1})$ is larger than $\rho_{\mathcal L^{\diamond}_{i}(c_{i})}(c_{i})$, hence it is larger than $\rho_{\mathcal L^+_p(c_i)}(c_i)$. It follows that the algorithm {\sc sum-be-drawer} constructs (and does not discard) a {\sc sum}-constrained book-embedding $\mathcal L^*_{j,i}$ of $G(b)\cup G^+(c_1)\cup \dots G^+(c_i)$ by replacing $c_i$ with an embedding $\mathcal L^+_q(c_i)$ in $\mathcal L^*_{j,i-1}$. The embedding $\mathcal L^+_q(c_i)$ is the embedding among $\mathcal L^+_1(c_i),\dots,\mathcal L^+_{k_i}(c_i)$ whose extension to the left of $c_i$ is smaller than $\omega((v_{x(i)-1},v_{x(i)}))$ and is maximum, subject to the previous constraint; note that at least one embedding among $\mathcal L^+_1(c_i),\dots,\mathcal L^+_{k_i}(c_i)$ exists whose extension to the left of $c_i$ is smaller than $\omega((v_{x(i)-1},v_{x(i)}))$, namely $\mathcal L^+_p(c_i)$.
    \item Suppose next that $x(i)=x(i-1)+1$, that is, $c_{i-1}$ and $c_i$ are consecutive in $\mathcal L(b)$. Since $\mathcal L^{\diamond}_i$ is a {\sc sum}-constrained book-embedding, the weight of the edge $(v_{x(i)-1},v_{x(i)})$ is larger than $\lambda_{\mathcal L^{\diamond}_{i}(c_{i})}(c_{i})+\rho_{\mathcal L^{\diamond}_{i-1}(c_{i-1})}(c_{i-1})$, hence it is larger than $\lambda_{\mathcal L^+_p(c_i)}(c_i)+\rho_{\mathcal L^+_{j,i-1}(c_{i-1})}(c_{i-1})$, and the weight of the edge $(v_{x(i)},v_{x(i)+1})$ is larger than $\rho_{\mathcal L^{\diamond}_{i}(c_{i})}(c_{i})$, hence it is larger than $\rho_{\mathcal L^+_p(c_i)}(c_i)$. It follows that the algorithm {\sc sum-be-drawer} constructs (and does not discard) a {\sc sum}-constrained book-embedding $\mathcal L^*_{j,i}$ of $G(b)\cup G^+(c_1)\cup \dots G^+(c_i)$ by replacing $c_i$ with an embedding $\mathcal L^+_q(c_i)$ in $\mathcal L^*_{j,i-1}$. The embedding $\mathcal L^+_q(c_i)$ is the embedding among $\mathcal L^+_1(c_i),\dots,\mathcal L^+_{k_i}(c_i)$ whose extension to the left of $c_i$ is smaller than $\omega((v_{x(i)-1},v_{x(i)}))-\rho_{\mathcal L^+_{j,i-1}(c_{i-1})}(c_{i-1})$ and is maximum, subject to the previous constraint; note that at least one embedding among $\mathcal L^+_1(c_i),\dots,\mathcal L^+_{k_i}(c_i)$ exists whose extension to the left of $c_i$ is smaller than $\omega((v_{x(i)-1},v_{x(i)}))-\rho_{\mathcal L^+_{j,i-1}(c_{i-1})}(c_{i-1})$, namely $\mathcal L^+_p(c_i)$.
\end{itemize}

The proofs that $\mathcal L^*_{j,1}$ satisfies Condition~(2) and that the total extension of $\mathcal L^*_{j,i}$ is smaller than or equal to the one of $\mathcal L^{\diamond}_{i}$ are the same as for the case in which $i=1$, except that $x(i)$, $c_i$, $\mathcal L^*_{j,i}$, $\mathcal L^{\diamond}_{i}$ replace $x(1)$, $c_1$, $\mathcal L^*_{j,1}$, and $\mathcal L^{\diamond}_{1}$, respectively. Further, the free spaces of $\mathcal L^*_{j,i}$ and $\mathcal L^{\diamond}_{i}$ coincide with the free spaces of $\mathcal L^*_{j,i-1}$ and $\mathcal L^{\diamond}_{i-1}$, respectively, hence by induction we have $\alpha_{\mathcal L^*_{j,i}}=\alpha_{\mathcal L^*_{j,i-1}}\geq \alpha_{\mathcal L^{\diamond}_{i-1}}=\alpha_{\mathcal L^{\diamond}_{i}}$. This concludes the induction.

By Condition (1), the algorithm {\sc sum-be-drawer} constructs (and does not discard) a {\sc sum}-constrained book-embedding $\mathcal L^*_{j,h}$ of $G^+(b)$ that up-down dominates or is up-down equivalent to $\mathcal L^{\diamond}_h=\mathcal L^{\diamond}$. Since $\mathcal L^*_{j,h}$ is in $\mathcal S$, then after the polishing, we have that $\mathcal S$ contains either $\mathcal L^*_{j,h}$ or an embedding that up-down dominates or is up-down equivalent to $\mathcal L^*_{j,h}$, and hence up-down dominates or is up-down equivalent to $\mathcal L^{\diamond}$. This contradicts the above supposition and concludes the proof that $\mathcal S$ satisfies Property~(B3).
\end{proof}

{\bf Processing the root.}
The way we deal with the root $b^*$ of $T$ is similar, and actually simpler, than the way we deal with a B-node $b\neq b^*$. 

First, since the algorithm {\sc sum-be-drawer} did not terminate because of Failure Condition~1, we have a {\sc sum}-constrained book-embedding $\mathcal L(b^*)=(v_0,v_1,\dots,v_k)$ of $G(b^*)$. Further, let $c_1, \dots, c_h$ be the C-nodes that are children of $c$, labeled in the same order as they appear in $\mathcal L(b^*)$. Since the algorithm {\sc sum-be-drawer} did not terminate when visiting $c_1, \dots, c_h$, we have, for each $c_i$ with $i=1,\dots,h$, a sequence $\mathcal{L}^+_1(c_i), \mathcal{L}^+_2(c_i), \dots, \mathcal{L}^+_{k_i}(c_i)$ of {\sc sum}-constrained book-embeddings of $G^+(c_i)$ satisfying Properties~(C1)--(C3).

Differently from the case in which $b\neq b^*$, it might happen that $c_1=v_0$, that is, the first vertex of $\mathcal L(b^*)$ corresponds to a C-node that is a child of $b^*$ in $T$, whereas for a B-node $b\neq b^*$ the vertex $v_0$ always corresponds to the C-node that is the parent of $b$ in $T$. However, here we do not need to construct all the Pareto-optimal (with respect to the free space and the total extension) {\sc sum}-constrained book-embeddings of $G$, but we just need to test whether any {\sc sum}-constrained book-embedding of $G$ exists (and in case it does, to construct such an embedding). Hence, if  $c_1=v_0$, we can choose $\mathcal{L}^+_{k_1}(c_1)$ as the embedding for $G^+(c_1)$, given that $\mathcal{L}^+_{k_1}(c_1)$ is an embedding of $G^+(c_1)$ that satisfies Property~(C1) and that has a minimum extension to the right of $c_1$ and hence leaves most room for the embedding of $G^+(c_2)$. After this choice, the algorithm continues as in the case of a B-node $b$ different from $b^*$. 

In the case in which $c_1\neq v_0$, we process $b^*$ exactly as we process a B-node $b\neq b^*$ in the case in which $c_1\neq v_1$. The proof of the following lemma is very similar (and in fact simpler) to the proof of Lemma~\ref{le:correctness-B}, and is hence omitted.

\begin{lemma} \label{le:sum-root}
If $G$ admits a {\sc sum}-constrained book-embedding, then the algorithm {\sc sum-be-drawer} constructs such an embedding, otherwise it concludes that $G$ admits no {\sc sum}-constrained book-embedding.
\end{lemma}


{\bf Running time.} The algorithm {\sc sum-be-drawer} processes a B-node in $O(n^2)$ time and a C-node in $O(h n^2 \log n)$ time, where $h$ is the number of children of the C-node. Since the BC-tree $T$ has $O(n)$ nodes and edges, the running time of the algorithm {\sc sum-be-drawer} is in $O(n^3 \log n)$. This completes the proof of Theorem~\ref{th:linear-sum-outerplanar}.

\section{Two-Dimensional Book-Embeddings}\label{se:two-dimensional}

In order to deal with weighted outerplanar graphs that admit no {\sc max}-constrained and no {\sc sum}-constrained $1$-page book-embedding (a cycle with three edges that all have the same weight is an example of such a graph), a possibility is to give to each edge not only a length but also a thickness, so that the area of the lune representing an edge is proportional to its weight.

Given a weighted outerplanar graph $G=(V,E,\omega)$ a \emph{two-dimensional book-embedding} $\Gamma$ of $G$ consists of a $1$-page book-embedding $\mathcal{L}$ and of a representation $\mathcal R$ of $G$ satisfying the following conditions:
\begin{enumerate}
    \item Each vertex $v \in V$ is assigned an $x$-coordinate $x(v)$ such that if $u \prec_{\mathcal{L}} v$ then $x(u) < x(v)$; further, each vertex $v \in V$ is assigned the $y$-coordinate $y(v)=0$.
    \item For each edge $e=(u,v) \in E$ such that $u \prec_{\mathcal{L}} v$ we have that:
        \begin{enumerate}
            \item The edge $e$ is represented by an axis-parallel rectangle $\mathcal{R}(e):=[x_{\min}(e),x_{\max}(e)]\times [y_{\min}(e),y_{\max}(e)]$, where $y_{\min}(e)\geq 0$.
            \item We have that $x_{\min}(e) = x(u)$ and $x_{\max}(e) = x(v)$.
            \item The area $\big(x_{\max}(e) - x_{\min}(e)\big) \times \big(y_{\max}(e) - y_{\min}(e)\big)$ is equal to $\omega(e)$.
            \item Let $e_1,\dots,e_k$ be the edges in $E$ that are nested into $e$. We have that $y_{\min}(e) = \max_{i=1,\dots,k}\{y_{\max}(e_i)\}$. 
        \end{enumerate} 
\end{enumerate}

The \emph{area} of $\Gamma$ is the area of the \emph{bounding box} of $\mathcal R$, which is the smallest axis-parallel rectangle enclosing $\mathcal R$. We say that $\mathcal{L}$ is the $1$-page book-embedding \emph{supporting} $\Gamma$ and that $\mathcal R$ is the \emph{representation underlying} $\Gamma$. Further, $\Gamma$ has the following properties.

\begin{property}
Let $e_1$ and $e_2$ be two distinct edges of $G$. We have that $\mathcal{R}(e_1)$ and $\mathcal{R}(e_2)$ are internally disjoint.
\end{property}
\begin{proof}
Suppose, for a contradiction, that two rectangles $\mathcal{R}(e_1)$ and $\mathcal{R}(e_2)$ are not internally disjoint, where $e_1=(u,v)$ and $e_2=(w,z)$. Assume, w.l.o.g., that  $u \prec_{\mathcal L} v$ and $w \prec_{\mathcal L} z$. Since $\mathcal{R}(e_1)$ and $\mathcal{R}(e_2)$ are not internally disjoint and by Condition~1, we have neither $v \prec_{\mathcal L} w$ nor $z \prec_{\mathcal L} u$. Since $\mathcal L$ is a  $1$-page book-embedding, we have neither $u \prec_{\mathcal L} w \prec_{\mathcal L} v \prec_{\mathcal L} z$ nor $w \prec_{\mathcal L} u \prec_{\mathcal L} z \prec_{\mathcal L} v$. It remains to consider the cases $u \preceq_{\mathcal L} w \prec_{\mathcal L} z \preceq_{\mathcal L} v$ and $w \preceq_{\mathcal L} u \prec v \preceq_{\mathcal L} z$. Suppose that $u \preceq_{\mathcal L} w \prec_{\mathcal L} z \preceq_{\mathcal L} v$ (the other case being analogous). This implies that $(u,v) \wraps (w,z)$ in $\mathcal L$. By Condition 2(d) we have that $y_{\min}(e_1) \geq y_{\max}(e_2)$, which contradicts the assumption that $\mathcal{R}(e_1)$ and $\mathcal{R}(e_2)$ are not internally disjoint.
\end{proof}

In the Introduction, we proposed to represent each vertex of $G$ as a point on the boundary of a disk and each edge $(u,v)$ of $G$ as a lune that connects the points representing $u$ and $v$ and that has an area equal to the weight of $(u,v)$. On the contrary, in the above definition, vertices are placed along a straight line and edges are represented as rectangles. This has been done to simplify the geometric constructions. However, Property~\ref{pro:connections} below allows us to connect the rectangle representing an edge $(u,v)$ with the points representing $u$ and $v$, without intersecting the internal points of any other rectangle, thus showing the topological equivalence of the two representations. See Fig. \ref{fig:2D-embedding}.

\begin{property}\label{pro:connections}
Let $e\in E$ and consider the rectangle $\mathcal{R}(e)$. Let $\ell$ (let $r$) be the segment connecting the points $(x_{\min}(e),y_{\min}(e))$ and $(x_{\min}(e),0)$ (respectively, the points $(x_{\max}(e),y_{\min}(e))$ and $(x_{\max}(e),0)$). For each edge $e'\in E$, the segments $\ell$ and $r$ do not contain any internal point of the rectangle $\mathcal{R}(e')$.
\end{property}
\begin{proof}
If $e'=e$, then the statement follows from the definition of $\ell$ and $r$ and from Condition 2(a). Otherwise, suppose, for a contradiction, that $\ell$ contains an internal point of $\mathcal{R}(e')$; the case in which $r$ contains an internal point of $\mathcal{R}(e')$ is analogous. Let $u$ and $v$ be the end-vertices of $e$ and let $w$ and $z$ be the end-vertices of $e'$. Assume, w.l.o.g., that $u \prec_{\mathcal L} v$ and $w \prec_{\mathcal L} z$.
Since $\ell$ contains an internal point of $\mathcal{R}(e')$, we have that $w \prec_{\mathcal L} u \prec_{\mathcal L} z$. We cannot have $z \prec_{\mathcal L} v$, as this would imply that $\mathcal L$ is not a $1$-page book-embedding. Hence, $w \prec_{\mathcal L} u  \prec_{\mathcal L} v \prec_{\mathcal L} z$. However, by Condition 2(d), this implies that $\mathcal{R}(e')$ lies above $\mathcal{R}(e)$, hence $\ell$ cannot intersect $\mathcal{R}(e')$, a contradiction.
\end{proof}


The next theorems show that all weighted outerplanar graphs admit  two-dimensional book-embeddings. 

The first theorem shows that a weighted biconnected outerplanar graph $G=(V,E,\omega)$ admits a two-dimensional book-embedding $\Gamma$ in area $\sum_{e \in E}\omega(e)$. This bound is clearly optimal, as each edge $e\in E$ occupies area $\omega(e)$ in any two-dimensional book-embedding of $G$; in other words, the representation $\mathcal R$ underlying $\Gamma$ fills its bounding box, leaving no ``holes'' inside, where a \emph{hole} is a maximal connected region of the plane that lies inside the bounding box of $\mathcal R$ and does not intersect the interior or the boundary of any rectangle $\mathcal{R}(e)$. Before proving the theorem, we show a simple property of such area-optimal embeddings, which will be used in the following. 

\begin{property}\label{pro:no-hole}
Let $\Gamma$ be a two-dimensional book-embedding of a weighted biconnected outerplanar graph $G=(V,E,\omega)$ with area $\sum_{e \in E}\omega(e)$ and let $\mathcal L$ be the $1$-page book-embedding supporting $\Gamma$. We say that an edge $e_1$ \emph{directly wraps around} an edge $e_2$ in $\mathcal L$ if $e_1 \wraps e_2$ and there is no edge $e_3$ such that $e_1 \wraps e_3 \wraps e_2$. 

Let $e$ be any edge in $E$ and let $e_1,\dots,e_k$ be the edges in $E$ such that $e$ directly wraps around $e_1,\dots,e_k$. Then $y_{\min}(e) = y_{\max}(e_1) = \cdots =  y_{\max}(e_k)$.  
\end{property}

\begin{proof}
Since $e$ directly wraps around $e_1,\dots,e_k$, it follows that $e_1,\dots,e_k$ are nested into $e$. By Condition~2(d) of a two-dimensional book-embedding, we have $y_{\min}(e) = \max_{i=1,\dots,k}\{y_{\max}(e_i)\}$, which implies that $y_{\min}(e) \geq y_{\max}(e_i)$, for $i=1,\dots,k$. Since $G$ is biconnected and since $e$ directly wraps around $e_1,\dots,e_k$, we have that $e,e_1,\dots,e_k$ induce a cycle $(u_1,\dots,u_{k+1})$, where $e_i=(u_i,u_{i+1})$, for $i=1,\dots,k$, and $e=(u_1,u_{k+1})$;  further, again since $e$ directly wraps around $e_1,\dots,e_k$, by Conditions~1 and 2(b) of a two-dimensional book-embedding, we have that either $x(u_1)<x(u_2)<\dots<x(u_{k+1})$ or that $x(u_1)>x(u_2)>\dots>x(u_{k+1})$. Hence, if $y_{\min}(e) > y_{\max}(e_i)$, for some $i\in \{1,\dots,k\}$, then there would be a hole above the rectangle $\mathcal{R}(e_i)$, contradicting the assumption that the area of $\Gamma$ is $\sum_{e \in E}\omega(e)$. 
\end{proof}

We are now ready to present the following theorem; see Fig. \ref{fig:2D-embedding} for an example of a drawing produced by the algorithm described in the proof of the theorem.

\begin{theorem}\label{th:2d-book-outerplanar-biconnected}
Let $G=(V,E,\omega)$ be an $n$-vertex weighted biconnected outerplanar graph; further, let $s$ and $t$ be two vertices that are consecutive in the clockwise order of the vertices of $G$ along the outer face of the outerplane embedding of $G$. Finally, let $L > 0$ and $H>0$ be two real values such that $L \times H = \sum_{e \in E}\omega(e)$. There exists an $O(n)$-time algorithm that constructs a two-dimensional book-embedding $\Gamma$ in area $L \times H$ such that $s$ and $t$ are the first and the last vertex of the $1$-page book-embedding supporting $\Gamma$, respectively. 
\end{theorem}
\begin{proof}
First, we construct in $O(n)$ time the $1$-dimensional book-embedding $\mathcal L$ supporting $\Gamma$ as the unique $1$-dimensional book-embedding of $G$ in which $s$ and $t$ are the first and the last vertex, respectively~\cite{d-iroga-07,m-laarogmog-79,w-rolt-87}. Note that $\mathcal L$ defines an outerplane embedding $\mathcal{O}_G$ of $G$ such that $s$ is encountered immediately before $t$ when traversing the cycle delimiting the outer face of $\mathcal{O}_G$ in clockwise direction. We construct in $O(n)$ time the extended dual tree $\mathcal{T}$ of $\mathcal{O}_G$; further, we root~$\mathcal{T}$ at the leaf $\rho$ that is incident to the edge $(\rho,\sigma)$ of $\mathcal{T}$ that is dual to the edge $e^*=(s,t)$.
Second, for each edge $e\in E$, we compute a value $\mathcal{A}(e)$ which is equal to the sum of $\omega(e)$ plus the weights of the edges that are nested into $e$ in $\mathcal L$. This is done in total $O(n)$ time by means of a bottom-up traversal of $\mathcal T$. 

The proof now proceeds by induction. The induction receives as an input: 

\begin{enumerate}[(1)]
	\item a weighted biconnected outerplanar graph $K=(\mathcal V,\mathcal E,\kappa)$, which is a subgraph of $G$;
	\item a $1$-dimensional book-embedding $\mathcal K$ of $K$, whose first and last vertex are denoted by $s'$ and $t'$, respectively;
	\item  an assignment for $x(s')$ and $x(t')$ with $x(t')-x(s')=L'>0$; and
	\item a rectangle $\mathcal B=[x(s'),x(t')]\times [0,H']$ such that $L' \times H' = \sum_{e \in \mathcal E}\kappa(e)$.
\end{enumerate} 

The induction defines an output which is a two-dimensional book-embedding $\Gamma$ of $K$ whose underlying representation has $\mathcal B$ as bounding box and whose supporting $1$-dimensional book-embedding is $\mathcal K$, so that $s'$ and $t'$ have $x$-coordinates $x(s')$ and $x(t')$, respectively. The induction implies the theorem with $K=G$, $\mathcal K=\mathcal L$, $s'=s$, $t'=t$, $\kappa=\omega$, $L'=L$, $H'=H$,  $x(s')=x(s)=0$, and $x(t')=x(t)=L$.

In the base case, $K$ is a single edge $e^{\circ}$. Then the representation $\mathcal R$ underlying $\Gamma$ consists only of the rectangle $\mathcal{R}(e^{\circ})$, which coincides with~$\mathcal B$.

\begin{figure}[htb]
	\centering
	\subfloat[]{\label{fig:2d-structure}\includegraphics[scale=0.5]{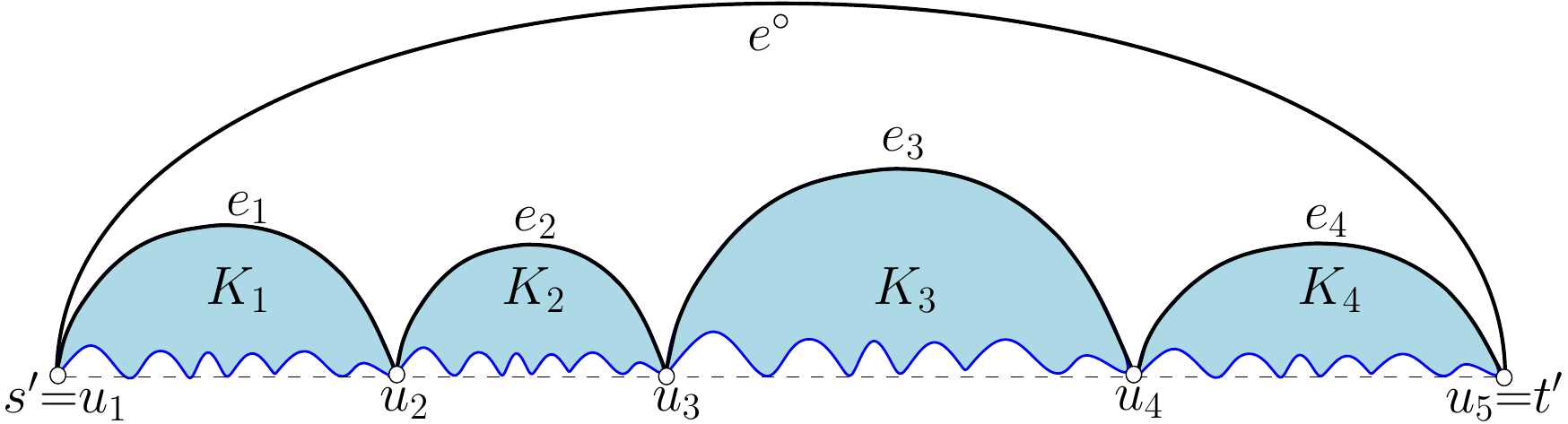}}
	\hfil
	\subfloat[]{\label{fig:2d-drawing}\includegraphics[scale=0.5]{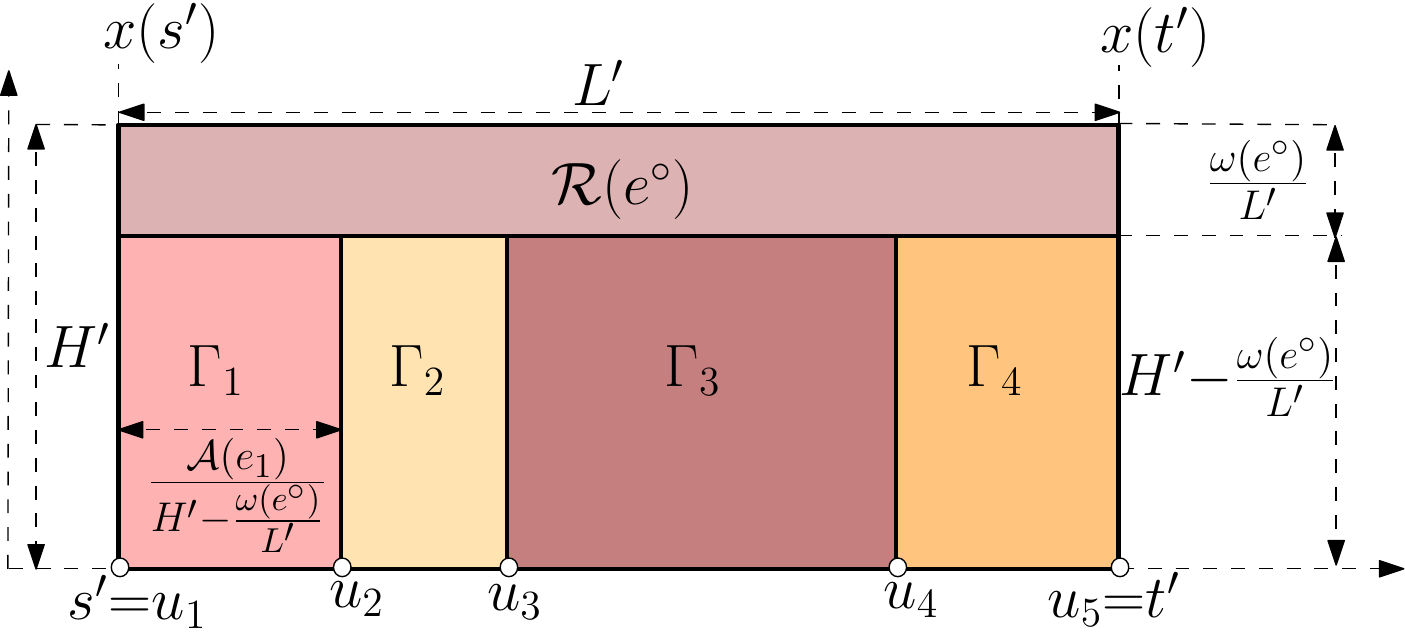}}
	\hfil
	\caption{Illustration for the inductive case of the proof of Theorem~\ref{th:2d-book-outerplanar-biconnected}. (a) The graphs $K,K_1,\dots,K_k$, the edges $e^{\circ},e_1,\dots,e_k$, and the vertices $u_1,\dots,u_{k+1}$. In this example, $k=4$. (b) Construction of a two-dimensional book-embedding $\Gamma$ of $K$ from two-dimensional book-embeddings $\Gamma_1,\dots,\Gamma_k$ of $K_1,\dots,K_k$.}
	\label{fig:dual-tree-and-stuff}
\end{figure}

In the inductive case, $K$ has more than one edge; refer to Fig.~\ref{fig:2d-structure}. Let $\mathcal{O}_K$ be the outerplane embedding of $K$ associated to $\mathcal K$; in particular, $s'$ is encountered immediately before $t'$ when traversing the cycle delimiting the outer face of $\mathcal{O}_K$ in clockwise direction. Since $K$ is biconnected and $e^{\circ}$ is incident to the outer face of $\mathcal O_K$, there exists an internal face of $\mathcal O_K$ that is delimited by a simple cycle containing $e^{\circ}$. Let $(s'=u_1,u_2,\dots,u_{k+1}=t')$ be such a cycle, where we define $e_i=(u_i,u_{i+1})$, for $i=1,\dots,k$; then $e^\circ$ directly wraps around $e_1,\dots,e_k$ in $\mathcal K$ and $u_1\prec_{\mathcal K} u_2 \prec_{\mathcal K} \dots \prec_{\mathcal K} u_{k+1}$. 

For $i=1,\dots,k-1$, we set $x(u_{i+1})=x(u_{i})+\frac{\mathcal A(e_{i})}{H'-\kappa(e^\circ)/L'}$ and $y(u_{i+1})=0$. Then we apply induction $k$ times, namely, for $i=1,\dots,k$, we apply induction with:

\begin{enumerate}[(1)]
	\item the weighted biconnected outerplanar graph $K_i=(\mathcal V_i,\mathcal E_i,\kappa_i)$ induced by $e_i$ and by the edges nested into $e_i$ in $\mathcal K$, where the weight function $\kappa_i$ is the restriction of $\kappa$ to the edges in $\mathcal E_i$;
	\item a $1$-dimensional book-embedding $\mathcal K_i$ of $K_i$, whose first and last vertex are $u_i$ and $u_{i+1}$, respectively; this book-embedding is the restriction of $\mathcal K$ to $K_i$;
	\item the assignment for $x(u_i)$ and $x(u_{i+1})$ defined above; and
	\item the rectangle $\mathcal B_i=[x(u_i),x(u_{i+1})]\times [0,H'-\frac{\kappa(e^\circ)}{L'}]$.
\end{enumerate} 

We denote by $\Gamma_i$ the two-dimensional book-embedding of $K_i$ constructed by induction. Finally, we draw $e^\circ$ as the rectangle $\mathcal{R}(e^\circ)=[x(s'),x(t')]\times [H'-\frac{\kappa(e^\circ)}{L'},H']$. See Fig.~\ref{fig:2d-drawing}.

We now prove the correctness of the above-described algorithm. First, we prove that, in the inductive case, the area of $\mathcal B_i$ is equal to $\sum_{e \in \mathcal E_i}\kappa_i(e)=\sum_{e \in \mathcal E_i}\kappa(e)=\sum_{e \in \mathcal E_i}\omega(e)$, which ensures the correctness of the inductive calls. 

If $i\leq k-1$ then, by construction, we have $x(u_{i+1})=x(u_{i})+\frac{\mathcal A(e_{i})}{H'-\kappa(e^\circ)/L'}$, hence the area of $\mathcal B_i$ is equal to $\frac{\mathcal A(e_{i})}{H'-\kappa(e^\circ)/L'} \times (H'-\frac{\kappa(e^\circ)}{L'})=\mathcal A(e_{i})=\sum_{e \in \mathcal E_i}\omega(e)$.

We now prove that the area of $\mathcal B_k$ is equal to $\sum_{e \in \mathcal E_k}\kappa(e)$. By construction, we have $x(u_k)=x(s')+\frac{\sum_{i=1}^{k-1}\mathcal A(e_{i})}{H'-\kappa(e^\circ)/L'}=x(s')+\frac{H'\times L' -\kappa(e^\circ)- A(e_k)}{H'-\kappa(e^\circ)/L'}=x(s')+L'- \frac{\mathcal A(e_k)}{H'-\kappa(e^\circ)/L'}=x(t')- \frac{\mathcal A(e_k)}{H'-\kappa(e^\circ)/L'}$, where the second equality exploits the fact that the sum of the weights of the edges in $\mathcal E$ is equal to $H'\times L'$ and to $\kappa(e^\circ)+\sum_{i=1}^{k}\mathcal A(e_{i})$. It follows that the area of $\mathcal B_k$ is equal to $\frac{\mathcal A(e_{k})}{H'-\kappa(e^\circ)/L'} \times (H'-\frac{\kappa(e^\circ)}{L'})=\mathcal A(e_{k})=\sum_{e \in \mathcal E_k}\omega(e)$. 

We now prove that the constructed representation satisfies Condition~(1) and Conditions~(2)a--(2)d of a two-dimensional book-embedding. 

\begin{itemize}
	\item Condition~(1): As described above, we have $u_{i}\prec_{\mathcal K} u_{i+1}$, for $i=1,2,\dots,k$. We prove that $x(u_{i+1})>x(u_{i})$, for $i=1,2,\dots,k$. 
	
	If $i\leq k-1$ then, by construction, we have $x(u_{i+1})=x(u_{i})+\frac{\mathcal A(e_{i})}{H'-\kappa(e^\circ)/L'}$. Since $H'=\frac{\sum_{e\in \mathcal E} \kappa(e)}{L'}>\frac{\kappa(e^\circ)}{L'}$, we have that $\frac{\mathcal A(e_{i})}{H'-\kappa(e^\circ)/L'}>0$, and hence $x(u_{i+1})>x(u_{i})$. 
	
	We now prove that $x(t')=x(u_{k+1})>x(u_k)$. As argued above, we have $x(u_k)=x(t')- \frac{\mathcal A(e_k)}{H'-\kappa(e^\circ)/L'}$. Since $H'>\frac{\kappa(e^\circ)}{L'}$, it follows that $x(u_{k+1})=x(t')>x(u_k)$.
	
	By induction, for $i=1,2,\dots,k$, we have that the $1$-dimensional book-embedding supporting $\Gamma_i$ is $\mathcal K_i$. Since $\Gamma_i$ satisfies Condition~(1), the order of the vertices of $K_i$ by increasing $x$-coordinates is $\mathcal K_i$; in particular, $u_i$ and $u_{i+1}$ are respectively the vertex with the smallest and the largest $x$-coordinate in $\Gamma_i$.
	
	Now consider any two distinct vertices $u$ and $v$ of $K$ respectively belonging to $K_i$ and $K_j$, for some $i,j\in \{1,\dots,k\}$; we assume w.l.o.g.\ that $i\leq j$. If $i=j$, then we have that $u\prec_{\mathcal K} v$	if and only if $x(u)< x(v)$, given that the same property is satisfied in $\Gamma_i$, as argued above, and given that the restriction of $\Gamma$ to $K_i$ is $\Gamma_i$. If $i<j$, then we have $u\preceq u_{i+1} \preceq u_{j} \preceq v$, where one of the three precedence relationships is strict, given that $u$ and $v$ are distinct. Further, $x(u)\leq x(u_{i+1})$, given that $u_{i+1}$ is the vertex with the largest $x$-coordinate in $\Gamma_i$; analogously, $x(u_j)\leq x(v)$, given that $u_j$ is the vertex with the smallest $x$-coordinate in $\Gamma_j$; finally, $x(u_{i+1})\leq x(u_j)$, where the equality holds only if $j=i+1$. Hence, $x(u)\leq x(u_{i+1})\leq x(u_j)\leq x(v)$, where one of the three inequalities is strict, given that $u$ and $v$ are distinct. It follows that $\Gamma$ satisfies Condition~(1).
	
	\item Condition~(2)a: At each step of the induction, by construction, we represent a single edge $e^\circ$ by an axis-parallel rectangle $\mathcal{R}(e^\circ)$. Hence, every edge of $K$ is represented by an axis-parallel rectangle.
	\item Condition~(2)b: At each step of the induction, by construction, we draw a single axis-parallel rectangle $\mathcal{R}(e^\circ)$ representing the edge $e^\circ=(s',t')$ of $K$, so that $x_{\min}(e^\circ)=x(s')$ and $x_{\max}(e^\circ)=x(t')$. Hence, every edge $e=(u,v)$ of $K$ is such that $x_{\min}(e)=x(u)$ and $x_{\max}(e)=x(v)$.
	\item Condition~(2)c: At each step of the induction, we draw a single axis-parallel rectangle $\mathcal{R}(e^\circ)$ representing the edge $e^\circ$ of $K$. In the base case, the area of $\mathcal{R}(e^\circ)$ is $(x(t')-x(s'))\times H' = L'\times H' = \sum_{e \in \mathcal E}\kappa(e) = \kappa(e^\circ)$, as requested. In the inductive case, the area of $\mathcal{R}(e^\circ)$ is $[x(s'),x(t')]\times [H'-\kappa(e^\circ)/L',H']=L'\times \kappa(e^\circ)/L' = \kappa(e^\circ)$, as requested. Hence, every edge $e$ of $K$ is represented by an axis-parallel rectangle $\mathcal{R}(e)$ whose area is $\kappa(e)$.
	\item Condition~(2)d: At each step of the induction, we assign the value $y_{\min}(e^\circ)=H'-\kappa(e^\circ)/L'$ for the edge $e^\circ$. Further, the inductive calls ensure that every edge $e$ of $K$ different from $e^\circ$ is represented by a rectangle whose $y$-coordinates are in $[0,H'-\kappa(e^\circ)/L']$, hence $y_{\max}(e)\leq y_{\min}(e^\circ)$.
\end{itemize}
	
Finally, we discuss the running time of the above-described algorithm. The $1$-page book-embedding $\mathcal L$, the extended dual tree $\mathcal T$ of the outerplane embedding $\mathcal O_G$ of $G$, and the value $\mathcal A(e)$ for each edge $e\in E$ can be computed in total $O(n)$ time, as discussed above. Assume that each edge $e$ of $G$ stores a linear list $\mathcal L(e)$, which represents what follows. Let $(a,b)$ be the edge of $\mathcal T$ that is dual to $e$, where $a$ is the parent of $b$. If $b$ is a leaf of $\mathcal T$ (and hence $e$ is an edge incident to the outer face of $\mathcal O_G$ and different from $e^*$), then $\mathcal L(e)=\emptyset$. Otherwise, $\mathcal L(e)$ represents the counter-clockwise order of the vertices along the cycle delimiting the internal face of $\mathcal O_G$ that is dual to $b$, where the end-vertices of $e$ are the first and the last vertex of $\mathcal L(e)$. Such lists can be set-up in total $O(n)$ time by means of a visit of $\mathcal O_G$.

In the base case of the inductive algorithm, the computation time is obviously constant. In the inductive case, the vertices $u_1,u_2,\dots, u_{k+1}$ are found in $O(k)$ time, as these are the vertices in the list $\mathcal L(e^\circ)$. Then the coordinates $x(u_1),x(u_2),\dots, x(u_{k+1})$ can also be found in $O(k)$ time from the pre-computed labels $\mathcal A(e_i)$. The graphs $K_1,\dots,K_k$ and the $1$-page book-embeddings $\mathcal K_1,\dots,\mathcal K_k$ do not need to be computed explicitly; indeed, the lists $\mathcal L(e_1),\dots,\mathcal L(e_k)$ represent all the information that is needed for the induction to continue. Hence, the algorithm spends $O(k)$ time when processing $e^\circ$. Since $k$ is the degree in $\mathcal T$ of the vertex that is dual to the internal face of $\mathcal O_K$ incident to $e^\circ$, and since the sum of the degrees of the vertices of $\mathcal T$ is in $O(n)$, it follows that the running time of the algorithm is in $O(n)$, as well.
\end{proof}

\begin{theorem}\label{th:2d-book-outerplanar}
For any constant $\varepsilon>0$, every $n$-vertex weighted outerplanar graph $G=(V,E,\omega)$ admits a two-dimensional book-embedding whose area is smaller than or equal to $\sum_{e \in E}\omega(e) + \varepsilon$. Such an embedding can be constructed in $O(n)$ time.
\end{theorem}

\begin{proof}
If $G$ is biconnected, then it suffices to apply Theorem~\ref{th:2d-book-outerplanar-biconnected} with arbitrary positive values for $L$ and $H$ such that $L\times H=\sum_{e\in E}\omega(e)$, and with $s$ and $t$ as any two vertices that are consecutive in the clockwise order of the vertices of $G$ along the outer face of the unique outerplane embedding of $G$. 

If $G$ is connected, but not biconnected, we augment $G$ to a weighted biconnected outerplanar graph $G'$, by adding at most $n$ dummy edges of weight $\varepsilon/n$ to it; then we construct a two-dimensional book-embedding of $G'$, and finally we remove the rectangles corresponding to dummy edges. 

More formally, we start by computing a $1$-page book-embedding $\mathcal{L}$ of $G$; this can be done in $O(n)$ time~\cite{d-iroga-07,m-laarogmog-79,w-rolt-87}. 

We then augment $G=(V,E,\omega)$ to a weighted biconnected outerplanar graph $G'=(V,E',\omega')$; this can be done in $O(n)$ time as follows. First, we initialize $G'$ to $G$. Second, we add to $G'$ an edge of weight $\varepsilon/n$ between any two vertices of $G$ that are consecutive in $\mathcal{L}$, if such an edge is not already in $G$. Third, we add to $G'$ an edge of weight $\varepsilon/n$ between the first vertex $s$ and the last vertex $t$ of $\mathcal{L}$, if such an edge is not already in $G$. This augmentation guarantees the outerplanarity of $G'$; note that the number $n'$ of dummy edges that are added to $G$ in order to obtain $G'$ is smaller than or equal to $n$. Also, $G'$ has a cycle connecting all its vertices and is, hence, biconnected.  

We apply Theorem~\ref{th:2d-book-outerplanar-biconnected} to $G'$ with arbitrary positive values for $L$ and $H$ such that $L\times H=\sum_{e\in E'}\omega(e) + n'\varepsilon/n$. We thus obtain a two-dimensional book-embedding $\Gamma'$ of $G'$. Finally, we remove from $\Gamma'$ each rectangle $\mathcal{R}(e)$ corresponding to a dummy edge $e$, thus obtaining a drawing $\Gamma$ of $G$. 

We now prove that $\Gamma$ is a two-dimensional book-embedding of $G$. In fact, Conditions~(1), (2)a, (2)b, and (2)c of the definition of two-dimensional book-embedding are satisfied by $\Gamma$ since they are satisfied by $\Gamma'$. As far as Condition~(2)d is concerned, we observe what follows. Consider any edge $e$ of $G$; let $e_1, \dots, e_k$ be the edges $e$ directly wraps around; further, let $\mathcal{R}(e_1), \dots, \mathcal{R}(e_k)$ be the rectangles representing $e_1, \dots, e_k$ in $\Gamma'$. By Property~\ref{pro:no-hole}, we have that $y_{\min}(e)=y_{\max}(e_1)=\dots=y_{\max}(e_k)$.
Since $G$ is connected, at least one of $e_1, \dots, e_k$ belongs to $G$. Hence, at least one of $\mathcal{R}(e_1), \dots, \mathcal{R}(e_k)$ belongs to $\Gamma$, satisfying Condition~(2)d.

By Theorem~\ref{th:2d-book-outerplanar-biconnected}, the area of $\Gamma'$ is $\sum_{e\in E}\omega(e) + n'\varepsilon/n\leq \sum_{e\in E}\omega(e) + \varepsilon$. Since $\Gamma$ only consists the vertices of $G'$ and of some rectangles of $\Gamma'$, its area is at most the one of $\Gamma'$.

Finally, we observe that a reduction from the case in which $G$ is not connected to the one in which $G$ is connected can be performed analogously as above, by means of the addition of at most $n$ dummy edges of weight $\varepsilon/n$. It is necessary for this augmentation that the $1$-page book-embedding $\mathcal{L}$ is chosen so that no vertex of a connected component lies under an edge of a different connected component, so that Condition~(2)d is satisfied by the resulting representation.
\end{proof}


\section{Two-Dimensional Book-Embeddings with Finite Resolution}\label{se:minres}

The algorithms in the proofs of Theorems~\ref{th:2d-book-outerplanar-biconnected} and~\ref{th:2d-book-outerplanar} may produce $2$-dimensional book-embeddings in which the rectangles representing some edges can be arbitrarily small in terms of height or width. This is clearly undesirable for visualization purposes.

Hence, we study two-dimensional book-embeddings that are constrained to adopt a finite resolution rule. A \emph{{\sc minres}-constrained two-dimensional book-embedding} of a weighted outerplanar graph $G=(V,E,\omega)$ is a two-dimensional book-embedding such that:
\begin{enumerate}[(A)]
\item\label{gio-1} For each edge $e$ in $E$, we have that $x_{\max}(e) - x_{\min}(e) \geq 1$.
\item\label{gio-2} For each edge $e$ in $E$, we have that $y_{\max}(e) - y_{\min}(e) \geq 1$.
\item\label{gio-3} For each pair $u,v$ of distinct vertices in $V$, we have that $|x(v)-x(u)| \geq 1$.
\end{enumerate}

A trivial necessary condition for a weighted outerplanar graph to have a {\sc minres}-constrained two-dimensional book-embedding is that all its edges have weight greater than or equal to one. More generally, we have the following characterization. Let $\mathcal{L}$ be a $1$-page book-embedding of a graph $G$ and let $e$ be an edge of $G$. We call \emph{burden} of $e$ in $\mathcal{L}$, and denote it by $\beta(e)$, the number of vertices that lie strictly under $e$ in $\mathcal{L}$.

\begin{theorem}\label{th:minres-characterization}
An $n$-vertex weighted outerplanar  graph $G=(V,E,\omega)$ admits a {\sc minres}-constrained two-dimensional book-embedding if and only if it admits a $1$-page book-embedding $\mathcal{L}$ such that, for each edge $e \in E$, we have that $\omega(e) \geq \beta(e)+1$. 
Also, if a $1$-page book-embedding $\mathcal{L}$ satisfying this condition is given, a {\sc minres}-constrained two-dimensional book-embedding supported by $\mathcal{L}$ can be constructed in $O(n)$ time.
\end{theorem}

\begin{proof}
The necessity is easy to prove. In fact, consider a weighted outerplanar graph that, in every $1$-page book-embedding $\mathcal L$, has an edge $e$ such that $\omega(e) < \beta(e) + 1$. By Condition (\ref{gio-3}), in any {\sc minres}-constrained two-dimensional book-embedding supported  by $\mathcal L$, we have that $x_{\max}(e)-x_{\min}(e) \geq \beta(e) + 1$. Hence, we obtain $\omega(e) < \beta(e) + 1 \leq x_{\max}(e)-x_{\min}(e)$.  
Condition (2)a of the definition of two-dimensional book-embedding requires that $(x_{\max}(e)-x_{\min}(e))\times(y_{\max}(e)-y_{\min}(e)) = \omega(e)$. Therefore, we have $(y_{\max}(e)-y_{\min}(e)) = \omega(e)/(x_{\max}(e)-x_{\min}(e)) < (x_{\max}(e)-x_{\min}(e))/(x_{\max}(e)-x_{\min}(e)) = 1$, contradicting Condition~(\ref{gio-2}).

Now we deal with the sufficiency. Namely, suppose that $G$ admits a $1$-page book-embedding $\mathcal{L}$ such that, for each edge $e \in E$, we have that $\omega(e) \geq \beta(e)+1$. We construct a {\sc minres}-constrained two-dimensional book-embedding $\Gamma$ for $G$ as follows.

Let $\mathcal{L} = (v_1,v_2, \dots, v_n)$. For $i=1,\dots,n$, we set $x(v_i)=i$ and $y(v_i)=0$, so that Condition~(1) of the definition of two-dimensional book-embedding and Condition~(\ref{gio-3}) of the definition of {\sc minres}-constrained two-dimensional book-embedding are satisfied. We also assign, for every edge $e=(u,v)\in E$ such that $u\prec_{\mathcal L} v$, the value $x_{\min}(e)=x(u)$ and $x_{\max}(e)=x(v)$ to the rectangle $\mathcal{R}(e)$ representing $e$ in $\Gamma$, so that Condition~(\ref{gio-1}) of the definition of {\sc minres}-constrained two-dimensional book-embedding and Condition~(2)b of the definition of two-dimensional book-embedding are satisfied.

We now assign values $y_{\min}(e)$ and $y_{\max}(e)$ to the rectangle $\mathcal{R}(e)$ representing each edge $e$.
If $e$ is such that there is no edge $e'$ with $e' \isnested e$, we set $y_{\min}(e)=0$ and $y_{\max}(e)= \omega(e)/(x_{\max}(e)-x_{\min}(e))$. 
Otherwise, we assign $y_{\min}(e)$ and $y_{\max}(e)$ to an edge $e$ only after assigning $y_{\min}(e')$ and $y_{\max}(e')$ to all edges $e'$ such that $e' \isnested e$. Then we set $y_{\min}(e) = \max_{e' \isnested e} y_{\max}(e')$ and $y_{\max}(e)= y_{\min}(e)+\omega(e)/(x_{\max}(e)-x_{\min}(e))$. In this way we satisfy Conditions~(2)c and~(2)d of the definition of two-dimensional book-embedding.

Since by hypothesis $\omega(e) \geq \beta(e)+1$ and since by construction $\beta(e) + 1 = x_{\max}(e) - x_{\min}(e)$, we have that $y_{\max}(e) - y_{\min}(e) = \omega(e)/(x_{\max}(e) - x_{\min}(e)) = \omega(e)/(\beta(e) + 1) \geq (\beta(e) + 1)/(\beta(e) + 1) = 1$, satisfying Condition~(\ref{gio-2}) of the definition of {\sc minres}-constrained two-dimensional book-embedding. 

The described construction can be easily implemented to run in overall $O(n)$ time.
\end{proof}



A $1$-page book-embedding with the properties in the statement of Theorem~\ref{th:minres-characterization} is said to be \emph{supporting} a \textsc{minres}-constrained representation or, that it is a \emph{\textsc{minres}-supporting embedding}.

A first algorithmic contribution in the direction of testing whether an outerplanar graph has a {\sc minres}-constrained two-dimensional book-embedding is given in the following lemma.

\begin{lemma}\label{le:linear-minres-outerplanar-biconnected-algo}
Let $G=(V,E,\omega)$ be an $n$-vertex weighted biconnected outerplanar graph and let $(s,t)  \in E$ be a prescribed edge. There exists an $O(n)$-time algorithm that tests whether $G$ admits a {\sc minres}-constrained two-dimensional book-embedding in which $s$ and $t$ are the first and the last vertex of the supporting $1$-page book-embedding, respectively. In the positive case, such a representation can be constructed in $O(n)$ time.
\end{lemma}
\begin{proof}
First, we determine in $O(n)$ time the unique outerplane embedding of $G$, up to a flip, and verify whether the edge $(s,t)$ is incident to the outer face of it. In the negative case, we conclude that $G$ does not admit the required {\sc minres}-constrained two-dimensional book-embedding. In the positive case, we construct in $O(n)$ time the $1$-page book-embedding $\mathcal{L}$ such that $s$ and $t$ are the first and the last vertex of $\mathcal{L}$, respectively; note that $(s,t) \wraps e$, for each $e \in E$ such that $e \neq (s,t)$. 

It remains to test whether $\mathcal{L}$ is a \textsc{minres}-supporting embedding.
We construct in $O(n)$ time the extended dual tree $\mathcal{T}$ of the outerplane embedding of $G$. We root $\mathcal{T}$ at the leaf $r$ such that the edge of $\mathcal{T}$ incident to $r$ is dual to $(s,t)$. 
We perform a bottom-up visit of  $\mathcal{T}$ in $O(n)$ time. During this visit, we compute, for each edge $(\alpha,\gamma)$ of $\mathcal{T}$, the burden $\beta(e)$ of $e$ in $\mathcal L$, where $e$ is the edge that is dual to $(\alpha,\gamma)$; this is done as follows. Assume, w.l.o.g., that $\gamma$ is the child of $\alpha$ in $\mathcal T$. If $\gamma$ is a leaf, then we set $\beta(e)= 0$. Otherwise, let $e_1, \dots, e_h$ be the edges of $G$ that are dual to the edges from $\gamma$ to its children in $\mathcal T$; then we set $\beta(e) = h-1+\sum_{i=1,\dots,h} \beta(e_i)$.

We now check in total $O(n)$ time whether $\omega(e) \geq \beta(e) + 1$ for each edge $e \in E$. By Theorem~\ref{th:minres-characterization}, if one of these checks fails, then a {\sc minres}-constrained two-dimensional book-embedding in which $s$ and $t$ are respectively the first and the last vertex of the supporting $1$-page book-embedding does not exist. Otherwise, by means of the same theorem, we construct such a representation in $O(n)$ time. 
\end{proof}

The rest of this section is devoted to a proof of the following theorem.

\begin{theorem}\label{th:minres-outerplanar}
Let $G=(V,E,\omega)$ be an $n$-vertex weighted outerplanar graph. There exists an $O(n^4)$-time algorithm that tests whether $G$ admits a {\sc minres}-constrained two-dimensional book-embedding and, in the positive case, constructs such an embedding.  
\end{theorem}

We present an algorithm, called {\sc minres-be-drawer}, that tests in $O(n^4)$ time whether $G$ admits a \textsc{minres}-supporting embedding and, in the positive case, constructs such an embedding. Then the statement follows by Theorem~\ref{th:minres-characterization}.

Preliminarly, we compute in $O(n)$ time the block-cut-vertex tree $T$ of $G$~\cite{h-gt-69,ht-aeagm-73}. Also, for each B-node $b$ of $T$ we compute the number of vertices $n(b)$ and the unique (up to a flip) outerplane embedding of $G(b)$; this can be done in overall $O(n)$ time.

We present an algorithm, called {\sc minres-be-drawer}$(e^*)$, that tests whether a \textsc{minres}-supporting embedding of $G$ exists with the further constraint that a given edge $e^*$ is not nested into any other edge of $G$. Then {\sc minres-be-drawer} simply calls {\sc minres-be-drawer}$(e^*)$ for each edge $e^*$ of $G$. Hence, the time complexity of {\sc minres-be-drawer} is $O(n)$ times the one of {\sc minres-be-drawer}$(e^*)$.

We root $T$ at the B-node $b^*$ containing $e^*$; then, for every B-node $b$ of $T$ (for every C-node $c$ of $T$), the graph $G^+(b)$ (resp.\ $G^+(c)$) is defined as for {\sc max}- and {\sc sum}-constrained book-embeddings. For every B-node $b$ of $T$ (for every C-node $c$ of $T$), we compute the number of vertices of $G^+(b)$ (resp.\ of $G^+(c)$) and denote it by $n^+(b)$ (resp.\ by $n^+(c)$); this can be done in total $O(n)$ time by means of a bottom-up traversal of $T$. 

Let $e^*=(u,v)$. By means of Lemma~\ref{le:linear-minres-outerplanar-biconnected-algo}, we check in $O(n(b^*))$ time whether $G(b^*)$ admits a \textsc{minres}-supporting embedding $\mathcal{L}(b^*,e^*)$ in which $u$ and $v$ are the first and the last vertex, respectively. If yes, we store $\mathcal{L}(b^*,e^*)$. If not, then {\sc minres-be-drawer}$(e^*)$ concludes that $G$ admits no \textsc{minres}-supporting embedding in which $e^*$ is not nested into any other edge of $G$ (Failure Condition~$1$); the correctness of this conclusion descends from considerations analogous to those in the proof of Lemma~\ref{le:failure-linear-max-outeplanar}.

We visit $T$ in arbitrary order. For each B-node $b \neq b^*$, {\sc minres-be-drawer}$(e^*)$ performs the following checks and computations. Let $c$ be the C-node that is the parent of $b$ in $T$. Let $(c,x)$ and $(c,y)$ be the two (not necessarily distinct) edges incident to $c$ that lie on the outer face of the outerplane embedding of $G(b)$. We check whether $G(b)$ admits a {\sc minres}-supporting embedding $\mathcal{L}(b,(c,x))$ in which $c$ and $x$ are the first and the last vertex, respectively. If yes, we store $\mathcal{L}(b,(c,x))$. Then we do an analogous check for the edge $(c,y)$, possibly storing $\mathcal{L}(b,(c,y))$. By Lemma~\ref{le:linear-minres-outerplanar-biconnected-algo}, this can be done in $O(n(b))$ time. Hence, these checks require overall $O(n)$ time.
If both the test for the edge $(c,x)$ and the test for the edge $(c,y)$ fail, then {\sc minres-be-drawer}$(e^*)$ concludes that $G$ admits no \textsc{minres}-supporting embedding in which $e^*$ is not nested into any other edge of $G$ (Failure Condition~$2$); the correctness of this conclusion descends from considerations analogous to those in the proof of Lemma~\ref{le:failure-linear-max-outeplanar}.

We introduce some definitions. Let $\mathcal{L}=(v_0,v_1,\dots,v_h)$ be a $1$-page book-embedding of a connected graph $H$ and let $v_i$ be a vertex that is visible in $\mathcal{L}$. We denote by $n_{\ell}(v_i,\mathcal{L})$ and $n_r(v_i,\mathcal{L})$ the number of vertices to the left and to right of $v_i$ in $\mathcal{L}$, respectively (that is, $n_{\ell}(v_i,\mathcal{L})=i$ and $n_r(v_i,\mathcal{L})=h-i$). For each vertex $v_i$, we define a value $r(v_i)$, which is called the \emph{right residual capacity} of $v_i$, as follows. Consider the set $E_i$ that contains all the edges $(v_{i'},v_{j'})$ of $H$ such that $i'\leq i$ and $i+1 \leq j'$; that is, $E_i$ consists of the edges $v_i$ lies strictly under and of the edges incident to $v_i$ and to a vertex that follows $v_i$ in $\mathcal L$. We set $r(v_i) = \min_{e \in E_i}(\omega(e)-(\beta(e)+1))$. The \emph{left residual capacity} $\ell(v_i)$ of $v_i$ is defined analogously. By convention, we set $r(v_h)=\ell(v_0)=\infty$. The \emph{residual capacity $r(\mathcal{L})$ of $\mathcal{L}$} is the right residual capacity of $v_0$. Let $\mathcal{L}$ and $\mathcal{L}'$ be two $1$-page book-embeddings of $H$ and $c$ be a vertex that is visible both in $\mathcal{L}$ and in $\mathcal{L}'$. We say that $\mathcal{L}$ and $\mathcal{L}'$ are \emph{left-right equivalent with respect to $c$} if $n_{\ell}(c,\mathcal{L}) = n_{\ell}(c,\mathcal{L}')$. This implies that $n_r(c,\mathcal{L}) = n_r(c,\mathcal{L}')$. 





Algorithm {\sc minres-be-drawer}$(e^*)$ now performs a bottom-up visit of $T$. 

After visiting each C-node $c$, {\sc minres-be-drawer}$(e^*)$ either concludes that $G$ admits no {\sc minres}-supporting  embedding such that $e^*$ is not nested into any edge of $G$, or determines a sequence of {\sc minres}-supporting embeddings $\mathcal{L}^+_1(c), \dots, \mathcal{L}^+_k(c)$ of $G^+(c)$ such~that: 
\begin{enumerate}[(C1)]
    \item\label{gio} for any $i=1,\dots,k$, we have that $c$ is visible in  $\mathcal{L}^+_i(c)$; 
    \item $n_\ell(c,\mathcal{L}^+_1(c))<n_\ell(c,\mathcal{L}^+_2(c))<\dots<n_\ell(c,\mathcal{L}^+_k(c))$; and
    \item for every {\sc minres}-supporting embedding $\mathcal L$ of $G^+(c)$ that respects (C1), there exists an index $i\in \{1,\dots,k\}$ such that $\mathcal{L}^+_i(c)$ is left-right equivalent to $\mathcal{L}$ with respect to $c$.
\end{enumerate}

Note that, by Property (C2), no two {\sc minres}-supporting embeddings among $\mathcal{L}^+_1(c), \dots, \mathcal{L}^+_k(c)$ are left-right equivalent with respect to $c$.

After visiting a B-node $b \neq b^*$, algorithm {\sc minres-be-drawer}$(e^*)$ either concludes that $G$ admits no {\sc minres}-supporting embedding such that $e^*$ is not nested into any edge of $G$, or determines a single {\sc minres}-supporting embedding $\mathcal{L}^+(b)$ of $G^+(b)$ such that:
\begin{enumerate}[(B1)]
    \item the parent $c$ of $b$ in $T$ is the first vertex of $\mathcal{L}^+(b)$; and
    \item $G^+(b)$ admits no {\sc minres}-supporting embedding that respects (B1) and whose residual capacity is larger than the one of $\mathcal{L}^+(b)$.
\end{enumerate}

Restricting the attention to {\sc minres}-supporting embeddings satisfying Condition~(C1) or Condition~(B1) is not a loss of generality, because of the following two lemmata.

\begin{lemma} \label{le:necessity-c1-minres}
	Suppose that $G$ admits a {\sc minres}-supporting embedding $\mathcal L$ such that $e^*$ is not nested into any edge of $G$. Let $c$ be a C-node of $T$ and let $\mathcal L^+(c)$ be the restriction of $\mathcal L$ to the vertices and edges of $G^+(c)$. Then $c$ is visible in $\mathcal L^+(c)$.
\end{lemma}

\begin{lemma} \label{le:necessity-b1-minres}
	Suppose that $G$ admits a {\sc minres}-supporting embedding $\mathcal L$ such that $e^*$ is not nested into any edge of $G$. Let $b\neq b^*$ be a B-node of $T$ and let $\mathcal L^+(b)$ be the restriction of $\mathcal L$ to the vertices and edges of $G^+(b)$. Then the parent $c$ of $b$ in $T$ is either the first or the last vertex of $\mathcal L^+(b)$.
\end{lemma}

The proofs of Lemmata~\ref{le:necessity-c1-minres} and~\ref{le:necessity-b1-minres} follow almost verbatim the proofs of Lemmata~\ref{le:necessity-c1} and~\ref{le:necessity-b1} and are hence omitted. The only difference is that here $e^*$ is not nested into any edge of $G$ by assumption, while in the proofs of Lemmata~\ref{le:necessity-c1} and~\ref{le:necessity-b1} the edge $e_M$ with maximum weight is not nested into any edge of $G$ by the constraints of a {\sc sum}-constrained book-embedding.

Similarly as for {\sc sum}-constrained book-embeddings, we provide a bound on the number of {\sc minres}-supporting embeddings that are pairwise not left-to-right equivalent.

\begin{lemma} \label{le:number-of-orderings-C-minres}
	Let $H=(V_H,E_H,\omega_H)$ be an $n$-vertex weighted outerplanar graph. For a vertex $c$ of $H$, let $\mathcal S$ be a set of {\sc minres}-supporting embeddings of $H$ such that:
	
	\begin{enumerate}[$(\gamma 1)$]
		\item for each $\mathcal L \in \mathcal S$, we have that $c$ is visible in  $\mathcal{L}$; and
		\item for any $\mathcal L,\mathcal L' \in \mathcal S$, we have that $\mathcal L$ is not left-right equivalent to $\mathcal L'$ with respect to $c$.
	\end{enumerate}
	Then  $\mathcal S$ contains at most $n$ embeddings.
\end{lemma}

\begin{proof}
Similarly to the proof of Lemma~\ref{le:number-of-orderings-C},	the statement descends from the following two claims.
	
First, for any value $\lambda\geq 0$, there exists at most one {\sc minres}-supporting embedding $\mathcal L \in \mathcal S$ such that $n_{\ell}(c,\mathcal{L})=\lambda$. Namely, if $\mathcal S$ contains two {\sc minres}-supporting embeddings $\mathcal L,\mathcal L'$ with $n_{\ell}(c,\mathcal{L})=n_{\ell}(c,\mathcal{L}')=\lambda$, we have $n_r(c,\mathcal{L})=n-n_{\ell}(c,\mathcal{L})$ and $n_r(c,\mathcal{L}')=n-n_{\ell}(c,\mathcal{L}')$, hence $n_r(c,\mathcal{L})=n_r(c,\mathcal{L}')$, which implies that $\mathcal L$ and $\mathcal L'$ are left-right equivalent with respect to $c$; this is not possible, by assumption.

Second, the value $n_{\ell}(c,\mathcal{L})$ for an embedding $\mathcal L \in \mathcal S$ is an integer value in $\{0,\dots,n-1\}$ (namely, it is the number of vertices to the left of $c$ in $\mathcal L$).
\end{proof}

Before describing the algorithm {\sc minres-be-drawer}$(e^*)$, we need the following algorithmic lemma.

\begin{lemma} \label{le:compute-minres}
Let $H=(V_H,E_H,\omega_H)$ be an $n$-vertex weighted outerplanar graph and let $\mathcal L$ be a $1$-page book-embedding of $H$. Then it is possible to determine in $O(n)$ time whether $\mathcal L$ is a {\sc minres}-supporting embedding; further, in the positive case, it is possible to determine in $O(n)$ time the residual capacity of $\mathcal L$.
\end{lemma}

\begin{proof}
We first discuss the case in which $H$ is biconnected. We compute, for each edge $e\in E_H$, the burden $\beta_H(e)$ of $e$ in $\mathcal L$; this is done in total $O(n)$ time, as described in the proof of Lemma~\ref{le:linear-minres-outerplanar-biconnected-algo}. Then, in order to determine whether $\mathcal L$ is a {\sc minres}-supporting embedding, it suffices to check whether $\omega_H(e) \geq \beta_H(e) + 1$, for each edge $e \in E_H$; this takes $O(1)$ time per edge, and hence $O(n)$ time in total. If $\mathcal L$ is a {\sc minres}-supporting embedding, the residual capacity of $\mathcal L$ is equal to $\min(\omega_H(e) - (\beta_H(e)+1))$, where the minimum is taken over all the edges $e \in E_H$ incident to the first vertex of $\mathcal L$; again, this takes $O(n)$ time in total. 

If $H$ is not biconnected, we augment it to a weighted biconnected outerplanar graph $H'$ in $O(n)$ time, as follows. First, we initialize $H'$ to $H$. Then we add to $H'$ an edge of weight $1$ between any two consecutive vertices of $\mathcal{L}$, if such an edge is not already in $H$. Finally, we add to $H'$ an edge $\overline{e}$ of weight $n-1$ between the first and the last vertex of $\mathcal{L}$, if such an edge is not already in $H$. Let $H'=(V_{H'},E_{H'},\omega_{H'})$. Since $V_{H'}=V_H$, we can define $\mathcal{L}'=\mathcal{L}$ and obtain that $\mathcal{L}'$ is a $1$-page book-embedding of $H'$. As in the proof of Theorem~\ref{th:2d-book-outerplanar}, we have that $H'$ is outerplanar and biconnected. 

We claim that no edge in $E_{H'}\setminus E_{H}$ has a weight smaller than its burden plus one. Namely, consider any edge $e\neq \overline{e}$ in $E_{H'}\setminus E_{H}$; by construction, $\omega_{H'}(e)=1$, while the burden of $e$ in $\mathcal L'$  is $0$, given that $e$ connects two consecutive vertices of $\mathcal L'$. Further, if $\overline{e} \in E_{H'}\setminus E_{H}$, then $\omega_{H'}(\overline{e})=n-1$, while the burden of $\overline{e}$ in $\mathcal L'$ is $n-2$, given that the end-vertices of $e$ are the first and the last vertex of $\mathcal L'$. 

By the above claim and since the weight and the burden of every edge $e\in E_H$ is the same in $\mathcal L$ as in $\mathcal L'$, it follows that $\mathcal{L}$ is a {\sc minres}-supporting embedding if and only if $\mathcal{L}'$ is a {\sc minres}-supporting embedding. Thus, in order to determine whether $\mathcal{L}$ is a {\sc minres}-supporting embedding, it suffices to test whether $\mathcal{L}'$ is a {\sc minres}-supporting embedding. Since $H'$ is biconnected, this can be done in $O(n)$ time as described above; in particular, such a computation determines the burden $\beta_{H'}(e)$ of every edge $e\in E_{H'}$ in $\mathcal L'$. If the test succeeds, in order to compute the residual capacity of $\mathcal L$, it suffices to compute $\min(\omega_{H'}(e) - (\beta_{H'}(e)+1))$, where the minimum is taken over all the edges $e\in E_H$ (hence, the edges in $E_{H'}\setminus E_{H}$ are excluded from this computation) incident to the first vertex of $\mathcal L'$; again, this takes $O(n)$ time in total.  
\end{proof}

We now describe the bottom-up visit of $T$ performed by the algorithm {\sc minres-be-drawer}$(e^*)$.


{\bf Processing a leaf.} Let $b$ be a leaf of $T$. Since the algorithm {\sc minres-be-drawer}$(e^*)$ did not terminate because of Failure Condition~$2$, we stored one or two  {\sc minres}-supporting embeddings of $G^+(b)=G(b)$ in which the parent $c$ of $b$ is the first vertex. For each of such embeddings, say $\mathcal L$, we compute the residual capacity of $\mathcal L$ in $O(n(b))$ time, by Lemma~\ref{le:compute-minres}. 

We now select as $\mathcal{L}^+(b)=\mathcal{L}(b)$ the {\sc minres}-supporting embedding of $G^+(b)=G(b)$ with the largest residual capacity (between the at most two stored embeddings). Hence, the single {\sc minres}-supporting embedding $\mathcal{L}^+(b)$ of $G^+(b)$ satisfies Properties~(B1) and~(B2) and can be constructed in $O(n(b))$ time.

{\bf Processing a C-node.} We process a C-node $c$ of $T$ as follows. Let $b_1, \dots, b_h$ be the B-nodes that are children of $c$ in $T$. Since the algorithm {\sc minres-be-drawer}$(e^*)$ did not terminate when visiting $b_1, \dots, b_h$, we have, for $i=1,\dots,h$, a {\sc minres}-supporting embedding $\mathcal{L}^+(b_i)$ of $G^+(b_i)$ satisfying Properties~(B1)--(B2); further, we assume to have already computed the residual capacity $r(\mathcal{L}^+(b_i))$. We relabel the B-nodes $b_1, \dots, b_h$ in such a way that $r(\mathcal{L}^+(b_i))+n^+(b_i)\leq r(\mathcal{L}^+(b_{i+1}))+n^+(b_{i+1})$, for $i=1,\dots,h-1$; this takes $O(n \log n)$ time. 

We now process the B-nodes $b_1, \dots, b_h$ in this order. When processing $b_i$, we construct a sequence $\mathcal{S}_i$ of at most $n$ {\sc minres}-supporting embeddings of $G^+(b_1) \cup \dots \cup G^+(b_i)$ satisfying Properties~(C1)--(C3). When constructing an ordering $\mathcal L$ in a sequence $\mathcal{S}_i$, we also compute $n_\ell(c,\mathcal{L})$ and $n_r(c,\mathcal{L})$.  We now describe the processing of the nodes $b_1, \dots, b_h$.

When processing $b_1$, we let $\mathcal{S}_1$ consist of $\mathcal{L}^+(b_1)$ and its flip, in this order.

Suppose that, for some $i\in \{2,\dots,h\}$, the B-node $b_{i-1}$ has been processed and that the sequence $\mathcal{S}_{i-1}$ has been constructed. We process $b_i$ as follows. We initialize $\mathcal{S}_i=\emptyset$. We individually consider each of the at most $n$ embeddings in $\mathcal{S}_{i-1}$, say $\mathcal{L}$. We now consider the embedding $\mathcal{L}^+(b_i)$ and we try to combine it with $\mathcal{L}$. This is done as follows. 

\begin{itemize}
	\item If the residual capacity of $\mathcal{L}^+(b_i)$ is larger than $n_r(c,\mathcal{L})$, then we construct a {\sc minres}-supporting embedding $\mathcal L'$ of $G^+(b_1)\cup\dots\cup G^+(b_i)$ by placing the vertices of $\mathcal{L}^+(b_i) \setminus \{c\}$ to the right of $\mathcal{L}$, in the same relative order as they appear in $\mathcal{L}^+(b_i)$; we insert $\mathcal L'$ into $\mathcal{S}_i$ and note that $n_\ell(c,\mathcal{L}')=n_\ell(c,\mathcal{L})$ and that $n_r(c,\mathcal{L}')=n_r(c,\mathcal{L})+n^+(b_i)-1$. 
	\item Analogously, if the residual capacity of $\mathcal{L}^+(b_i)$ is larger than $n_\ell(c,\mathcal{L})$, then we construct a {\sc minres}-supporting embedding $\mathcal L'$ of $G^+(b_1)\cup\dots\cup G^+(b_i)$ by placing the vertices of $\mathcal{L}^+(b_i) \setminus \{c\}$ to the left of $\mathcal{L}$, in the opposite order as they appear in $\mathcal{L}^+(b_i)$; we insert $\mathcal L'$ into $\mathcal{S}_i$ and note that $n_\ell(c,\mathcal{L}')=n_\ell(c,\mathcal{L})+n^+(b_i)-1$ and that $n_r(c,\mathcal{L}')=n_r(c,\mathcal{L})$.
\end{itemize}

After we considered each of the at most $n$ embeddings in $\mathcal{S}_{i-1}$, if $\mathcal{S}_{i}$ is empty then we conclude that $G$ admits no {\sc minres}-supporting embedding such that $e^*$ is not nested into any edge of $G$ (we call this Failure Condition~$3$). Otherwise, we order and polish the sequence $\mathcal{S}_{i}$ by leaving only one copy of left-right equivalent embeddings. This is done in $O(n \log n)$ time as follows. 


Since $|\mathcal{S}_{i-1}|$ is at most $n$, it follows that the cardinality of $\mathcal{S}_{i}$ before the polishing is at most $2n$. We order $\mathcal{S}_{i}$ in $O(n \log n)$ time by increasing value of the number of vertices to the left of $c$. Then we scan $\mathcal{S}_i$; during the scan, we process the elements of $\mathcal{S}_i$ one by one. When we process an element $\mathcal{L}$, we compare $\mathcal{L}$ with its predecessor $\mathcal{L}'$. If $\mathcal{L}$ and $\mathcal{L}'$ are left-right equivalent with respect to $c$, then we remove $\mathcal L$  from $\mathcal{S}_i$. Note that this scan takes $O(n)$ time. After this scan, we have that no two embeddings in $\mathcal{S}_i$ are left-right equivalent with respect to~$c$. By Lemma~\ref{le:number-of-orderings-C-minres}, there are at most $n$ embeddings in $\mathcal{S}_i$. 

This concludes the description of the processing of the B-node $b_i$ and the subsequent construction of the sequence $\mathcal{S}_i$. After processing the last child $b_h$ of $c$, the sequence $\mathcal{S}_h$ contains the required {\sc minres}-supporting embeddings of $G^+(b_1) \cup \dots \cup G^+(b_h) = G^+(c)$. The proof that such a (possibly empty) sequence $\mathcal{S}_h$ satisfies Properties~(C1)--(C3) is similar to the proof of Lemma~\ref{le:correctness-C} and is hence omitted. We only note here that, in a {\sc minres}-supporting embedding of $G^+(c)$ in which $c$ is visible, if $G^+(b_{i+1})$ lies under an edge of $G^+(b_i)$, then $r(\mathcal{L}^+(b_i))>n^+(b_{i+1})$; however, if that is the case, the inequality $r(\mathcal{L}^+(b_i))+n^+(b_i)\leq r(\mathcal{L}^+(b_{i+1}))+n^+(b_{i+1})$ ensures that $r(\mathcal{L}^+(b_{i+1}))>n^+(b_i)$, and hence that $G^+(b_i)$ can lie under an edge of $G^+(b_{i+1})$ as well. This is the core of the argument for proving that choosing the ordering $b_1,\dots,b_h$ for the B-nodes that are children of $c$ does not introduce any loss of generality.

Since we process each B-node $b_i$ that is child of $c$ in $T$ in $O(n\log n)$ time, the overall time needed to process $c$ is $O(h n \log n)$. This sums up to $O(n^2 \log n)$ time over all the C-nodes of $T$.

{\bf Processing an internal B-node different from the root.} We now describe how to process an internal B-node $b\neq b^*$ of $T$. The goal is to either conclude that $G$ admits no {\sc minres}-supporting embedding such that $e^*$ is not nested into any edge of $G$, or to construct a {\sc minres}-supporting embedding $\mathcal{L}^+(b)$ of $G^+(b)$ satisfying Properties (B1)--(B2). In the latter case, we also compute the residual capacity of $\mathcal{L}^+(b)$.

Since {\sc minres-be-drawer}$(e^*)$ did not terminate because of Failure Condition~$2$, we have either one or two {\sc minres}-supporting embeddings of $G(b)$ in which $c$ is the first vertex. We process each embedding $\mathcal L$ of $G(b)$ at our disposal independently, by means of a procedure which tries to extend $\mathcal L$ to an embedding of $G^+(b)$, as described below. If the procedure fails for every {\sc minres}-supporting embedding of $G(b)$ at our disposal, we conclude that $G$ admits no {\sc minres}-supporting embedding such that $e^*$ is not nested into any edge of $G$ (we call this Failure Condition~$4$). If the procedure succeeds in constructing a {\sc minres}-supporting embedding of $G^+(b)$ satisfying Properties (B1)--(B2) for a single {\sc minres}-supporting embedding of $G(b)$, then we retain the computed embedding of $G^+(b)$. Finally, if the procedure succeeds in constructing a {\sc minres}-supporting embedding of $G^+(b)$ satisfying Properties (B1)--(B2) for two {\sc minres}-supporting embeddings of $G(b)$, then we retain the embedding of $G^+(b)$ with the maximum residual capacity.

Let $\mathcal{L}$ be the current embedding of $G(b)$. Let $c_1, \dots, c_h$ be the C-nodes that are children of $b$, labeled in the same order as they appear in $\mathcal{L}$.  Since the algorithm {\sc minres-be-drawer}$(e^*)$ did not terminate when visiting $c_1, \dots, c_h$, we have, for each $i=1,\dots,h$, a sequence $\mathcal{L}^+_1(c_i), \mathcal{L}^+_2(c_i), \dots, \mathcal{L}^+_{k_i}(c_i)$ of {\sc minres}-supporting embeddings of $G^+(c_i)$ satisfying Properties~(C1)--(C3). Further, for $i=1,\dots,h$ and $j=1,\dots,k_i$, we have already computed the values $n_\ell(c_i,\mathcal{L}^+_j(c_i))$ and $n_r(c_i,\mathcal{L}^+_j(c_i))$.

Our strategy is to process the C-nodes that are children of $b$ in the order $c_h,\dots,c_1$ and, for each $i=h,\dots,1$, to choose a {\sc minres}-supporting embedding $\mathcal{L}^+_{j}(c_i)$ of $G^+(c_i)$ in such a way that $n_{\ell}(c_i,\mathcal{L}^+_{j}(c_i))$ is minimum, subject to the constraint that the embedding resulting from the replacement of $c_i$ with $\mathcal{L}^+_{j}(c_i)$ is a {\sc minres}-supporting embedding. We formalize this idea as follows.

We process the C-nodes $c_h,\dots,c_1$ in this order. Before any C-node is processed, we initialize $\mathcal{L}^*_{h+1}:= \mathcal{L}$. Then, for each $i=h,\dots,1$, when processing $c_i$, we try to construct a {\sc minres}-supporting embedding $\mathcal L^*_i$ of $G(b)\cup G^+(c_h)\cup \dots \cup G^+(c_i)$. This is done by trying to insert the {\sc minres}-supporting embedding $\mathcal{L}^+_j(c_i)$ of $G^+(c_i)$ into $\mathcal L^*_{i+1}$, for $i=1,\dots, k_i$. That is, we replace $c_i$ with $\mathcal{L}^+_{j}(c_i)$ in $\mathcal L^*_{i+1}$, and then check whether the resulting embedding is a {\sc minres}-supporting embedding; this can be done in $O(n)$ time by Lemma~\ref{le:compute-minres}. The first time a check succeeds, we stop the computation and set $\mathcal L^*_{i}$ to be the resulting embedding of $G(b)\cup G^+(c_h)\cup \dots \cup G^+(c_i)$. If no check succeeds, then we let the procedure fail for the current embedding $\mathcal L$ of $G(b)$. When $i=1$, if the procedure did not fail, then we constructed a {\sc minres}-supporting embedding; by means of Lemma~\ref{le:compute-minres}, we compute in $O(n)$ time the residual capacity of this embedding. 

We have the following.

\begin{lemma}\label{le:correctness-B-minres}
If {\sc minres-be-drawer}$(e^*)$ constructs an embedding of $G^+(b)$, then this is a {\sc minres}-supporting embedding satisfying Properties~(B1)--(B2). Further, if {\sc minres-be-drawer}$(e^*)$ concludes that $G$ admits no {\sc minres}-supporting embedding such that $e^*$ is not nested into any edge of $G$, then this conclusion is correct.
\end{lemma}

\begin{proof}
We first discuss the case in which the algorithm {\sc minres-be-drawer}$(e^*)$ constructs an embedding $\mathcal L^+(b)$ of $G^+(b)$. Recall that $\mathcal L^+(b)$ is constructed starting from an embedding $\mathcal{L}^*_{h+1}$ of $G(b)$ by replacing, for $i=h,\dots,1$, the vertex $c_i$ with an embedding $\mathcal{L}^+_j(c_i)$ of $G^+(c_i)$ into $\mathcal L^*_{i+1}$ in order to obtain $\mathcal L^*_i$. Since after each of such replacements a check is performed on whether the resulting embedding is a {\sc minres}-supporting embedding, it follows that $\mathcal L^+(b)=\mathcal{L}^*_{1}$ is indeed a {\sc minres}-supporting embedding. 

Since {\sc minres-be-drawer}$(e^*)$ did not terminate because of Failure Condition~2, it follows that the parent $c$ of $b$ is the first vertex of the embedding $\mathcal{L}^*_{h+1}$ of $G(b)$ in $\mathcal L^+(b)$. Further, for $i=h,\dots,1$, the replacement of $c_i$ with an embedding $\mathcal{L}^+_j(c_i)$ of $G^+(c_i)$ into $\mathcal L^*_{i+1}$ does not change the first vertex of the embedding, given that $c\neq c_i$; it follows that $c$ is the first vertex of $\mathcal L^+(b)$ as well, hence $\mathcal L^+(b)$ satisfies Property (B1).

We now prove that $\mathcal L^+(b)$ satisfies Property~(B2). Suppose, for a contradiction, that there exists a {\sc minres}-supporting embedding $\mathcal L^{\diamond}$ of $G^+(b)$ satisfying Property~(B1) whose residual capacity is larger than the one of $\mathcal L^+(b)$. Let $\mathcal L^{\diamond}_{h+1}$ be the restriction of $\mathcal L^{\diamond}$ to $G(b)$; further, for $i=h,\dots,1$, let $\mathcal L^{\diamond}_i(c_i)$ be the restriction of $\mathcal L^{\diamond}$ to $G^+(c_i)$ and let $\mathcal L^{\diamond}_i$ be the restriction of $\mathcal L^{\diamond}$ to $G(b)\cup G^+(c_h)\cup \cdots \cup G^+(c_i)$; note that $\mathcal L^{\diamond}_1=\mathcal L^{\diamond}$. 

Since $\mathcal L^{\diamond}$ satisfies Property~(B1), we have that $\mathcal L^{\diamond}_{h+1}$ satisfies Property~(B1) as well; that is, $c$ is the first vertex of $\mathcal L^{\diamond}_{h+1}$.  Then $\mathcal L = \mathcal L^{\diamond}_{h+1}$ is one of the (at most two) embeddings of $G(b)$ processed by {\sc minres-be-drawer}$(e^*)$. We show that processing $\mathcal L$ leads to the construction of a {\sc minres}-supporting embedding $\mathcal L^*_1$ of $G^+(b)$ whose residual capacity is larger than or equal to the one of $\mathcal L^{\diamond}$; by construction, the residual capacity of $\mathcal L^+(b)$ is larger than or equal to the one of $\mathcal L^*_1$, which provides the desired contradiction. 

In order to prove that, when processing $\mathcal L$, {\sc minres-be-drawer}$(e^*)$ constructs a {\sc minres}-supporting embedding $\mathcal L^*_1$ of $G^+(b)$ whose residual capacity is larger than or equal to the one of $\mathcal L^{\diamond}$, we actually prove a stronger statement. Let $\mathcal L=(v_0,\dots,v_k)$ and, for $i=1,\dots,h$, let $x(i)$ be such that $v_{x(i)}=c_i$. We prove, by reverse induction on $i$, that, when processing $\mathcal L$, {\sc minres-be-drawer}$(e^*)$ constructs a {\sc minres}-supporting embedding $\mathcal L^*_i$ of $G(b)\cup G^+(c_h)\cup \dots \cup G^+(c_i)$ such that, for any $j\in \{0,\dots,x(i)-1\}$ and for any edge $e$ incident to $v_j$, the burden of $e$ in $\mathcal L^*_i$ is smaller than or equal to the one in $\mathcal L^\diamond_i$. By using the values $i=1$ and $j=0$, this statement implies that the residual capacity of $\mathcal L^*_1$ is indeed larger than or equal to the one of $\mathcal L^{\diamond}$. 

We now prove the induction. In the base case, we have $i=h+1$. Then the statement is clearly satisfied, as $\mathcal L^*_{h+1}$ and $\mathcal L^\diamond_{h+1}$ both coincide with $\mathcal L$. Now suppose that the statement is true for some $i+1\in \{2,\dots,h+1\}$. We prove that the statement is true for $i$, as well. Since $\mathcal L^+_1(c_i),\dots,\mathcal L^+_{k_i}(c_i)$ satisfy Properties (C1)--(C3), there exists a {\sc minres}-supporting embedding $\mathcal L^+_p(c_i)$ that is left-right equivalent to $\mathcal L^{\diamond}_i(c_i)$ with respect to $c_i$; that is, $n_{\ell}(c_i,\mathcal L^+_p(c_i))=n_{\ell}(c_i,\mathcal L^{\diamond}_i(c_i))$ and $n_r(c_i,\mathcal L^+_p(c_i))=n_r(c_i,\mathcal L^{\diamond}_i(c_i))$. This implies that, if $c_i$ is replaced with $\mathcal L^+_p(c_i)$ in $\mathcal L^*_{h+1}$, then the resulting embedding is a {\sc minres}-supporting embedding; in fact, the edges whose burden might change after the replacement are of three types: (i) edges $(v_y,v_{x(i)})$ with $y<x(i)$; (ii) edges $(v_{x(i)},v_z)$ with $x(i)<z$; and (iii) edges $(v_y,v_z)$ with $y<x(i)<z$.
By the inductive hypothesis, the burden of any of such edges in $\mathcal L^*_{i+1}$ is smaller than or equal to the one in $\mathcal L^\diamond_{i+1}$, hence the same is true after the replacement happens in both embeddings (as such burden is decreased by the same quantity, possibly $0$, in both embeddings); then the resulting embedding is a {\sc minres}-supporting embedding given that $\mathcal L^\diamond_i$ is. 
Since {\sc minres-be-drawer}$(e^*)$ replaces $c_i$ with the first embedding among $\mathcal L^+_1(c_i),\dots,\mathcal L^+_{k_i}(c_i)$ such that the resulting embedding is a {\sc minres}-supporting embedding and since as proved above the replacement of $c_i$ with $\mathcal L^+_p(c_i)$ does result in {\sc minres}-supporting embedding, it follows that $\mathcal L^*_i$ is a {\sc minres}-supporting embedding. 

In order to prove the inductive hypothesis, however, we need to address the fact that the embedding of $G^+(c_i)$ that is used in $\mathcal L^*_i$ might not be $\mathcal L^+_p(c_i)$, but rather an embedding $\mathcal L^+_q(c_i)$ with $q<p$; recall that $n_{\ell}(c_i,\mathcal L^+_q(c_i))<n_{\ell}(c_i,\mathcal L^+_p(c_i))$ and $n_r(c_i,\mathcal L^+_q(c_i))>n_r(c_i,\mathcal L^+_p(c_i))$. For an edge $(v_y,v_z)$ with $y<x(i)<z$, using $\mathcal L^+_q(c_i)$ rather than $\mathcal L^+_p(c_i)$ makes no difference, as the burden of such an edge increases by $n^+(c_i)-1$ in any case. The burden of an edge $(v_y,v_{x(i)})$ with $y<x(i)$ after the replacement of $c_i$ with $\mathcal L^+_q(c_i)$ is actually smaller than the burden of $(v_y,v_{x(i)})$ after the replacement of $c_i$ with $\mathcal L^+_p(c_i)$, given that $n_{\ell}(c_i,\mathcal L^+_q(c_i))<n_{\ell}(c_i,\mathcal L^+_p(c_i))$. On the contrary, the burden of an edge $(v_{x(i)},v_z)$ with $x(i)<z$ after the replacement of $c_i$ with $\mathcal L^+_q(c_i)$ is larger than the burden of $(v_{x(i)},v_z)$ after the replacement of $c_i$ with $\mathcal L^+_p(c_i)$, given that $n_r(c_i,\mathcal L^+_q(c_i))>n_r(c_i,\mathcal L^+_p(c_i))$; however, the inductive hypothesis only needs to provide guarantees about the burden of the edges incident to vertices $v_j$ with $j\in \{0,\dots,x(i)-1\}$, and $(v_{x(i)},v_z)$ is not among such edges. This completes the induction and hence the proof that, if {\sc minres-be-drawer}$(e^*)$ constructs an embedding of $G^+(b)$, then this is a {\sc minres}-supporting embedding satisfying Properties~(B1)--(B2). 

We now prove that, if {\sc minres-be-drawer}$(e^*)$ concludes that $G$ admits no {\sc minres}-supporting embedding such that $e^*$ is not nested into any edge of $G$, then this conclusion is correct. During the processing of $b$, it is concluded that $G$ admits no {\sc minres}-supporting embedding such that $e^*$ is not nested into any edge of $G$ only if the algorithm {\sc minres-be-drawer}$(e^*)$ incurs in Failure Condition~4. Assume that {\sc minres-be-drawer}$(e^*)$ incurs in Failure Condition~4. If $G$ admits no {\sc minres}-supporting embedding $\mathcal L_G$ such that $e^*$ is not nested into any edge of $G$, then the conclusion is indeed correct, so assume the contrary. By Lemma~\ref{le:necessity-b1-minres}, the restriction of $\mathcal L_G$ to $G^+(b)$ is a {\sc minres}-supporting embedding $\mathcal L^{\diamond}$ satisfying Property~(B1). The rest of the proof is the same as the proof that $\mathcal L^+(b)$ satisfies Property~(B2). Namely, it is proved by reverse induction that  {\sc minres-be-drawer}$(e^*)$ constructs a {\sc minres}-supporting embedding $\mathcal L^*_i$ of $G(b)\cup G^+(c_h)\cup \dots \cup G^+(c_i)$ such that, for any $j\in \{0,\dots,x(i)-1\}$ and for any edge $e$ incident to $v_j$, the burden of $e$ in $\mathcal L^*_i$ is smaller than or equal to the one in $\mathcal L^\diamond_i$; this implies that {\sc minres-be-drawer}$(e^*)$ constructs a {\sc minres}-supporting embedding $\mathcal L^*_1$ (whose residual capacity is larger than or equal to the one of $\mathcal L^{\diamond}$). The fact that {\sc minres-be-drawer}$(e^*)$ constructs such an embedding implies that it does not incur in Failure Condition~4, a contradiction.
\end{proof}

By Lemma~\ref{le:number-of-orderings-C-minres}, for any C-node $c_i$ that is a child of $b$, the number $k_i$ of {\sc minres}-supporting embeddings $\mathcal{L}^+_1(c_i), \dots, \mathcal{L}^+_{k_i}(c_i)$ of $G^+(c_i)$ is at most $n^+(c_i)$; each of these embeddings is processed in $O(n)$ time, hence the overall time complexity for processing $b$ is in $O((n^+(c_1)+\dots+n^+(c_h)) \cdot n) \in O(n^2)$. This sums up to $O(n^3)$ over all the B-nodes of $T$. 

{\bf Processing the root.} Since algorithm \textsc{minres-be-drawer}$(e^*)$ did not terminate because of Failure Condition~$1$, it constructed a {\sc minres}-supporting embedding $\mathcal{L}(b^*,e^*)$ of $G(b^*)$ in which the end-vertices of $e^*$ are the first and the last vertex. We apply the same algorithm as for a B-node $b\neq b^*$, while using $\mathcal{L}(b^*,e^*)$ in place of the at most two embeddings of $G(b)$. This again requires $O(n^2)$ time. The proof of the following lemma is very similar to the one of Lemma~\ref{le:correctness-B-minres}, and is hence omitted.

\begin{lemma} \label{le:sum-root-minres}
	If $G$ admits a {\sc minres}-supporting embedding such that $e^*$ is not nested into any edge of $G$, then the algorithm {\sc minres-be-drawer} constructs such an embedding, otherwise it concludes that $G$ admits no {\sc minres}-supporting embedding such that $e^*$ is not nested into any edge of $G$.
\end{lemma}

{\bf Running time.} As proved above, the C-nodes of $T$ are processed in overall $O(n^2 \log n)$ time, while the B-nodes of $T$ are processed in overall $O(n^3)$ time. Hence, the running time of the algorithm \textsc{minres-be-drawer}$(e^*)$ is in $O(n^3)$ and the one of the algorithm \textsc{minres-be-drawer} is in $O(n^4)$. This concludes the proof of Theorem~\ref{th:minres-outerplanar}.


\section{Conclusions and Open Problems}\label{se:conclusions}

With the aim of constructing schematic representations of biconnected graphs consisting of a large component plus several smaller components, we studied several types of constrained $1$-page book-embeddings and presented polynomial-time algorithms for testing whether a graph admits such book-embeddings. All the algorithms presented in this paper have been implemented; Figs.~\ref{fig:opening-figure} and~\ref{fig:sum-minres-figure} have been generated by means of such implementations.

Our paper opens several problems. 
\begin{enumerate}
\item Our algorithms allow us to represent only an outerplanar arrangement of small components around a large component. How to generalize the approach to the non-outerplanar case? One could study the problem of minimizing the crossings between components and/or minimizing the violations to the constraints on the weights of the nesting components. 
\item We proposed to linearly arrange the vertices of the separation pairs of the large component on the boundary of a disk. What happens if such an arrangement is instead circular? It is probably feasible to generalize our techniques in this direction, but an extra effort is required.
\item We focused our attention on a ``flat'' decomposition of a graph with just one large component plus many small components. What happens if the small components have their own separation pairs with further levels of decomposition? In other words, how to represent the decomposition of a biconnected graph in all its triconnected components?
\item The algorithms in Section~\ref{se:minres}, which construct two-dimensional book-embeddings with finite resolution, may output drawings whose area is not minimum. Can one minimize the area of such drawings in polynomial time?
\end{enumerate}




\subsubsection*{Acknowledgments} Thanks to an anonymous reviewer for observing that computing a {\sc max}-constrained book-embedding has a time complexity that is lower-bounded by the one of sorting.

\bibliographystyle{splncs04}
\bibliography{bibliography}


%
%

\end{document}